\newif\ifanon
\anonfalse

\documentclass[11pt]{article}

\usepackage[T1]{fontenc}
\usepackage[utf8]{inputenc}
\usepackage[american]{babel}
\usepackage{enumitem}
\usepackage{cite}

\usepackage{color}
\usepackage[section]{placeins} %
\usepackage{framed}

\usepackage[letterpaper,margin=1in]{geometry}
\usepackage{braket}

\usepackage{amthm}

\usepackage{amsfonts}
\usepackage{amsmath}
\usepackage{amssymb}

\usepackage{complexity} %
\usepackage{local-models}
\usepackage[ruled,vlined,linesnumbered]{algorithm2e}

\newlang{\CHSH}{CHSH}
\newlang{\GAMES}{GAMES}
\newlang{\SYMM}{SYMM}
\newlang{\GHZ}{GHZ}

\let\H\relax

\usepackage{ammacros}

\newcommand{\CC}{\mathcal{C}}
\newcommand{\black}{\mathsf{b}}
\newcommand{\iin}{\text{in}} %
\newcommand{\oout}{\mathsf{out}}
\newcommand{\white}{\mathsf{w}}

\usepackage{tikz}
\usetikzlibrary{arrows.meta}
\usetikzlibrary{positioning}
\tikzset{>=stealth}

\usepackage[mode=tex]{standalone}
\usepackage{subcaption}

\usepackage{hyperref}
\usepackage{xcolor}
\definecolor{myblue}{HTML}{0088cc}
\hypersetup{
  colorlinks=true,
  linkcolor=black,
  citecolor=black,
  filecolor=black,
  urlcolor=myblue,
}
\usepackage{doi}
\usepackage{hyphenat} %
\usepackage{csquotes}
\usepackage{microtype}
\usepackage[capitalize,noabbrev]{cleveref}

\usepackage[skins]{tcolorbox}
\usepackage{xcolor}
\newcommand{\nodeconst}{\ensuremath{\mathcal{N}}}
\newcommand{\edgeconst}{\ensuremath{\mathcal{E}}}
\newcommand{\gen}[1]{\langle #1 \rangle}

\newcommand{\nodeconstone}{\nodeconst}
\newcommand{\edgeconstone}{\edgeconst}

\newtcbox{\mybox}[2][]{%
	enhanced,colback=white,colframe=lightgray,coltitle=black,baseline=3mm,
	sharp corners,boxrule=0.4pt,
	size=fbox,
	fonttitle=\itshape,
	attach boxed title to bottom right={yshift=8pt,xshift=2mm},
	boxed title style={tile,size=minimal,left=1mm,right=1mm,
		colback=white,before upper=\strut},
	title=#2,#1
}
\newcommand{\minus}{\scalebox{0.75}[1.0]{$-$}}

\newtheorem{observation}[theorem]{Observation}

\DeclareMathOperator{\re}{\mathcal R}
\DeclareMathOperator{\rere}{\overline{\mathcal R}}

	\newcommand{\reXY}[3]{(#2 #3)_{#1}}
	\newcommand{\reXM}[2]{(#2 \minus)_{#1}}
	\newcommand{\reMY}[2]{(\minus #2)_{#1}}
	\newcommand{\reMM}[1]{(\minus \minus)_{#1}}
	\newcommand{\reEE}[1]{(! !)_{#1}}
	\newcommand{\reQ}[1]{(?)_{#1}}
	\newcommand{\reE}[1]{(!)_{#1}}

\newenvironment{myabstract}%
{\list{}{\listparindent 1.5em
        \itemindent    \listparindent
        \leftmargin    1cm
        \rightmargin   1cm
        \parsep        0pt}%
    \item\relax}%
{\endlist}

\newenvironment{mycover}%
{\list{}{\listparindent 0pt
        \itemindent    \listparindent
        \leftmargin    1cm
        \rightmargin   1cm
        \parsep        0pt}%
    \raggedright
    \item\relax}%
{\endlist}

\newcommand{\myaff}[1]{\,$\cdot$\, {\small #1}\par\medskip}

\begin{document}

\begin{mycover}
{\huge\bfseries Distributed Quantum Advantage \\ for Local Problems \par}
\bigskip
\bigskip

\ifanon
\textbf{Anonymous authors}
\else

\textbf{Alkida Balliu}
\myaff{Gran Sasso Science Institute, Italy}

\textbf{Sebastian Brandt}
\myaff{CISPA Helmholtz Center for Information Security, Germany}

\textbf{Xavier Coiteux-Roy}
\myaff{Technical University of Munich, Germany \,$\cdot$\, Munich Center for Quantum Science and Technology, Germany}

\textbf{Francesco d'Amore}
\myaff{Bocconi University, Italy \,$\cdot$\, BIDSA, Italy}

\textbf{Massimo Equi}
\myaff{Aalto University, Finland}

\textbf{Fran{\c c}ois Le Gall}
\myaff{Nagoya University, Japan}

\textbf{Henrik Lievonen}
\myaff{Aalto University, Finland}

\textbf{Augusto Modanese}
\myaff{Aalto University, Finland}

\textbf{Dennis Olivetti}
\myaff{Gran Sasso Science Institute, Italy}

\textbf{Marc-Olivier Renou}
\myaff{Inria Paris-Saclay, France \,$\cdot$\, CPHT, CNRS, Ecole Polytechnique, Institut Polytechnique de Paris, France}

\textbf{Jukka Suomela}
\myaff{Aalto University, Finland}

\textbf{Lucas Tendick}
\myaff{Inria Paris-Saclay, France}

\textbf{Isadora Veeren}
\myaff{Inria Paris-Saclay, France \,$\cdot$\, CPHT, CNRS, Ecole Polytechnique, Institut Polytechnique de Paris, France}
\fi %

\bigskip
\end{mycover}

\begin{myabstract}
\noindent\textbf{Abstract.}
We present the first \emph{local} problem that shows a super-constant separation between the classical randomized LOCAL model of distributed computing and its quantum counterpart. By prior work, such a separation was known only for an artificial graph problem with an inherently \emph{global} definition [Le Gall et al.\ 2019].

We present a problem that we call \emph{iterated GHZ}, which is defined using only local constraints. Formally, it is a family of \emph{locally checkable labeling} problems [Naor and Stockmeyer 1995]; in particular, solutions can be verified with a constant-round distributed algorithm.

We show that in graphs of maximum degree $\Delta$, any classical (deterministic or randomized) LOCAL model algorithm will require $\Omega(\Delta)$ rounds to solve the iterated GHZ problem, while the problem can be solved in $1$ round in quantum-LOCAL.

We use the \emph{round elimination} technique to prove that the iterated GHZ problem requires $\Omega(\Delta)$ rounds for classical algorithms. This is the first work that shows that round elimination is indeed able to separate the two models, and this also demonstrates that round elimination \emph{cannot} be used to prove lower bounds for quantum-LOCAL. To apply round elimination, we introduce a new technique that allows us to discover appropriate \emph{problem relaxations} in a mechanical way; it turns out that this new technique extends beyond the scope of the iterated GHZ problem and can be used to e.g.\ reproduce prior results on maximal matchings [FOCS 2019, PODC 2020] in a systematic manner.
\end{myabstract}

\thispagestyle{empty}
\setcounter{page}{0}
\clearpage

\section{Introduction}

The key question we study in this work is this: \textbf{which problems admit a distributed quantum advantage}? That is, if we have a computer network, and we replace classical computers with quantum computers and classical communication channels with quantum communication channels, which distributed tasks can be now solved faster?

\subsection{Setting: LOCAL vs.\ quantum-LOCAL}

There is a range of prior work that has studied such questions in \emph{bandwidth-limited} networks, e.g., comparing the relative power of the classical CONGEST and CONGESTED CLIQUE models with their quantum counterparts \cite{ApeldoornV22,Censor-HillelFG22,IzumiG19,IzumiGM20,LeGallM18,MagniezN22,Wang2022,WuY22,FraigniaudLMT24}. In essence, the question is: Can a quantum network which can send $b$ qubits (quantum bits) per time unit over each communication link outperform a classical network which can send the same number of bits per time unit? It turns out that in many cases we can indeed demonstrate quantum advantage in this sense, see, e.g., \cite{LeGallM18,Censor-HillelFG22,IzumiG19,IzumiGM20}.

However, quantum advantage in settings in which \emph{large distances} are the key limitation (instead of bandwidth) are much less understood. We emphasize that large distances (i.e., network latency, or time for information to propagate from one node to another) are a fundamental physical limitation, and not merely an engineering challenge. One can at least in principle increase bandwidth by installing multiple parallel communication channels (say, a larger number of fiber-optic links between a pair of nodes), and it might be significantly harder for quantum communication channels, while there is no way to increase the speed of information propagation beyond the speed of light in a spatially distributed system. Can quantum computation and communication still help? This is our focus in this work.

We study the classical \textbf{LOCAL model} of distributed computing, and compare it with its quantum counterpart. In the LOCAL model computation proceeds in synchronous rounds, and in each round each computer can send an \emph{arbitrarily large message} to each neighbor, receive a message from each neighbor, and do \emph{unbounded local computation}. Initially all nodes are only aware of their own local input, and eventually all nodes must stop and announce their own local output (e.g.\ in the graph coloring problem: what is its own color). We are interested in \textbf{how many communication rounds} are needed (in the worst case) until all nodes stop.

There is extensive line of research since the early 1990s on understanding the classical LOCAL model and the round complexity of various graph problems in this model, see, e.g., \cite{chang16exponential,ChangP19,Balliu2019,Balliu0KO23,Brandt2019,hideandseek,Dahal0LPS23,balliurules,ghaffari16improved,GhaffariKuhn20,GhaffariHKM18,Linial1992,NaorS95,binary,Brandt2016,BrandtHKLOPRSU17,Balliu0COSS22,AkbariELMSS23}. The main focus has been on problems defined with local constraints (e.g.\ graph coloring), and especially in the recent years one of the main research themes has been understanding the distributed complexity landscape of all such problems.

But do any of these problems strictly benefit from quantum computation and communication? In the quantum-LOCAL model, the local state of a node consists of an arbitrary number of qubits, each message consists of an arbitrary number of qubits, and after each communication round, each node can do arbitrary quantum operations and measurements with the qubits that it holds. Does this help in comparison with the classical randomized LOCAL model, where the nodes have only access to a classical source of randomness?

\subsection{Prior work vs.\ our main contribution}

By prior work \cite{LeGallNR19}, it is known that there exists a (very artificial) graph problem with the following properties: in $n$-node graphs, any classical (deterministic or randomized) LOCAL algorithm will require $\Omega(n)$ rounds to solve the problem, while there is an $O(1)$-round quantum-LOCAL algorithm for solving it.

However, the problem from prior work is not only artificial, but it is also very different from the problems usually studied in the context of the LOCAL model. In particular, it has got an inherently \emph{global} definition: to decide if a given output is valid, one has to inspect the joint outputs of nodes in distant parts of the networks. Such a problem is very different from problems typically studied in the context of the LOCAL model. Usually we are interested in problems similar to graph coloring, which can be defined with \emph{local} constraints (e.g., a vertex coloring is valid if and only if each node has an output label different from the labels of its neighbors). Put otherwise, these are problems where the solution may be hard to find for a distributed algorithm, but it is at least easy to verify (cf.\ class FNP in the centralized setting).

The idea of problems defined with local constraints was formalized in the seminal paper by Naor and Stockmeyer in 1995 \cite{NaorS95}. They defined the family of \emph{locally checkable labeling} problems, or in brief LCLs. These are problems that can be defined by specifying a finite set of valid labeled neighborhoods. For example, proper vertex coloring with 11 colors in graphs of degree at most 10 is a concrete example of a graph problem that is formally an LCL in the strict sense (very often we study \emph{families} of LCLs parameterized by the maximum degree, and e.g.\ $(\Delta+1)$-coloring in graphs of degree at most $\Delta$ is a natural example of such a problem family). Most of the problems of interest in the context of the LOCAL model are (families of) LCL problems and there is a huge body of work that has established all kinds of structural results related to the complexity of LCL problems (see, e.g., \cite{chang16exponential, ChangP19, BrandtHKLOPRSU17, Balliu0COSS22, binary, BalliuHKLOS18, lcl-decidability-paths, how-much-randomness-helps, BalliuCMOS21, Balliu0OSST21, Balliu0FLMOU22, Balliu0KOS23, Balliu0KOS24, brandt21trees, BalliuHOS19, BBOS18almostGlobal, chang20, lcls_on_paths_and_cycles, Balliu0OSST21}). Now the key question is:
\begin{framed}
  \noindent
  Does quantum-LOCAL help in comparison with classical LOCAL for some LCL problem?
\end{framed}

Before this work, the answer was: we do not know, beyond constant advantage. There are several \emph{negative} results that exclude any quantum advantage for a specific problem, or at least put limits on it (see, e.g., \cite{GavoilleKM09,arfaoui2014,Coiteux-RoyDGKG24}). It is also easy to engineer a problem that exhibits a \emph{constant} quantum advantage, e.g., a problem that is solvable in $1$ round with quantum-LOCAL and requires $2$ rounds with classical LOCAL. One can also amplify this if we cheat by defining a family of LCL in which the checking radius (i.e., how local the constraints are, or what is the maximum radius of a valid labeled neighborhood in the definition of the LCL) increases as a function of $\Delta$. But beyond this, \emph{nothing positive was known}; the global graph problem from \cite{LeGallNR19} was the only known example of any asymptotic quantum advantage in the LOCAL model. Our main result is this:
\begin{framed}
\noindent
We present a family of LCL problems $\mathcal{P}_\Delta$ called \textbf{iterated GHZ}, parameterized by the maximum degree $\Delta$, with these three properties:
\begin{enumerate}[noitemsep]
  \item the checking radius of $\mathcal{P}_\Delta$ is $O(1)$, independently of $\Delta$,
  \item in quantum-LOCAL, we can solve $\mathcal{P}_\Delta$ in $O(1)$ rounds, independently of $\Delta$,
  \item in classical (deterministic or randomized) LOCAL, any algorithm that solves $\mathcal{P}_\Delta$ requires $\Omega(\Delta)$ rounds.
\end{enumerate}
\end{framed}
\noindent
This gives the \textbf{first local problem} (in any reasonable sense) that demonstrates \textbf{distributed quantum advantage} in the LOCAL model, also answering the open questions from \cite{akbari2024,Coiteux-RoyDGKG24} positively. We emphasize that the problem is engineered so that the feasibility of solutions is \emph{easy to verify} also for classical distributed algorithms, but it is easy to compute only for distributed quantum algorithms.

\subsection{Key techniques and new ideas}
\label{ssec:intro-iterated-def}

As a gentle introduction to our work, let us first introduce a toy version of our problem, called \textbf{iterated CHSH}. While this problem does not quite give us distributed quantum advantage, it already contains all the key ingredients; and the final step that we need to take from iterated CHSH to iterated GHZ is relatively straightforward (albeit technical).

\paragraph{CHSH game.}

The CHSH game~\cite{CHSH, Brunner_Review} is the following two-player game: Alice gets an input bit $x$ and Bob gets an input bit $y$. Alice has to produce an output bit $a$ and Bob has to produce an output bit $b$. The bits have to satisfy $xy = a \oplus b$, that is, if both inputs are $1$, the players must produce different outputs, while otherwise the players must produce the same output.
This game has two important properties:
\begin{enumerate}
	\item It is ``classically hard'' in the following sense: Alice and Bob cannot win this game without any communication (for any classical strategy, Alice and Bob will fail on at least one combination $(x,y)$ of inputs, and thus win with probability at most $75\%$ on uniformly random inputs).
  \item It is ``non-signaling'' in the following sense: if a friendly oracle picks a uniformly random valid output, and informs Alice about her output $a$ and informs Bob about his output $b$, then Alice will learn absolutely nothing about Bob's input and vice versa.
\end{enumerate}
The non-signaling property implies that there \emph{might} be a quantum strategy that enables Alice and Bob to win this game (with a nontrivial probability); at least the existence of such a strategy would not violate causality. And indeed this is the case: if Alice and Bob share a single entangled qubit pair, they can win the game with up to $(1+\sqrt{2})/2\approx 85\%$ on uniformly random inputs, which is better than the classical strategy \cite{Tsirelson_bound}.

Now 85\% probability is not going to suffice for us, and this is exactly why this is just a toy version of the real problem. But let us put this aside for now. It may be useful to imagine for a while that we are living in a \emph{super-quantum} world where the laws of physics actually allowed Alice and Bob to win the CHSH game with probability 1 if they hold some magic pre-shared entangled super-qubits. We will eventually get back from the super-quantum world to the usual quantum world. (Slightly more formally, an entangled pair super-qubit is going to be exactly what is commonly called in the literature a \emph{PR-box} first introduced by Popescu and Rohrlich~\cite{PRBox}, but we do not assume familiarity with this concept.)

\paragraph{Iterated CHSH.}

Now we are ready to explain what the \emph{iterated CHSH} problem looks like. Fix some maximum degree $\Delta$. We can without loss of generality assume that our input graph is $\Delta$-regular and it is $\Delta$-edge-colored, that is, for each $c = 1, 2, \dotsc, \Delta$, each node is incident to exactly one edge that is labeled with color $c$. (We note that we do not really need a promise here: we can simply define the problem so that all nodes that violate this property are unconstrained, and hence hard instances are exactly those graphs in which all nodes satisfy this locally checkable property, and we do not have any unconstrained nodes.)

We define a graph problem in which each node $v$ has to output a sequence of bits $v_0, v_1, \dotsc, v_\Delta$ that satisfies these constraints:
\begin{itemize}[noitemsep]
  \item For each node $v$ we have $v_0 = 1$.
  \item If the color of the $e = \{u,v\}$ is $c$, then $u_{c-1} v_{c-1} = u_c \oplus v_c$.
\end{itemize}
This can be interpreted as follows:
\begin{itemize}[noitemsep]
  \item Each edge $\{u,v\}$ of color $1$ represents an instance of the CHSH game where the inputs are hardcoded to $u_0 = 1$ and $v_0 = 1$, and the outputs are $u_1$ and $v_1$.
  \item Each edge $\{u,v\}$ of color $c > 1$ represents an instance of the CHSH game where the inputs are $u_{c-1}$ and $v_{c-1}$ and the outputs are $u_c$ and $v_c$.
\end{itemize}
That is, the outputs of the CHSH games that we play on the edges of color $1$ form the inputs of the CHSH games that we play on the edges of color $2$, and the outputs of those games form the inputs of the CHSH games that we play on the edges of color $3$, etc., all the way up to the games that we play on the edges of color $\Delta$.
We will now study some properties of the iterated CHSH problem.

\paragraph{Iterated CHSH is a local problem.}

First, we observe that the iterated CHSH problem is a locally verifiable problem. Indeed, each edge represents the constraints of the CHSH game, and if we are given some solution, it is sufficient that we check that for each edge the relevant bits satisfy these constraints; hence the validity of a solution is locally verifiable.

Furthermore, for each fixed maximum degree $\Delta$ we satisfy the usual technical requirements of the definition of LCL problems in \cite{NaorS95}, namely, the set of input labels ($\Delta$ edge colors) and the set of output labels ($2^{\Delta}$ bit strings) are finite.
In summary, we have a family of LCL problems parameterized by $\Delta$, and the checking radius does not depend on $\Delta$.

\paragraph{Iterated CHSH in super-quantum-LOCAL.}

Now if we were in the hypothetical super-quantum world in which pre-shared super-qubits enabled Alice and Bob to win the CHSH game with probability 1, we could also solve the iterated CHSH problem in only \textbf{one round of super-quantum communication}, as follows.

First, each node locally prepares $\Delta$ pairs of entangled super-qubits. Then we use one round of communication so that each node exchanges one entangled super-qubit with each of its neighbors (it is critical here that we can do this in parallel with each neighbor). Additionally, each node can share e.g.\ its own unique identifier that we can later use for symmetry-breaking purposes. This way for each edge we have prepared two pairs of entangled super-qubits, and with the help of unique identifiers we can consistently agree on which of these pairs to use. In summary, for each node $v$ there are super-qubits $v^1, v^2, \dotsc, v^\Delta$ so that if $\{u,v\}$ is an edge of color $v$, super-qubits $u^c$ and $v^c$ form a pair of entangled super-qubits powerful enough to win the CHSH game. 

Then all the rest is local computation. Consider a node $v$. It initializes $v_0 = 1$, and then iteratively repeats the following step for each $c = 1, 2, \dotsc, \Delta$: Let $e = \{u,v\}$ be the edge of color $c$ incident to $v$; in the first round $v$ already learned the unique identifier of $u$. We check the unique identifiers to decide which of the nodes takes the role of Alice and which takes the role of Bob; let us say that $v$ becomes Alice. Now $v$ knows that it shares an entangled pair of super-qubits $(v^c, u^c)$ with node $u$, who plays the role of Bob. Alice knows its own input $v_{c-1}$, and Bob knows its own input $u_{c-1}$. Node $v$, Alice, uses its own super-qubit $v^c$ appropriately to figure out its output $v_c$, and without any communication node $u$ will do the same using its own entangled super-qubit $u^c$ to find~$u_c$. 

This way just 1 round of communication, followed by $\Delta$ steps of purely local super-quantum operations is enough for all nodes to produce valid solutions.
Hence we have constructed a family of LCL problems that is easy in super-quantum-LOCAL. This is not particularly surprising, as the whole intuition here is that the problem is \emph{by construction} easy in the super-quantum world. All the heavy lifting goes to showing that it is hard in the classical world.

\paragraph{Iterated CHSH in classical LOCAL.}

To gain more intuition on the iterated CHSH problem, let us first make the following observation: in graphs of maximum degree $\Delta$, we can solve this problem in $\Theta(\Delta)$ rounds in the classical deterministic LOCAL model.

The algorithm is very simple: Each node $v$ first initializes $v_0 = 1$. Then for each $c = 1, 2, \dotsc, \Delta$, we consider all edges $e = \{u,v\}$ of color $c$. For each such edge, $v$ will send its own identifier and the value of $v_{c-1}$ to $u$, and conversely $u$ will send its own identifier and $u_{c-1}$ to $v$. This way $u$ and $v$ both know the values of $u_{c-1}$ and $v_{c-1}$, plus their own identifiers. Now assume that the identifier of $u$ is smaller than $v$; then the nodes will find the lexicographically smallest pair $(u_c, v_c)$ such that this is a valid output for the CHSH game assuming that node $u$ is Alice with input $u_{c-1}$ and node $v$ is Bob with input $v_{c-1}$. In particular, node $v$ learns the value of $v_c$ and node $u$ learns the value of $u_c$ so that these satisfy the constraints for edge $\{u,v\}$.

In $\Delta$ rounds this will result in a valid solution. Essentially, each node plays CHSH games in a sequential order, starting with the game on the edge of color $1$, and eventually reaching the game on the edge of color $\Delta$.

Incidentally, it is good to note that both this $\Theta(\Delta)$-round classical algorithm and the $O(1)$-round super-quantum algorithm send $\Theta(\Delta)$ messages---each node communicates with each neighbor exactly once. The key difference is that in the super-quantum algorithm, the messages that we send (entangled super-qubits) do not depend on each other, and we can send all of them in the first round, simultaneously in parallel, while in the classical algorithm, the messages that we send in round $c$ depend on the messages sent in round $c-1$, and hence they cannot be parallelized.

Now what we would like to do is to show that \textbf{the sequential strategy is the best possible classical algorithm}, i.e., there is no $o(\Delta)$-round classical algorithm that in some clever way circumvents the need for coordinating with each neighbor in a sequential manner. But let us first see how we get from the imaginary super-quantum world to the familiar quantum world, and then discuss how to show the analogous statement for the real problem.

\paragraph{From super-quantum to quantum: replacing CHSH with GHZ.}

So far we have discussed the toy problem, iterated CHSH, that (at best) can prove a separation between the classical world and an imaginary super-quantum world in which the CHSH game can be solved with probability~1. Let us now see how to apply the same idea to separate the classical world and the real quantum world. To do that, we need a suitable non-signaling game with all these properties:
\begin{itemize}[noitemsep]
  \item It is impossible to solve without any communication using classical randomness.
  \item It can be solved with probability 1 with the help of pre-shared entangled qubits.
  \item The input and output values come from some finite set, so that we can turn this into an LCL problem.
  \item The set of possible input values has to be the set of possible output values, so that we can use the output of one game as the input of another game.
  \item The game does not require any promises on allowed input combinations, so that the game is always well-defined.
\end{itemize}
The first two conditions define games called quantum pseudo-telepathy~\cite{brassard2005quantum}. One good candidate to fulfill also the last three conditions is the well-known GHZ game~\cite{GHZineq,mermin1990quantum}. This is a three-player game, where Alice, Bob, and Charlie each get one input bit, $x$, $y$, and $z$, and each of them has to produce one output bit, $a$, $b$, and $c$. The usual phrasing of the game is as follows: we promise that $x \oplus y \oplus z = 0$, and the requirement is that $a \oplus b \oplus c = 0$ if and only if $x = y = z = 0$. This phrasing of the game comes with a promise, but we can eliminate it by specifying that any output is valid if $x \oplus y \oplus z = 1$.

The good news is that this game is indeed solvable with probability 1 given a suitable entangled quantum state~\cite{GHZineq}. The bad news is that it is very relaxed: it is not clear if this game can be used to construct any problem that is classically hard.
We now follow the blueprint of the iterated CHSH problem; in essence we replace each instance of the CHSH game with an instance of the GHZ game. We obtain a local problem that we call \emph{iterated GHZ}. We prove that this problem is classically hard.

In the CHSH problem, each edge represented a $2$-player game, and we need to embed $3$-player games. In essence, we need to switch from graphs to $3$-uniform hypergraphs (where each hyperedge represents a $3$-player game), and then we will use a bipartite graph to represent the hypergraph: we will have \emph{white nodes} of degree $\Delta$ that represent players and \emph{black nodes} of degree $3$ that represent games. Black nodes are \emph{colored} so that for each $c = 1, 2, \dotsc, \Delta$ and for each white node there is exactly one black neighbor of color $c$.

Again the intuition is that we play GHZ games in a sequential order, so that for each white node $v$, the output of the GHZ game represented by the black neighbor of color $c-1$ will be the input of the GHZ game represented by the black neighbor of color $c$.

This time extra care is needed with bootstrapping. In the CHSH setting, we simply assumed that the first game has all inputs hardcoded to $1$. If we did the same here, the iterated GHZ problem would become trivial to solve (note that all-$1$ input implies that all-$1$ output is valid, which means that all nodes could blindly output a vector of $1$ bits). Similarly, all-$0$ input leads to a trivial problem. To force a nontrivial solution, we make the first game special: black nodes of color $1$ represent the task in which we have no inputs, and exactly one of the players Alice, Bob, and Charlie has to output $1$ and the other two must output $0$. (We note that this is in some sense analogous to a CHSH game in which the inputs are hardcoded to $1$, as there exactly one of Alice and Bob has to output $1$ and the other one has to output $0$.)

\paragraph{Iterated GHZ.}

Let us now summarize the definition of the iterated GHZ problem. We have a bipartite graph in which white nodes have degree $\Delta$, and black nodes have degree $3$. For each $c = 1, 2, \dotsc, \Delta$ and for each white node there is exactly one black neighbor of color $c$.

Each white node $v$ has to output a bit vector $v_1, v_2, \dotsc, v_\Delta$, so that the following holds. Assume that $s$, $t$, and $u$ are the three white neighbors of a black node $x$, and let $c$ be the color of $x$. We require the following:
\begin{itemize}
  \item If $c = 1$, then exactly one of $s_1$, $t_1$, and $u_1$ has to be $1$.
  \item If $c > 1$, then $s_c$, $t_c$, and $u_c$ have to be a valid output for the GHZ game with inputs $s_{c-1}$, $t_{c-1}$, and $u_{c-1}$:
  \begin{alignat*}{2}
      s_c \oplus t_c \oplus u_c &= 0 &\text{ if } s_{c-1} + t_{c-1} + u_{c-1} &= 0, \\
      s_c \oplus t_c \oplus u_c &= 1 &\text{ if } s_{c-1} + t_{c-1} + u_{c-1} &= 2, \\
      s_c \oplus t_c \oplus u_c &\in \{0, 1\} &\text{ if } s_{c-1} + t_{c-1} + u_{c-1} &\in \{1,3\}.
  \end{alignat*}
\end{itemize}

This problem can be solved with a one-round quantum strategy, as follows. Black nodes of color $c = 1$ inform their first white neighbor $v$ that it should set $v_1 = 1$, and inform other neighbors to set $v_1 = 0$. Black nodes of color $c > 1$ prepare a suitable three-qubit GHZ state that enables Alice, Bob, and Charlie to win the GHZ game with probability $1$. Additionally, black nodes assign the roles (Alice, Bob, Charlie) to their white neighbors, and also send their own color $c$, so that white nodes are aware of which games they are supposed to play, in which roles and in which order. After this one communication step, everything is ready, each white node $v$ already knows $v_1$, it holds $\Delta-1$ suitably entangled qubits, and it can locally play its own part in $\Delta-1$ GHZ games and this way discover valid outputs $v_2, v_3, \dotsc, v_\Delta$.

This problem also admits a $\Theta(\Delta)$-round classical solution. Black nodes of color $1$ start and choose which of their neighbors will set $v_1 = 1$; all other white nodes merely indicate their colors. Then white nodes send $v_1$ to black neighbors of color $2$, black nodes of color $2$ can find a valid solution to the GHZ game that they represent, and they reply with $v_2$, etc. After $\Delta$ iterations (and approximately $2\Delta$ communication rounds) all white nodes know their final output.

However, as we discussed earlier, the GHZ game itself is very relaxed, and the details of the definition of the iterated GHZ game are subtle; minor changes in the way we define the first game lead to trivial problems. It is not at all clear intuitively if this problem is hard for classical algorithms, or if there is some clever way to cheat. This is what we show in this work: no classical $o(\Delta)$-round algorithm exist for this problem.

\paragraph{Proving the classical lower bound.}

Now we are given a specific problem with a fairly complicated specification, namely iterated GHZ. There is a simple $\Theta(\Delta)$-round classical algorithm, and we would like to show that this is also (asymptotically) optimal.

Here it is also critical to note that we are heavily limited in the availability of possible proof techniques that one could apply here. In distributed graph algorithms, one commonly uses, e.g., propagation arguments (fixing the solution in one place implies something in a distant part of the graph) or arguments based on indistinguishability (two neighborhoods look locally identical, so any distributed algorithm has to produce the same outputs). However, many of these techniques are also (directly or indirectly) applicable in the quantum-LOCAL model \cite{GavoilleKM09,Coiteux-RoyDGKG24}. If the proof technique would also imply that the problem is hard not only in classical LOCAL but also in quantum-LOCAL, it cannot work here, as we already know that the problem is easy in quantum-LOCAL.

There are very few lower-bound proof techniques that even in principle \emph{could} separate classical LOCAL and quantum-LOCAL. One rare example is the \emph{round elimination} technique \cite{Brandt2019}. It is promising in the informal sense that at least known proofs that use round elimination do not generalize to quantum-LOCAL, as a direct generalization would violate the no-cloning theorem.

In this work we show that \textbf{round elimination is indeed able to separate classical LOCAL and quantum-LOCAL}: we are able to use round elimination to show that iterated GHZ cannot be solved in $o(\Delta)$ rounds in classical (deterministic or randomized) LOCAL, while the problem can be solved in $O(1)$ rounds in quantum-LOCAL. This is both good news and bad news: the good news is that we indeed have a proof technique that can demonstrate distributed quantum advantage, but the bad news is that none of previously-known lower bounds that were proved with round elimination can be directly generalized to quantum-LOCAL.

\paragraph{Round elimination: basic idea.}

Round elimination is a proof technique that is specifically applicable to the study of LCL problems in the classical LOCAL model. It is commonly used to prove lower bounds (see, e.g., \cite{BalliuBHORS21, hideandseek, BBKOmis, Balliu0KO23, BalliuBBO24, balliu2024towards, binary, balliurules, Brandt2016, so-made-simple, Linial1992}), but it can also be used to e.g.\ synthesize efficient algorithms (see, e.g., \cite{balliurules}). It is closely related to the fanout inflation technique~\cite{Wolfe_Inflation} developed to study network nonlocality~\cite{NetworkNonlocality_Review}.

In essence, round elimination is a function $\re$ that answers the following question:
\begin{quote}
Assume that an LCL problem $\Pi$ admits a (black box) classical LOCAL algorithm $A$ working in $T$ rounds. What is the most general problem $\Pi' = \re(\Pi)$ which can be solved in $T-1$ rounds using $A$?
\end{quote}
Round elimination precisely constructs that new problem $\Pi'$. Informally, we construct a $(T-1)$-round algorithm $A'$ that works as follows: each node outputs the set of all values that $A$ might output, given what we see in the radius-$(T-1)$ neighborhood (by considering all possible extensions that we might see at distance exactly $T$ and simulating $A$ for each of them). Now function $\re$ is defined so that assuming that $A$ indeed solved $\Pi$, algorithm $A'$ will solve $\re(\Pi)$, and given any solution to $\re(\Pi)$, we can solve the original problem $\Pi$ with only one extra round.

We note that the definition of $A'$ crucially needs to \emph{clone} the information that we have in the $(T-1)$-radius neighborhood. Hence it is natural to expect that round elimination does not extend to quantum algorithms. Yet, other properties of quantum information could (at least in principle) imply that if a quantum LOCAL algorithm can solve $\Pi$ in $T$ rounds, then there is a quantum algorithm that can solve $\Pi' = \re(\Pi)$ in $T-1$ rounds. Our main result implies that this cannot be the case.

\paragraph{Round elimination: formalism.}

To make use of round elimination, we need to have an LCL problem that is written in a specific formalism:
\begin{itemize}[noitemsep]
  \item We have a bipartite graph, where the two sets of nodes are called \emph{active} and \emph{passive}.
  \item Each active node has to label its incident edges with labels from some finite alphabet $\Sigma$.
  \item We have an \emph{active constraint}, which is a collection of multisets, where each multiset is one possible valid assignment of edge labels incident to an active node.
  \item Similarly, we have a \emph{passive constraint}, which specifies the collection of multisets that are valid on the edges incident to passive nodes.
\end{itemize}
Now if we have any problem $\Pi$ written in this formalism, we can apply in a mechanical way round elimination to obtain another problem $\Pi' = \re(\Pi)$. We emphasize that this part is just syntactic manipulation of the problem description, without any references to any specific model of computation or any distributed algorithm. Yet as explained above, if the round complexity of $\Pi$ is $T$ rounds in the classical LOCAL model (for any sufficiently small $T$), then the round complexity of $\Pi'$ will be exactly $T-1$ rounds.

\paragraph{Round elimination: basic strategy.}

To apply round elimination in our task, we could try to construct a sequence of problems $\Pi_0, \Pi_1, \dotsc, \Pi_T$ such that the following holds:
\begin{itemize}[noitemsep]
  \item $\Pi_0$ is the original problem of interest, iterated GHZ,
  \item $\Pi_{i+1} = \re(\Pi_{i})$,
  \item we can show that $\Pi_T$ is not trivial, i.e., it cannot be solved in $0$ rounds.
\end{itemize}
Now if we could solve $\Pi_0$ in $T$ rounds, it would imply that $\Pi_T$ can be solved in $0$ rounds, but as this is not the case, $\Pi_0$ has to require at least $T+1$ rounds. If we can do this for, e.g., $T = \Delta$, we would have the result that we want: iterated GHZ cannot be solved in $o(\Delta)$ rounds.

However, what typically happens is that $\re(\Pi)$ is \emph{much} more complicated than $\Pi$, in the sense that the size of the alphabet of labels used to describe $\re(\Pi)$ may increase \emph{exponentially} compared to the size of the alphabet used to describe $\Pi$, and the number of valid multisets in the active and passive constraints may increase accordingly. Even if the original problem $\Pi_0$ has got a simple description, there is typically no simple way to describe $\Pi_i$, other than saying that it can be obtained from $\Pi_0$ with $i$ applications of $\re$. In particular, if we do not have a ``closed-form'' description of $\Pi_T$, how do we argue that it is not trivial?
Hence, what typically happens is that we augment this mechanical procedure with some well-chosen \emph{relaxations}.

\paragraph{Round elimination: relaxations.}

We say that $\Pi'$ is a relaxation of $\Pi$ if given a feasible solution to $\Pi$ we can construct (in a distributed way, without any communication) a feasible solution to $\Pi'$.
A simple example of a relaxation is what can be obtained by \emph{merging labels}. Let $a, b \in \Sigma$ be two distinct labels that we can use in $\Pi$. To define $\Pi'$, we relax all active and passive constraints so that whenever $a$ can be used, we can also use $b$, and vice versa; now trivially $\Pi'$ is at least as easy as $\Pi$. Then we can simplify the description of $\Pi'$ by completely removing label $b$, as it is redundant.

The revised process now looks as follows. We construct a sequence of problems $\Pi_0, \Pi_1, \dotsc, \Pi_T$ such that the following holds:
\begin{itemize}[noitemsep]
  \item $\Pi_0$ is some relaxation of the original problem of interest, iterated GHZ,
  \item $\Pi_{i+1}$ is some relaxation of $\re(\Pi_{i})$,
  \item we can show that $\Pi_T$ is not trivial, i.e., it cannot be solved in $0$ rounds.
\end{itemize}
Now if we could solve the original problem in $T$ rounds, we could also solve its relaxation $\Pi_0$ in $T$ rounds, and hence we could solve $\re(\Pi_0)$ in $T-1$ rounds, and hence we could solve its relaxation $\Pi_1$ in $T-1$ rounds, etc., and eventually we could solve $\Pi_T$ in $0$ rounds, but this was not possible. If we can do this for, e.g., $T = \Delta$, we would have the result that we want: iterated GHZ cannot be solved in $o(\Delta)$ rounds.
In essence, our task is to \textbf{guess a good sequence of problems} such that all of the above holds. How to do that?

\paragraph{Comparison with bipartite maximal matching.}

At this point it is useful to compare the situation with the \emph{bipartite maximal matching} problem \cite{Balliu2019,BalliuBHORS21}. There is a very simple classical LOCAL algorithm that solves the problem in $O(\Delta)$ rounds, but for decades it was unknown whether this is optimal. The breakthrough result from 2019 \cite{Balliu2019} eventually established this, by constructing a suitable problem sequence $\Pi_0, \Pi_1, \dotsc, \Pi_T$, and showing that it indeed has all the right properties: $\Pi_{i+1}$ is a relaxation of $\re(\Pi_i)$.

Bipartite maximal matching is conceptually a very simple problem; it can be specified in the appropriate round elimination formalism with only $3$ labels, independently of the degree. Hence there is a very simple starting point $\Pi_0$. It is possible to compute $\re(\Pi_0)$, $\re(\re(\Pi_0))$, etc., write down explicitly some of the first problems in this sequence for some small value of $\Delta$, study their structure, make some educated guesses on possible relaxations, and eventually with a sufficient amount of trial and error, human intuition, and computational help one can make a reasonable guess of the problem sequence $\Pi_0, \Pi_1, \dotsc, \Pi_T$, generalize this to an arbitrary $\Delta$, and verify that it indeed has the right properties.

Iterated GHZ seems to be fundamentally different. The problem description is relatively long and complicated, and when we increase $\Delta$, the size of the problem description increases exponentially. Studying the sequence $\re(\Pi_0), \re(\re(\Pi_0)), \dotsc$ is not particularly enlightening to a human being, and even if there are educated guesses of suitable relaxations, they are very hard to verify, even with the help of computers.

In this work we made the following surprising discovery: even though the maximal matching problem and the iterated GHZ problem do not seem to have any relation between them, beyond the fact that both of them admit an $O(\Delta)$-round classical algorithm, we can nevertheless \textbf{apply qualitatively similar relaxations in both maximal matching and iterated GHZ} to construct a suitable problem sequence! This technique is likely to generalize to a broad range of other problems.

\paragraph{Finding the right relaxations.}

Now we come to the mysterious part: when we look at the problems at the \emph{right abstraction level}, we can prove a lower bound for both maximal matching and iterated GHZ using the ``same'' relaxations, in a sense, even though the problems, their alphabets, and their constraints do not seem to have anything in common.

\emph{We do not know why this technique works, and we do not know how far it generalizes.} But we do know that it seems to be applicable to many previously-studied variants of the maximal matching problem, and it seems to be applicable to both iterated CHSH and iterated GHZ. The reader does not need to believe in this technique or trust that it is correct---this was merely a heuristic rule that enabled us to guess the right problem sequence $\Pi_0, \Pi_1, \dotsc, \Pi_T$, and once we have the sequence, we can prove that it indeed satisfies all the right properties. But as the technique is apparently applicable to two very different problem families, it may find other surprising uses, and hence we outline it here.

The following description is primarily intended for a reader who is already experienced in the use of round elimination; we refer to \cite{Brandt2019, roundeliminator,hs21} for additional background. But for all other readers this hopefully demonstrates that even though the problem sequence that we will present later in this paper may look at first arbitrary and appearing out of thin air, there is a relatively short procedure that can be used to mechanically generate it.

When we compute $\Pi' = \re(\Pi)$, each label in $\Pi'$ represents a \emph{nonempty subset} of labels in $\Pi$; singleton sets represent \emph{old labels} that already existed in $\Pi$, while all other sets are \emph{new labels}. Now we can construct a directed graph $G$ that represents the relative strength of labels in $\Pi'$ from the perspective of passive constraints: we will have an arrow $a \to b$ in the graph if the passive constraints satisfy the property that any configuration remains valid if you replace any $a$ with $b$ (of course this might violate active constraints). The transitive reduction of graph $G$ is called the \emph{diagram} for $\Pi'$ in this context.
Equipped with these definitions, we then find any pair of labels $a$ and $b$ that satisfy all these properties:
\begin{itemize}[noitemsep]
  \item $a$ is a new label,
  \item $a$ has a nonzero indegree in the diagram,
  \item $b$ is the only successor of $a$ in the diagram,
  \item $b$ has outdegree zero or $b$ is a new label.
\end{itemize}
If such a pair exists, we merge the labels $a$ and $b$ to obtain a relaxation of $\Pi'$. We then update the diagram and repeat the procedure until no such pair of labels $a$ and $b$ exists.

It turns out that following exactly this procedure will be able to reproduce the results of \cite{Balliu2019,trulytight}, but it will also let us find exactly the right relaxations that can be used to show that the iterated GHZ problem cannot be solved in $o(\Delta)$ rounds with classical LOCAL algorithms. As a final remark we note that odd and even steps of the round elimination sequence are fundamentally different and here we applied this rule in even steps.

\subsection{Implications and discussion}

Our work presents the first local problem (more precisely, a family of LCL problems) that admits a distributed quantum advantage in the LOCAL model of distributed computing. In this section we will discuss some further implications and connections with other areas.

\paragraph{Communication complexity.}

As we discussed earlier, there is much more prior work on the quantum advantage in the CONGEST model than in the LOCAL model. One reason for this is that the CONGEST model is similar in spirit to questions in communication complexity. In essence, if we have a result in communication complexity that shows that $b$ qubits can be used to solve a task that $O(b)$ classical bits cannot solve, one can try to turn that into an analogous result that separates the classical CONGEST model from the quantum-CONGEST model. However, when we compare classical LOCAL and quantum-LOCAL, we cannot use any communication complexity arguments. In essence our goal is to show that \emph{unbounded quantum communication is stronger than unbounded classical communication}. In essence, we show that being able to communicate $b$ qubits is strictly stronger than being able to communicate $2^b$ or even $2^{2^b}$ classical bits!

\paragraph{Circuit complexity.}

The question of distributed quantum advantage for LCLs is also related to establishing an \emph{efficiently-verifiable} quantum advantage in circuit complexity. The quantum advantage obtained in \cite{LeGallNR19} was obtained by ``lifting'' to the distributed setting a quantum advantage for circuits by Bravyi, Gosset and K\"onig \cite{Bravyi+18}, which was inspired by an older work in quantum information \cite{Barrett+PRA07}. Ref.~\cite{Bravyi+18} showed that a variant of the global graph problem from \cite{LeGallNR19} can be solved by a constant-depth quantum circuit but cannot be solved by any small-depth classical circuit. This computational problem, however, could not be checked by a quantum constant-depth circuit, leading to the question whether a similar advantage can be proved for an efficiently-verifiable computational problem. Since this question is still wide open in quantum circuit complexity, such a simple lifting strategy cannot be used to tackle our question. Indeed, our separation uses an approach that does not seem to be applicable to circuit complexity.

\paragraph{Lower-bound proof techniques.}

Our work is the first proof that shows that round elimination is able to separate classical LOCAL and quantum-LOCAL. As a corollary, we now know that round elimination is not applicable in the quantum-LOCAL model: it is possible that there is a round elimination sequence $\Pi_0, \Pi_1, \dotsc, \Pi_T$, but the existence of this sequence has no implications on the round complexity in the quantum-LOCAL model. This also represents a new formal barrier for resolving major open questions related to distributed quantum advantage: it is not possible to prove quantum lower bounds for e.g.\ $3$-coloring cycles or for finding sinkless orientations by constructing a round elimination sequence.

\paragraph{Networks of quantum games.}

Our problem, iterated GHZ, is based on the idea that we build an LCL problem that asks us to play a number of quantum games, so that the games are represented by black nodes and the players are represented by white nodes. Then inside each white node the output of one game is fed as input to another game. We call such a problem a \emph{network of quantum games}. Such a problem can be always solved in $O(1)$ rounds in quantum-LOCAL, as we can create shared entangled qubits in the first round, and then all further steps can be done without any additional communication.

Our work shows that the classical complexity of such an LCL can be as high as $\Omega(\Delta)$. Now the obvious question is whether we can push this further: can we maybe engineer a network of quantum games such that its classical complexity is super-constant as a function of the number of nodes $n$, even if $\Delta$ is constant? For example, could we show a lower bound of $\Omega(n)$ or $\Omega(\log n)$ or maybe $\Omega(\log^* n)$ for such a problem?

It turns out that the answer is \textbf{no}. We generalize the idea so that instead of quantum games, we can have a network of any non-signaling games (regardless of whether they happen to have quantum strategies or not); we call such a problem a \emph{network of non-signaling games}. We show that if $\Pi$ is any LCL problem that is constructed from networks of non-signaling games (even if we connect the games with each other using arbitrarily complicated arithmetic circuits), it can be always solved in $O(1)$ rounds with classical LOCAL algorithms; here the constant depends on the specific LCL problem $\Pi$ (and in particular, if we have a family of LCL problems parameterized by $\Delta$, the constant may depend on $\Delta$). However, the running time will be independent of $n$.

The key reason for this is that any non-signaling game is \emph{completable}: for example, for a two-player game this means that for any Alice's input $x$ there is some Alice's output $a$ such that for any Bob's input $y$ there is still a valid Bob's output $b$, and vice versa.

We can exploit completability to solve any $\Pi$ in a distributed manner as follows. Each white node starts to process its own arithmetic circuit in a topological order. As soon as it encounters a step that involves a game, it sends a message to the black neighbor responsible for that specific game, together with its own input for that game. The black nodes keep track of the inputs they have seen so far, and they always pick \emph{safe} outputs for those players. This way in one pair of rounds all white nodes can learn their own output for the game that appears first in their own circuit, and we can repeat this for each game in a sequential order---thanks to completability, black nodes will be always able to assign valid outputs also for players that join the game late. This way the running time will be proportional to the size of the circuit held by a single white node, and hence for a fixed LCL (with a finite set of possible local circuits) it will be bounded by some constant.

The result may at first seem counterintuitive: A non-signaling game may require the players to \emph{break symmetry} (e.g., in the CHSH game, if the inputs are both $1$, the outputs must be different). Hence it would seem that we could put together non-signaling games in cycle so that the input does not break symmetry, but the output has to break symmetry. Then we could apply the classical result that shows that any symmetry-breaking LCL requires $\Omega(\log^* n)$ rounds \cite{Linial1992}. However, the above algorithm demonstrates that \emph{any network of non-signaling games admits a symmetric solution in those neighborhoods where the input (together with the port numbering) does not break symmetry}.

This represents a new barrier for proving distributed quantum advantage: the idea of networks of non-signaling games can only prove $f(\Delta)$ separations between classical LOCAL and quantum-LOCAL, not $f(n)$ separations.

\subsection{Roadmap}

We introduce some basic definitions in \cref{sec:preliminaries}. In \cref{sec:iterated-ghz} we give the formal definition of the iterated GHZ problem in the right formalism that we need for round elimination, and we show that there is indeed a one-round quantum-LOCAL algorithm for solving this problem. In \cref{sec:lower-bound} we prove our main result: iterated GHZ requires $\Omega(\Delta)$ rounds in classical LOCAL. Finally, in \cref{sec:networks} we formalize the idea of networks of non-signaling games, and show that any fixed LCL in this problem family can be solved in constant time in classical LOCAL.

\section{Preliminaries}\label{sec:preliminaries}

We write $\N_0$ for the set of non-negative integers, $\N_+$ for that of
positive integers, and $[n] = \{ 1,\dots,n \}$ for the set of the first $n \in
\N_0$ positive integers.
For a relation $R \subseteq A \times B$ and $x \in A$, we write $R(x)$ for the
set $\{ y \in B \mid (x,y) \in R \}$.

For $m \in \N_0$ and a set $\Sigma$, we write $\Sigma^m$ for the set of
strings over elements of $\Sigma$ that have length $m$.
We will refer to such strings in two different ways, both as concatenations of
$m$ many elements of $\Sigma$ (e.g., $x = x_1 \cdots x_m$ where $x_i \in
\Sigma$) and as maps $[m] \to \Sigma$.
These two formalisms may be used interchangeably without confusion.
For a permutation $\sigma\colon [m] \to [m]$ and a string $x = x_1 \cdots x_m
\in \Sigma^m$, we let $x^\sigma = x_{\sigma(1)} \cdots x_{\sigma(m)}$.

Let $G = (V,E)$ be a graph and $v \in V$ a node of $G$.
Since $E$ is a relation, recall we may use our notation for relations to write
$E(v)$ as shorthand for the set $\{ e \in E \mid \text{$e$ is incident to $v$}
\}$.
For any node \(v\) of a graph \(G = (V,E)\), we denote by \(\deg_G(v)\) its degree in \(G\).
If the underlying graph is clear from the context, we omit the subscript and write simply \(\deg(v)\).

In this paper we are interested in locally checkable labeling (LCL) problems, first introduced by Naor and Stockmeyer in \cite{NaorS95}. 
We now define this class of problems using a formalism that is slightly different from the original one.

\paragraph{LCLs in the black-white formalism.}
While there are different ways to define LCLs, in this paper we consider the so-called \emph{LCLs in the black-white formalism}. Such a formalism, compared to the definition given by Naor and Stockmeyer in \cite{NaorS95}, is more suitable for applying the round elimination technique.

In this formalism, problems are defined on bipartite graphs.
We refer to nodes in the first set as \emph{white} nodes and to nodes in the
other as \emph{black} nodes.
The connection to a general graph $G$ is that we see white nodes as nodes of $G$
and black nodes as edges of $G$.
In this sense, we are able to handle constraints for nodes and edges in a
uniform manner.

\begin{definition}[LCL problem in the black-white formalism]%
  \label{def:black-white-lcl}
  An \emph{LCL problem in the black-white formalism} is a tuple $\Pi =
  (\Sigma_\iin, \Sigma_\oout, \CC_\white,
  \CC_\black)$ where:
  \begin{itemize}
    \item $\Sigma_\iin$ and $\Sigma_\oout$ are finite sets of \emph{labels}.
    \item \(\CC_\white \) is the \emph{white constraint}, formally a set of
    multisets where each element of a multiset is a pair of labels \((\ell_\iin,
    \ell_\oout) \in \Sigma_\iin \times \Sigma_\oout\).
    \item  \(\CC_\black \) is the \emph{black constraint}, formally a set of
    multisets where each element of a multiset is a pair of labels \((\ell_\iin,
    \ell_\oout) \in \Sigma_\iin \times \Sigma_\oout\).
  \end{itemize}
  Let $G = (V,E)$ be a bipartite graph where $V = W \cup B$ is a partition of
  $V$ that implicitly \enquote{colors} the nodes white and black by placing them
  in $W$ and $B$, respectively.
  In addition, let $i\colon E \to \Sigma_\iin $ be an
  \emph{input}, that is, a labeling of edges such that $i(e) \in \Sigma_\iin$ for each \(e \in E\).
  An assignment $o \colon E \to \Sigma_\oout$ is a \emph{solution} to $\Pi$ if
  \emph{it satisfies both the white and the black constraints}, that is:
  \begin{itemize}
    \item For every white node $w \in W$, if \(e_1, \dots, e_{d_w}\) are the
    edges that are incident to \(w\), then $\{(i(e_1),o(e_1)), \dots,
    (i(e_{d_w}),o(e_{d_w}))\} \in \CC_\white$.
    \item For every black node $b \in B$, if \(e_1, \dots, e_{d_b}\) are the
    edges that are incident to \(b\), then we have $\{(i(e_1),o(e_1)), \dots,
    (i(e_{d_b}),o(e_{d_b}))\} \in \CC_\black$.
  \end{itemize}
\end{definition}

In general graphs, the black-white formalism captures a strict subset of LCLs as
defined originally by \cite{NaorS95}; however, in the LOCAL model of computation (whose definition follows later in this section) the two notions are equivalent in
general trees and high-girth graphs, that are the classes of graphs needed to
apply round-elimination \cite{hideandseek}.

We first introduce the port-numbering model, and on top of it we define the LOCAL model of computation \cite{Linial1992} and the quantum-LOCAL model \cite{GavoilleKM09}.
We stick to the black-white formalism to keep all definitions consistent.

\paragraph{The port-numbering model.}
A port-numbered network is a triple \(N = (V, P, p)\) where \(V\) is the set of nodes, \(P\) is the set of \emph{ports}, and \(p: P \to P\) is a function specifying connections between ports.
Each element \(x\in P\) is a pair \((v,i)\) where \(v \in V\), \(i \in \N_+\).
The connection function \(p\) between ports is an involution, that is, \(p(p(x)) = x\) for all \(x \in P\).
If \((v,i) \in P\), we say that \((v,i)\) is port number \(i\) in node \(v\).
The degree of a node \(v\) in the network \(N\) is \(\deg_N(v)\) is the number of ports in \(v\), that is, \(\deg_N(v) = \abs{\{i \in \N: (v,i) \in P\}}\).
We assume that port numbers are consecutive, i.e., the ports of any node \(v \in V \) are \((v,1), \dots, (v, \deg_N(v))\).
Clearly, a port-numbered network identifies an \emph{underlying graph} \(G = (V,E)\) where, for any two nodes \(u,v \in V\), \(\{u,v\} \in E\) if and only if there exists ports \(x_u,x_v \in P\) such that \(p(x_u) = x_v\).
Here, the degree of a node \(\deg_N(v)\) corresponds to \(\deg_G(v)\).

In the port-numbering model we are given a distributed system consisting of a port-numbered network of \( \abs{V} =
n\) \emph{processors} (or \emph{nodes}) that operates in a sequence of
synchronous rounds.
In each round the processors may perform unbounded computations on their
respective local state variables and subsequently exchange of messages of
arbitrary size along the links given by the underlying input graph.
Nodes identify their neighbors by using ports as defined before, where the port assignment may be done adversarially.
Barring their degree, all nodes are identical and operate according to the
same local computation procedures.
Initially all local state variables have the same value for all processors;
the sole exception is a distinguished local variable \(\mathrm{x}(v)\) of
each processor \(v\) that encodes input data.

Let \(\Sigma_{\text{in}}\) be a set of input labels.
The input of a problem is defined in the form of a labeled, bipartite graph \((G,i)\) where \(G = (V = B \cup W, E)\), \(V\) is the set of
processors (hence it is specified as part of the input), and \(i\colon E
\to \Sigma_{\text{in}}\) is an assignment of an input label \(i(v) \in \Sigma_{\text{in}}\) to
each edge \(e \in E\).
The nodes in \(V\) are properly \(2\)-colored with either color \emph{black} (elements of \(B\))  or \emph{white} (elements of \(W\)).
The output of the algorithm is given in the form of an assignment of local output
labels \(o\colon E \to \Sigma_{\text{out}}\), and the algorithm is assumed to
terminate once all labels \(o(v)\) are definitely fixed.
These outputs are determined and assigned by the white nodes, which are also referred to as \emph{active nodes}, whereas black nodes are called \emph{passive nodes} (they do not output anything).
We assume that nodes and their links are fault-free.
The local computation procedures may be randomized by giving each processor
access to its own set of random variables; in this case, we are in the
\emph{randomized} port-numbering model as opposed to the \emph{deterministic}
port-numbering model.
If the algorithm is randomized, we also require that the failure probability while solving any problem is at most \(1/n\), where \(n\) is the size of the input graph.

The running time of an algorithm is the number of synchronous rounds required by all nodes to produce output labels.
If an algorithm running time is \(T\), we also say that the algorithm has locality \(T\).
Notice that \(T\) can be a function of the size of the input graph. 

We remark that the notion of an (LCL) problem is a graph problem, and does not depend on the specific model of computation we consider (hence, the problem definition cannot depend on, e.g., port numbers).

\paragraph{\boldmath The LOCAL model.}

The LOCAL model was first introduced by \cite{Linial1992}, here adapted to the black-white formalism.
It is just the port-numbering model augmented with an assignment of unique identifiers to nodes.
Let \(c \ge 1\) be a constant, and let \(\Sigma_{\text{in}}\) be a set of input labels.
The nodes of the input graph \(G = (V,E)\) are given as input also unique identifiers specified by an injective function \(\text{id}: V \to [n^c] \).
This assignment might be adversarial and is stored in the local state variable \(\mathrm{x}(v)\), and nodes can exploit these values during their local computation.

The local computation procedures may be randomized by giving each processor
access to its own set of random variables; in this case, we are in the
\emph{randomized} LOCAL model as opposed to the \emph{deterministic}
LOCAL model.
If the algorithm is randomized, we also require that the failure probability while solving any problem is at most \(1/n\), where \(n\) is the size of the input graph.

\paragraph{The quantum-LOCAL model.}
The quantum-LOCAL of computing is similar to the deterministic \local model above, but now with quantum computers and quantum communication links. More precisely, the quantum computers manipulate local states consisting of an unbounded number of qubits with arbitrary unitary transformations, the communication links are quantum communication channels (adjacent nodes can exchange any number of qubits), and the local outputs can be the result of any quantum measurement.

\section{Iterated GHZ problem}
\label{sec:iterated-ghz}

We start by defining the \emph{iterated GHZ} problem.
While in \Cref{sec:preliminaries} we provided the definition of LCLs in their full generality, in this section we consider a restricted setting, i.e., the case in which the constraints do not depend on the provided input, and the considered graph is biregular.
Hence, in this section (and in the following one), an LCL problem $\Pi$ in the black-white formalism is simply a tuple $(\Sigma_\Pi, \nodeconstone_\Pi, \edgeconstone_\Pi)$ where $\Sigma_\Pi$ is a finite set of elements, called (output) \emph{labels}, $\nodeconstone_\Pi$ is a collection of cardinality-$\Delta$ multisets of labels from $\Sigma_\Pi$, and $\edgeconstone_\Pi$ is a collection of cardinality-$3$ multisets of labels from $\Sigma_\Pi$.
We call $\nodeconstone_\Pi$ the \emph{white constraint} of $\Pi$ and $\edgeconstone_\Pi$ the \emph{black constraint} of $\Pi$.
We will use the term \emph{configuration} to refer to any multiset of labels of cardinality $\Delta$, resp.\ $3$; we will call a configuration of cardinality $\Delta$, resp.\ $3$, an \emph{allowed configuration} if it is contained in the white constraint, resp.\ black constraint, of the respectively considered LCL problem.
For simplicity, we will write a configuration $\{ \ell_1, \dots, \ell_k\}$ (which, technically, is a multiset) in the form $\ell_1~\dots~\ell_k$ (where $k \in \{ 3, \Delta \}$).

For degree $\Delta$, the iterated GHZ problem is defined as an LCL $\mathcal{P}_\Delta =
(\Sigma_{\mathcal{P}_\Delta},\nodeconst_{\mathcal{P}_\Delta},\edgeconst_{\mathcal{P}_\Delta})$
in the black-white formalism (in the aforementioned restricted setting).
While the constraints of this problem do not depend on the provided input, we will consider lower and upper bounds for this problem in the setting in which black nodes $B$ are provided with a labeling $c : B \rightarrow \{1,\ldots,\Delta\}$ satisfying that, if two black nodes $u$ and $v$ are incident to the same white node, then $c(u) \neq c(v)$. In the following, this input will simply be called $\Delta$-edge coloring (where the name comes from the fact that we could think of a $(\Delta,3)$-biregular graph as a hypergraph of degree $\Delta$ and rank $3$, and we consider a coloring of the hyperedges).

\paragraph{Labels.}
The label set is defined as:
\begin{align*}
	\Sigma_{\mathcal{P}_\Delta} = 
	~ & \{\reMY{1}{y} \mid y \in \{0,1\}\}\\
	\cup ~ &\{\reXY{j}{x\,}{y} \mid 2 \le j \le \Delta \text{ and }x,y \in \{0,1\}\}.
\end{align*}

\paragraph{White constraint.}
For all vectors $(y_1, \dots, y_{\Delta}) \in \{0, 1\}^{\Delta}$,
$\nodeconst_{\mathcal{P}_\Delta}$ contains the configuration $L_1 ~ \ldots ~ L_\Delta$ satisfying the following properties.
\begin{itemize}[noitemsep]
	\item $L_1 = \reMY{1}{y_1}$.
	\item For all $j$ such that $2 \le j \le \Delta$, $L_j = \reXY{j}{y_{j-1}\,}{y_{j}}$.
\end{itemize}

\paragraph{Black constraint.}
The black constraint $\edgeconst_{\mathcal{P}_\Delta}$ contains the following configurations.
\begin{itemize}
	\item $\reMY{1}{0} ~~~ \reMY{1}{0} ~~~ \reMY{1}{1}$.
	\item $\reXY{j}{x_1\,}{y_1} ~~~ \reXY{j}{x_2\,}{y_2} ~~~ \reXY{j}{x_3\,}{y_3}$, for all $2 \le j \le \Delta$ and $x_1,x_2,x_3,y_1,y_2,y_3$ satisfying:
	\begin{itemize}
		\item  $x_i \in \{0,1\}$ and $y_i \in \{0,1\}$ for all $1 \le i \le 3$;
		\item if $x_1 + x_2 + x_3$ is even, then $y_1 \oplus y_2 \oplus y_3 = x_1 \lor x_2 \lor x_3$.
	\end{itemize}
\end{itemize}

\paragraph{Some intuition.}
Informally, each label encodes an input and an output to a GHZ game.
The white constraint enforces that the output of one game flows correctly to become the input of the next game.
The black constraint enforces that, on a black node, the GHZ game is played correctly  by the neighboring white nodes, except for the first black constraint, which is a special case and only requires to break symmetry between adjacent white nodes.
This special case is required to keep the game from becoming trivial.

\paragraph{Considered setting.}
The problem $\mathcal{P}_\Delta$ is only defined on $(\Delta,3)$-biregular graphs. However, it is trivial to lift this definition to arbitrary graphs: white nodes of degree $\Delta$ and black nodes of degree $3$ need to satisfy the constraints of $\mathcal{P}_\Delta$; all other nodes can produce an arbitrary output.

Intuitively, the lower bound that we will prove can be interpreted as follows: for any $n$ and any small-enough $\Delta$, there exists a balanced tree of $n$ nodes where all white nodes that are ``far enough'' from the leaves have degree $\Delta$, all black nodes that are ``far enough'' from the leaves have degree $3$, and there exists a white node that is ``far enough'' from the leaves that needs to spend $\Omega(\Delta)$ rounds. Observe that restricting the possible input instances only makes our lower bound stronger (i.e., this problem is hard \emph{already} on regular trees).

We will provide upper bounds for $(\Delta,3)$-biregular graphs. However, given that the definition of $\mathcal{P}_\Delta$ lifted to arbitrary graphs allows an arbitrary output for white nodes of degree different from $\Delta$ and black nodes of degree different from $3$, it is trivial to adapt our algorithm to the case of graphs with nodes of arbitrary degree.

\subsection{Quantum algorithm for iterated GHZ}

We now prove that the iterated GHZ problem can be solved in $O(1)$ rounds in the quantum LOCAL model, no matter how large $\Delta$ is. More formally, we devote the rest of the section to proving the following theorem.
\begin{theorem}
	For any integer $\Delta > 0$, the iterated GHZ problem can be solved in $O(1)$ rounds in the quantum LOCAL model, if a $\Delta$-edge coloring is provided.
\end{theorem}

\paragraph{Quantum algorithm for iterated GHZ problem.}
Each black node receives a color~$c \in [\Delta]$ as an input; recall that this is what we call a $\Delta$-edge coloring.
Each black node prepares three qubits~$q_1, q_2, q_3$ in the GHZ state, that is
\begin{equation*}
	\ket{0}_{q_1}\ket{0}_{q_2}\ket{0}_{q_3} \rightarrow \frac{1}{\sqrt{2}}\left(\ket{0}_{q_1}\ket{0}_{q_2}\ket{0}_{q_3}+\ket{1}_{q_1}\ket{1}_{q_2}\ket{1}_{q_3}\right) .
\end{equation*}
It then sends message~$(c, i, q_i)$ to port~$i$.
Note that this is a message containing a qubit belonging to a three-party entangled state.
Then the black node stops.

The white nodes wait for the messages from the black nodes.
Let $(c_j', i_j', q_j')$ be the message a white node receives from port~$j$.
Recall that~$c_j$ is just the color of the $\Delta$-edge coloring, and hence all colors $c_1', c_2', \dots, c_\Delta'$ form a permutation.
Let us reindex the messages according to this permutation, i.e. $(k, i_k, q_k)$ is the message received from the black neighbor with color~$k$.

The white node first breaks the symmetry based on $i_1$.
If $i_1 = 2$, it sets $g_1 = 1$, otherwise $g_1 = 0$.
Now the white node plays sequentially the GHZ games locally with no further communication by exploiting the GHZ states distributed by the black nodes.
For each game~$k \in \{2, 3, \dots, \Delta\}$, in order starting with $k=2$, the white node measures the qubit $q_k$, with input $g_{k-1}$, according to  
the well-known quantum strategy~\cite{mermin1990quantum,brassard2005quantum} for winning the GHZ game perfectly. This gives the measurement outcome $g_k$, which is then used for the $(k+1)^{
\rm th}$ GHZ game, until all $\Delta$ games have been played.

Finally, the white node produces an output for each of its incident edges.
For edge neighboring color~$1$, the node outputs~$\reMY{1}{g_1}$.
For edge neighboring color~$k \in \{2, 3, \dots, \Delta\}$, the node outputs~$\reXY{k}{g_{k-1}\,}{g_k}$.
Then the node stops.

Now the only thing left is to prove that this algorithm produces a correct output.
\begin{proof}
	It is easy to see that the white constraint is satisfied:
	there is exactly one edge with each subscript, and the labels~$g_i$ propagate correctly between the labels.

	The black constraint is also easy to check:
	Nodes with input color~$1$ satisfy the first constraint, as exactly one neighbor received message~$(1, 2, q_2)$, and only that node outputs~$\reMY{1}{1}$; the other two neighbors output~$\reMY{1}{0}$.
	Nodes with input color~$k \in \{2, \dots, \Delta\}$ satisfy the latter constraint.
	In particular, all of them have subscript~$k$ and describe the input and output to a GHZ game.
	Moreover, as the algorithm uses the quantum GHZ strategy to produce the outputs from the inputs, they satisfy the GHZ constraint.
\end{proof}

\section{Classical lower bound for iterated GHZ}\label{sec:lower-bound}

In this section, we prove that the iterated GHZ problem requires $\Omega(\Delta)$ rounds in the LOCAL model. More formally, we prove the following.
\begin{theorem}\label{thm:main-lb}
	Let $\Delta \ge 3$ be an integer. The iterated GHZ problem requires $\Omega(\min\{\Delta, \log_\Delta n\})$ rounds in the deterministic LOCAL model and $\Omega(\min\{\Delta, \log_\Delta \log n\})$ rounds in the randomized LOCAL model, even if a $\Delta$-edge coloring is provided.
\end{theorem}
The lower bound is based on the round elimination technique \cite{Brandt2019}. 
In \Cref{ssec:re-preliminaries}, we first provide definitions for a variety of  helpful concepts and objects and introduce round elimination.
For simplicity, we will introduce all the definitions in the context of the LCL problems relevant in this paper; in particular, we will assume that black nodes have degree $3$.
After giving an overview of the structure of the proof in \Cref{ssec:roadmap}, we present the whole proof in Sections \ref{ssec:re-def} to \ref{ssec:re-lift}.

\subsection{Preliminaries}\label{ssec:re-preliminaries}
\paragraph{Definitions.}
We will use so-called \emph{condensed configurations} to represent certain sets of configurations efficiently.
Formally, a condensed configuration is a configuration of sets of labels from $\Sigma_\Pi$.
For a condensed configuration $C = S_1~\dots~S_k$, by \emph{configurations represented by $C$} we denote the set of all configurations $\ell_1~\dots~\ell_k$ such that $\ell_j \in S_j$ for each $1 \leq j \leq k$.
To distinguish a condensed configuration (which represents a set of configurations of labels from $\Sigma$) from a configuration of sets, we will use the term \emph{disjunction} instead of ``set'' and write the disjunction consisting of some labels $\ell^{(1)}, \dots, \ell^{(z)}$ as $[\ell^{(1)}, \dots, \ell^{(z)}]$.
However, when convenient, we may consider a disjunction of labels as the set of the same labels.
\begin{definition}[Picking a configuration]
	We say that we can \emph{pick} a configuration $\ell_1,\ldots,\ell_k$ from a configuration $S_1,\ldots,S_k$ of sets of labels, if there exists a permutation $\sigma \colon \{1, \dots, k\} \to \{1, \dots, k\}$ such that, for each $1 \leq i \leq k$, it holds that $\ell_i \in S_{\sigma(i)}$. Moreover, we say that we can pick a configuration $C$ from a set $\mathcal{C}$ of configurations of sets of labels if there exists a configuration $C' \in \mathcal{C}$ from which we can pick $C$. 
\end{definition}
By configurations represented by a set of condensed configurations, we denote the set of configurations that can be picked from the condensed configurations.

\paragraph{Round elimination.}
On a high level, the round elimination framework introduced in~\cite{Brandt2019} provides a mechanism to prove a lower bound for a given LCL problem $\Pi$ by generating a sequence of problems starting with $\Pi$ such that every problem has a complexity of exactly one round less than the previous problem (under some mild assumptions).
The underlying intuition for obtaining a lower bound with this mechanism is that if one can show that the problem obtained after $k$ steps in the sequence cannot be solved in $0$ rounds, then this yields a lower bound of $k - 1$.
We remark that this is a highly simplified description of the overall approach with a lot of details omitted; in the following, we introduce round elimination more formally.

\paragraph{The round elimination operator $\re()$.}
For a given LCL problem $\Pi = (\Sigma_\Pi, \nodeconstone_\Pi, \edgeconstone_\Pi)$, the problem $\re(\Pi) = (\Sigma_{\re(\Pi)}, \nodeconstone_{\re(\Pi)}, \edgeconstone_{\re(\Pi)})$ is defined as follows.
\begin{itemize}
	\item $\edgeconstone_{\re(\Pi)}$ is the set of all configurations $S_1~S_2~S_3$ such that
	\begin{itemize}
		\item for all $1 \leq j \leq 3$, $S_j$ is a nonempty subset of $\Sigma_\Pi$,
		\item for each tuple $(\ell_1, \ell_2, \ell_3) \in S_1 \times S_2 \times S_3$, we have that $\ell_1~\ell_2~\ell_3 \in \edgeconstone_\Pi$, and
		\item $S_1~S_2~S_3$ is \emph{maximal}, i.e., there exists no configuration $S'_1~S'_2~S'_3$ of nonempty subsets of $\Sigma_\Pi$ such that
		\begin{itemize}
			\item for all $1 \leq j \leq 3$, $S_j \subseteq S'_j$,
			\item there exists some $1 \leq j \leq 3$ such that $S_j \subsetneq S'_j$, and
			\item for each tuple $(\ell'_1, \ell'_2, \ell'_3) \in S'_1 \times S'_2 \times S'_3$, we have that $\ell'_1~\ell'_2~\ell'_3 \in \edgeconstone_\Pi$.
		\end{itemize}
	\end{itemize}
	\item $\Sigma_{\re(\Pi)}$ is the set of all nonempty subsets of $\Sigma$ that appear in at least one configuration contained in $\edgeconstone_{\re(\Pi)}$.
	\item $\nodeconstone_{\re(\Pi)}$ is the set of all configurations $S_1~\dots~S_{\Delta}$ of sets contained in $\Sigma_{\re(\Pi)}$ such that there exists some tuple $(\ell_1, \dots, \ell_{\Delta}) \in S_1 \times \dots \times S_{\Delta}$ satisfying $\ell_1~\dots~\ell_{\Delta} \in \nodeconstone_\Pi$.
\end{itemize}

\paragraph{The round elimination operator $\rere()$.}
For a given LCL problem $\Pi = (\Sigma_\Pi, \nodeconstone_\Pi, \edgeconstone_\Pi)$, the problem $\rere(\Pi) = (\Sigma_{\rere(\Pi)}, \nodeconstone_{\rere(\Pi)}, \edgeconstone_{\rere(\Pi)})$ is defined as follows.
\begin{itemize}
	\item $\nodeconstone_{\rere(\Pi)}$ is the set of all configurations $S_1~\dots~S_{\Delta}$ such that
	\begin{itemize}
		\item for all $1 \leq j \leq \Delta$, $S_j$ is a nonempty subset of $\Sigma_\Pi$,
		\item for each tuple $(\ell_1, \dots, \ell_{\Delta}) \in S_1 \times \dots \times S_{\Delta}$, we have that $\ell_1~\dots~\ell_{\Delta} \in \nodeconstone_\Pi$, and
		\item $S_1~\dots~S_{\Delta}$ is \emph{maximal}, i.e., there exists no configuration $S''_1~\dots~S''_{\Delta}$ of nonempty subsets of $\Sigma_\Pi$ such that
		\begin{itemize}
			\item for all $1 \leq j \leq \Delta$, $S_j \subseteq S''_j$,
			\item there exists some $1 \leq j \leq \Delta$ such that $S_j \subsetneq S''_j$, and
			\item for each tuple $(\ell''_1, \dots, \ell''_{\Delta}) \in S''_1 \times \dots \times S''_{\Delta}$, we have that $\ell''_1~\dots~\ell''_{\Delta} \in \nodeconstone_\Pi$.
		\end{itemize}
	\end{itemize}
	\item $\Sigma_{\rere(\Pi)}$ is the set of all nonempty subsets of $\Sigma_\Pi$ that appear in at least one configuration contained in $\nodeconstone_{\rere(\Pi)}$.
	\item $\edgeconstone_{\rere(\Pi)}$ is the set of all configurations $S_1~S_2~S_3$ of sets contained in $\Sigma_{\rere(\Pi)}$ such that there exists some tuple $(\ell_1, \ell_2, \ell_3) \in S_1 \times S_2 \times S_3$ satisfying $\ell_1~\ell_2~\ell_3 \in \edgeconstone_\Pi$.
\end{itemize}

\paragraph{A simple way to compute $\nodeconstone_{\re(\Pi)}$ and $\edgeconstone_{\rere(\Pi)}$.}
It has been observed in previous works that use the round elimination technique (see, e.g, \cite{hideandseek}), that $\nodeconstone_{\re(\Pi)}$ can be easily computed as follows. 
Start from all the configurations in $\nodeconstone_\Pi$, and for each configuration $C \in \nodeconstone_\Pi$, add to $\nodeconstone_{\re(\Pi)}$ all the configurations that can be picked from the condensed configuration obtained by replacing each label $\ell$ in $C$ with the disjunction of all label sets in $\Sigma_{\re(\Pi)}$ that contain $\ell$.

Similarly,  $\edgeconstone_{\rere(\Pi)}$ can be easily computed as follows. 
Start from all the configurations in $\edgeconstone_\Pi$, and for each configuration $C \in \edgeconstone_\Pi$, add to $\edgeconstone_{\rere(\Pi)}$ all the configurations that can be picked from the condensed configuration obtained by replacing each label $\ell$ in $C$ with the disjunction of all label sets in $\Sigma_{\rere(\Pi)}$ that contain $\ell$.

Throughout this section, in order to compute $\nodeconstone_{\re(\Pi)}$ and $\edgeconstone_{\rere(\Pi)}$, we will implicitly use the described procedures.

\paragraph{Further notation and terminology.}
For a set $\mathcal L$ of configurations of labels and a configuration $S = S_1~\dots~S_k$ of sets of labels, we say that $S$ \emph{satisfies the universal quantifier} if, for each choice $(\ell_1, \dots, \ell_k) \in S_1 \times \dots \times S_k$, we have $\ell_j \in S_j$ for each $1 \leq j \leq k$.
A configuration $S_1~\dots~S_k$ satisfying the universal quantifier is called \emph{maximal} if there is no configuration $S'_1~\dots~S'_k$ satisfying the universal quantifier for which $S_j \subseteq S'_j$ for all $1 \leq j \leq k$ and $S_j \subsetneq S'_j$ for at least one $1 \leq j \leq k$.
In particular, $\nodeconstone_{\rere(\Pi)}$, resp.\ $\edgeconstone_{\re(\Pi)}$, is precisely the set of all maximal configurations satisfying the universal quantifier (for $\mathcal L = \nodeconstone_\Pi$, resp.\ $\mathcal L = \edgeconstone_\Pi$).
Utilizing the fact that a disjunction is nothing else than a set, we will call a set $\mathcal S$ of condensed configurations \emph{maximal-complete} if the set of configurations of sets obtained from $\mathcal S$ by interpreting each disjunction as a set contains all the set of all maximal configurations satisfying the universal quantifier (w.r.t.\ the set of label configurations represented by $\mathcal S$).

Let $(\Sigma_\Pi, \nodeconstone_\Pi, \edgeconstone_\Pi)$ be an LCL problem, and $\ell, \ell'$ two labels from $\Sigma_\Pi$.
We say that \emph{$\ell'$ is at least as strong as $\ell$ w.r.t.\ $\edgeconstone_\Pi$} (resp.\ \emph{w.r.t.\ $\nodeconstone_\Pi$}) if the following holds: for any configuration $\ell~\ell_2~\ell_3 \in \edgeconstone_\Pi$ (resp.\ for any configuration $\ell~\ell_2~\dots~\ell_{\Delta} \in \nodeconstone_\Pi$), we have $\ell'~\ell_2~\ell_3 \in \edgeconstone_\Pi$ (resp.\ $\ell'~\ell_2~\dots~\ell_{\Delta} \in \nodeconstone_\Pi$).
We may omit the dependency on the constraint when clear from the context; in particular, for the problems $\Pi_{i, \Delta}$ we define later and all problems obtained from applying $\rere$, the relative strength of labels will be considered w.r.t.\ the respective black constraint, while for all problems obtained from applying $\re$ it will be considered w.r.t.\ the respective white constraint.
If $\ell'$ is at least as strong as $\ell$ but $\ell$ is not at least as strong as $\ell'$, we say that $\ell'$ \emph{is stronger than $\ell$}.
We write $\ell \leq \ell'$ to represent that $\ell'$ is at least as strong as $\ell$; we write $\ell < \ell'$ to represent that $\ell'$ is stronger than $\ell$.

We will make use of so-called \emph{diagrams} to represent the strength relations between labels.
Formally, the diagram of a problem $(\Sigma_\Pi, \nodeconstone_\Pi, \edgeconstone_\Pi)$ (w.r.t. $\nodeconstone_\Pi$, resp.\ $\edgeconstone_\Pi$) is the directed graph where the set of nodes is $\Sigma_\Pi$ and there is a directed edge $(\ell, \ell')$ between two labels $\ell, \ell' \in \Sigma_\Pi$ if
\begin{itemize}[noitemsep]
	\item $\ell < \ell'$ and
	\item there is no label $\ell'' \in \Sigma_\Pi$ satisfying $\ell < \ell'' < \ell'$
\end{itemize}
(where the strength relations are w.r.t.\ $\nodeconstone_\Pi$, resp.\ $\edgeconstone_\Pi$).
We remark that the definition of strength implies that any diagram is acyclic.

We call two labels $\ell, \ell'$ \emph{comparable} if $\ell \leq \ell'$ or $\ell' \leq \ell$, and \emph{incomparable} otherwise.
From the definition of a diagram, it follows that two labels $\ell, \ell'$ are comparable if and only if there is a directed path between $\ell$ and $\ell'$ (in one of the two possible directions) in the diagram.
Furthermore, for two configurations $L = \ell_1~\dots~\ell_k$ and $L' = \ell'_1~\dots~\ell'_k$ of labels, we say that $L'$ dominates $L$ if $\ell_j \leq \ell'_j$ for all $1 \leq j \leq k$.
Similarly, for configurations $C = S_1~\dots~S_k$ and $C' = S'_1~\dots~S'_k$ of sets of labels, we say that $C'$ dominates $C$ if $S_j \subseteq S'_j$ for all $1 \leq j \leq k$.
(Note that as configurations are multisets, a configuration already dominates another if there is a permutation $\rho \colon \{1,\dots,k\} \rightarrow \{1,\dots,k\}$ such that $\ell_j \leq \ell'_{\rho(j)}$ for all $1 \leq j \leq k$, resp.\ $S_j \subseteq S'_{\rho(j)}$ for all $1 \leq j \leq k$.)
Moreover, for two sets  $\mathcal{C}$ and $\mathcal{C'}$ of configurations, we say that $\mathcal{C'}$ dominates  $\mathcal{C}$ if, for all $C \in \mathcal{C}$, there exists a configuration $C' \in  \mathcal{C'}$  that dominates $C$.  

For labels $\ell_1, \dots, \ell_k$, the set of all labels that are at least as strong as at least one of the $\ell_j$ is denoted by $\gen{\ell_1, \dots, \ell_k}$.
In other words, $\gen{\ell_1, \dots, \ell_k}$ is the set of all labels that are identical to one of the $\ell_j$ or a (not necessarily direct) successor of one of the $\ell_j$ in the respective diagram.
We say that $\gen{\ell_1, \dots, \ell_k}$ is \emph{generated} by $\ell_1, \dots, \ell_k$.
As before, we will omit the dependency on the constraint when clear from the context, and in particular use the conventions mentioned above. Moreover, for a set $L=\{\ell_1,\ldots,\ell_k\}$ by writing $\gen{\ell \mid \ell\in L}$ we denote $\gen{\ell_1,\ldots,\ell_k}$. Observe that $\gen{L}$ is different from $\gen{\ell \mid \ell\in L}$.

\paragraph{An alternative method to compute $\re(\Pi)$ and $\rere(\Pi)$.}
In a recent work~\cite[Section 4]{balliu2024towards}, a new method was developed for computing, for a given problem $\Pi = (\Sigma_\Pi, \nodeconstone_\Pi, \edgeconstone_\Pi)$, the problems $\re(\Pi)$ and $\rere(\Pi)$ in a different way.
On a high level, this new method is based on a very simple operation that if applied iteratively will provide the same result as the more complex application of the universal quantifier in the definitions of $\nodeconstone_{\re(\Pi)}$ and $\edgeconstone_{\rere(\Pi)}$ above.
The simplicity of said operation makes it considerably easier to apply round elimination and prove the correctness of claims such as that a certain problem is indeed what is obtained by applying $\re$ or $\rere$. 
In the following, we describe this new method formally.
We start by defining the notion of a combination of two configurations.

Let $C = S_1~\dots~S_k$ and $C' = S'_1~\dots~S'_k$ be two configurations of sets of labels.
Let $1 \leq u \leq k$, and let $\sigma \colon \{ 1, \dots, k\} \rightarrow \{ 1, \dots, k\}$ be a bijective function, i.e., a permutation.
The \emph{combination of $C$ and $C'$ w.r.t.\ $u$ and $\sigma$} is defined as the configuration $C'' = S''_1~\dots~S''_k$ satisfying
\begin{itemize}[noitemsep]
	\item $S''_j = S_j \cap S'_{\sigma(j)}$ for each $j \in \{ 1, \dots, k\} \setminus \{ u \}$, and
	\item $S''_u = S_j \cup S'_{\sigma(j)}$.
\end{itemize}

Now consider the following process.
Start with $\edgeconstone_\Pi$, given as a collection of condensed configurations, where, as usual, disjunctions are considered as sets.
(If $\edgeconstone_\Pi$ is given as a collection of uncondensed configurations, it is straightforward to transform it into a list of condensed configurations by interpreting each configuration of labels as the configuration of singleton disjunctions/sets containing those labels.)
Now arbitrarily choose two configurations $C$ and $C'$ from the collection, combine them w.r.t.\ some arbitrarily chosen $u$ and $\sigma$, and add the obtained configuration to the maintained collection of configurations of sets (unless it is already present).
Repeat this, for all possible choices of $C$, $C'$, $u$, and $\sigma$ (including all newly obtained combined configurations as choices for $C$ and $C'$) until no new configuration (i.e., no configuration not already obtained in the collection) is produced anymore.
Then, remove from the obtained collection all \emph{non-maximal configurations}, i.e., all configurations $S_1~\dots~S_k$ for which there exists a configuration $S'_1~\dots~S'_k$ satisfying $S_j \subseteq S'_j$ for all $1 \leq j \leq k$ and $S_j \subsetneq S'_j$ for at least one $1 \leq j \leq k$.

Now~\cite[Theorem 4.1]{balliu2024towards} states that the obtained collection of configurations is identical to $\nodeconstone_{\re(\Pi)}$, and that the same procedure, when started from $\edgeconstone_\Pi$ instead of $\nodeconstone_\Pi$, yields $\edgeconstone_{\rere(\Pi)}$.
In other words, the described simple iterative process can be used to replace the traditional and more complex computation of $\nodeconstone_{\re(\Pi)}$, resp.\ $\edgeconstone_{\rere(\Pi)}$, when computing $\re(\Pi)$, resp.\ $\rere(\Pi)$ (while $\Sigma_\Pi$ and $\edgeconstone_{\re(\Pi)}$, resp.\ $\nodeconstone_{\rere(\Pi)}$, are computed in the same way in both the traditional and the new method).
We obtain the following observation (which follows directly from \cite{balliu2024towards}).
\begin{observation}\label{lem:combining}
	Let $\mathcal{C}$ be a collection of configurations of sets. Assume that $\mathcal{C}$ satisfies the following: for all pairs of configurations $C_1, C_2 \in \mathcal{C}$, all permutations $\sigma$ and all indices $u$, the combination of $C_1$ and $C_2$ w.r.t.\ $u$ and $\sigma$ is dominated by some configuration in $\mathcal{C}$. Then, $\mathcal{C}$ is maximal-complete.
\end{observation}

We call two sets $S, S'$ of labels incomparable if $S \nsubseteq S'$ and $S' \nsubseteq S$, and comparable otherwise.
Based on the notion of (in)comparable sets, the following observation (which is analogous to \cite[Observation 5.4]{balliu2024towards}) essentially shows that many choices for $C$, $C'$, $u$, and $\sigma$ do not need to be considered in the described process as they cannot produce configurations that will be contained in the final collection.
\begin{observation}\label{obs:no-union-comparable}
	Let $C = S_1~\dots~S_k$ and $C' = S'_1~\dots~S'_k$ be two configurations of sets of labels, $1 \leq u \leq k$ an index, and $\sigma \colon \{ 1, \dots, k\} \rightarrow \{ 1, \dots, k\}$ a permutation.
	Let $C'' = S''_1~\dots~S''_k$ be the configuration obtained by combining $C$ and $C''$ w.r.t.\ $u$ and $\sigma$.
	Then, if $S_u$ and $S'_u$ are comparable, $C''$ is dominated by $C$ or $C'$.
\end{observation}
\begin{proof}
	If $S_u$ and $S'_u$ are comparable, then $S_u \subseteq S'_u$ or $S'_u \subseteq S_u$.
	Now the definition of $C''$ ensures that, in the former case, $C''$ is dominated by $C'$ and, in the latter case, $C''$ is dominated by $C$.  
\end{proof}
From the definition of the above process and \Cref{obs:no-union-comparable}, it follows that the final collection obtained when the process terminates will not change if we modify the process so that combinations where $S_u$ and $S'_u$ are comparable are not considered.
Note that it is possible that some configurations that are produced in the unmodified process along the way are not produced in the modified process; however, all such configurations are configurations that will not be maximal, i.e., that are removed in the last step of the unmodified process.
In particular, all configurations that are maximal in the collection produced by the unmodified process will also be produced (and will be maximal) in the modified process.

We now define what it means for a problem to be the relaxation of another problem. Informally, for a problem $\Pi'$ to be a relaxation of $\re(\Pi)$ it must hold that there exists a $0$-round algorithm operating on black nodes that, given a solution for $\re(\Pi)$, it is able to solve $\Pi'$. Similarly, for a problem $\Pi'$ to be a relaxation of $\rere(\Pi)$ it must hold that there exists a $0$-round algorithm operating on white nodes that, given a solution for $\rere(\Pi)$, it is able to solve $\Pi'$. In this paper, we consider a stricter (i.e., less permissive) form of relaxation, defined as follows.
\begin{definition}[Relaxation of a problem]
	Let $\Pi = (\Sigma,\nodeconst,\edgeconst)$. A problem $\Pi' = (\Sigma',\nodeconst',\edgeconst')$ is a relaxation of $\re(\Pi)$ if the following holds:
	\begin{itemize}[noitemsep]
		\item For all configurations $C$ in the black constraint of $\re(\Pi)$, there exists a configuration $C' \in \edgeconst'$ that dominates $C$;
		\item $\Sigma'$ is the collection of sets that appear in at least one configuration contained in $\edgeconst'$;
		\item $\nodeconst'$ is the set of all configurations $S_1~\ldots~S_\Delta$ of sets contained in $\Sigma'$ such that there exists some tuple $(\ell_1, \ldots, \ell_\Delta) \in S_1 \times \ldots \times S_\Delta$ satisfying $\ell_1~\ldots~\ell_\Delta \in \nodeconst$.
	\end{itemize}
	Similarly, a problem $\Pi' = (\Sigma',\nodeconst',\edgeconst')$ is a relaxation of $\rere(\Pi)$ if the following holds:
	\begin{itemize}[noitemsep]
		\item For all configurations $C$ in the white constraint of $\rere(\Pi)$, there exists a configuration $C' \in \nodeconst'$ that dominates $C$;
		\item $\Sigma'$ is the collection of sets that appear in at least one configuration contained in $\nodeconst'$;
		\item $\edgeconst'$ is the set of all configurations $S_1~S_2~S_3$ of sets contained in $\Sigma'$ such that there exists some tuple $(\ell_1, \ell_2, \ell_3) \in S_1 \times S_2 \times S_3$ satisfying $\ell_1~\ell_2~\ell_3 \in \edgeconst$.
	\end{itemize}
\end{definition}

\begin{definition}[Right-closed]
	Let $\Pi=(\Sigma_\Pi, \nodeconst_\Pi, \edgeconst_\Pi)$ be an LCL problem. We say that a set $S = \{ \ell_1, \dots, \ell_k \} \subseteq \Sigma_\Pi$ is \emph{right-closed} if and only if, for each label $\ell_i\in S$, it holds that all successors of $\ell_i$ in the diagram (taken w.r.t.\ the constraint on which the universal quantifier is applied to) are also contained in $S$. More precisely, a set $S = \{ \ell_1, \dots, \ell_k \} \subseteq \Sigma_\Pi$ is called right-closed if $S = \gen{\ell_1, \dots, \ell_k}$.
\end{definition}

We observe the following statement, which was shown in \cite{balliurules}.

\begin{observation}[\hspace{1sp}\cite{balliurules}]\label{obs:rcs}
Let $\Pi=(\Sigma_\Pi, \nodeconst_\Pi, \edgeconst_\Pi)$ be an LCL problem. Consider an arbitrary collection of labels $L_1, \dots, L_k \in \Sigma_\Pi$.
If $\{ L_1, \dots, L_k  \} \in \Sigma_{\re(\Pi)}$, then the set $\{ L_1, \dots, L_k \}$ is right-closed w.r.t.\ $\edgeconst_{\Pi}$.
If $\{ L_1, \dots, L_k \} \in \Sigma_{\rere(\Pi)}$, then the set $\{ L_1, \dots, L_k \}$ is right-closed w.r.t.\ $\nodeconst_{\Pi}$.
\end{observation}

\begin{definition}[A configuration generated by a condensed one]
	Let $C=S_1,\ldots, S_k$ be a configuration of sets. We say that the configuration $C'=S'_1,\ldots, S'_k$ \emph{is generated by the condensed configuration} $C$ if there exists a permutation $\sigma \colon \{1, \dots, k\} \to \{1, \dots, k\}$ such that, for every $1\le i\le k$, it holds that $S'_i=\gen{c\mid c\in S_{\sigma(i)}}$.
\end{definition}

We are now ready to state the theorem that relates round elimination with the LOCAL model of distributed computing.
\begin{theorem}[\hspace{1sp}\cite{hideandseek}, rephrased]\label{thm:lifting}
	Let $\Pi_0 \rightarrow \Pi_1 \rightarrow \ldots \rightarrow \Pi_t$ be a sequence of problems.
	Assume that, for all $0 \le i < t$, and for some function $f$, the following holds:
	\begin{itemize}[noitemsep]
		\item There exists a problem $\Pi'_i$ that is a relaxation of $\re(\Pi_i)$;
		\item $\Pi_{i+1}$ is a relaxation of $\rere(\Pi'_i)$;
		\item The number of labels of $\Pi_i$, and the ones of $\Pi'_i$, are upper bounded by $f(\Delta)$.
	\end{itemize}
	Also, assume that $\Pi_t$ has at most $f(\Delta)$ labels and is not $0$-round solvable in the deterministic port numbering model, even when given a $\Delta$-edge coloring.\footnote{In \cite{hideandseek}, this theorem is stated in its full generality, and it talks about restricted assignments of port numbers. It is easy to see that a $\Delta$-edge coloring is indeed a restricted assignment of port numbers.} Then, $\Pi_0$ requires $\Omega(\min\{t, \log_\Delta n - \log_\Delta \log f(\Delta)\})$ rounds in the deterministic LOCAL model and $\Omega(\min\{t, \log_\Delta \log n - \log_\Delta \log f(\Delta)\})$ rounds in the randomized LOCAL model, even when given a $\Delta$-edge coloring.
\end{theorem}

\subsection{Roadmap}\label{ssec:roadmap}
Informally, in order to prove a lower bound for $\mathcal{P}_\Delta$, we introduce a family of problems $\Pi_{i,\Delta}$ (where $0 \leq i \leq \Delta - 2$) with the following properties.
\begin{itemize}
	\item $\mathcal{P}_\Delta$ is at least as hard as $\Pi_{0,\Delta}$.
	\item By applying round elimination to $\Pi_{i,\Delta}$, we obtain a problem that is at least as hard as $\Pi_{i+1,\Delta}$, for all $0 \leq i \le \Delta - 3$.
	\item $\Pi_{\Delta-2,\Delta}$ cannot be solved in $0$ rounds given a $\Delta$-edge coloring.
\end{itemize}
Our result will follow by applying \Cref{thm:lifting} to such a sequence of problems.

More in detail, in \Cref{ssec:re-def} we define the problems $\Pi_{i,\Delta}$, and in \Cref{ssec:diagram} we prove the strength relation between their labels. Then, in \Cref{ssec:re-first} we define a problem $\Pi'_{i,\Delta}$ and show that it is a relaxation of $\re(\Pi_{i,\Delta})$. In \Cref{ssec:re-diagram-2} we prove the strength relation between the labels of $\Pi'_{i,\Delta}$. In \Cref{ssec:re-second} we prove that $\rere(\Pi'_{i,\Delta})$ can be relaxed to $\Pi_{i+1,\Delta}$. Then, in \Cref{ssec:re-non-trivial} we show that $\Pi_{\Delta-2,\Delta}$ cannot be solved in $0$ rounds given a $\Delta$-edge coloring. In \Cref{ssec:re-start}, we show that $\mathcal{P}_\Delta$ is at least as hard as $\Pi_{0,\Delta}$. Finally, in \Cref{ssec:re-lift}, we combine the results of all previous sections to prove \Cref{thm:main-lb}.

\subsection{Problem definition}\label{ssec:re-def}
For each $0 \le i \le \Delta-2$, we define the problems $\Pi_{i,\Delta}$ in the black-white formalism. 
We remark that while the problem $\Pi_{i,\Delta}$ is defined also for $i= \Delta-2$, from \Cref{ssec:diagram} onwards we will assume $i \le \Delta-3$ for technical reasons. 

\paragraph{Labels.}
The label set of $\Pi_{i,\Delta}$ is defined as
\begin{align*}
	\Sigma_{i,\Delta} = 
	~ & \{\reEE{j} \mid 1 \le j \le \Delta - i -1 \} \\
	\cup ~ &\{\reE{j} \mid \Delta - i + 1 \le j \le \Delta \}\\
	\cup ~ &\{\reMM{j} \mid 1 \le j \le \Delta - i \}\\
	\cup ~ &\{\reMY{j}{y} \mid 1 \le j \le \Delta - i \text{ and } y \in \{0,1\}\}\\
	\cup ~ &\{\reXM{j}{x} \mid 2 \le j \le \Delta - i -1 \text{ and } x \in \{0,1\}\}\\
	\cup ~ &\{\reXY{j}{x}{y} \mid 2 \le j \le \Delta - i \text{ and }x,y \in \{0,1\}\}\\
	\cup ~ &\{\reQ{j}\mid \Delta - i + 1 \le j \le \Delta \}.
\end{align*}
We call the subscript of a label its \emph{color}. For example, label $\reMM{j}$ has color $j$. When considering some fixed $i$, we say that a color $j$ is \emph{gone} if $j > \Delta - i$, otherwise it is \emph{present}. Moreover, color $\Delta - i$ is called \emph{special}. Throughout the section, when writing a configuration $L_1 \ldots L_\Delta$ we may implicitly assume that label $L_i$ is of color $i$.

\paragraph{White constraint.}
We now define the white constraint $\nodeconst_{i,\Delta}$ of problem $\Pi_{i,\Delta}$.
For an example, see the paragraph after the definition.
$\nodeconst_{i,\Delta}$ contains node configurations of three different kinds.

First, for all integers $a,b$ satisfying $1 \le a<b \le \Delta - i + 1$ and $b \neq \Delta - i$, and for all vectors $(y_a, \dots, y_{b-1}) \in \{0, 1\}^{b - a}$,
$\nodeconst_{i,\Delta}$ contains the configuration $L_1 \ldots L_\Delta$ satisfying the following properties.
\begin{itemize}[noitemsep]
	\item For all $j$ such that $1 \le j < a$, $L_j = \reMM{j}$,
	\item $L_a = \reMY{a}{y_a}$,
	\item for all $j$ such that $a < j < b$, $L_j = \reXY{j}{y_{j-1}\,}{y_{j}}$,
	\item if $b \le \Delta - i$, $L_b = \reXM{b}{y_{b-1}\,}$,
	\item for all $j$ such that $b < j \le \Delta - i$, $L_j = \reMM{j}$, and
	\item for all $j$ such that $\Delta - i + 1 \le j \le \Delta$, $L_j = \reQ{j}$.
\end{itemize}
Second, for all integers $a$ satisfying $1 \le a \le \Delta - i - 1$, $\nodeconst_{i,\Delta}$ contains the configuration $L_1 \ldots L_\Delta $ satisfying the following properties.
\begin{itemize}[noitemsep]
	\item $L_a = \reEE{a}$,
	\item for all $j$ such that $1 \le j \le \Delta - i$ and $j \neq a$, $L_j = \reMM{j}$, and
	\item for all $j$ such that $\Delta - i + 1 \le j \le \Delta$, $L_j = \reQ{j}$.
\end{itemize}
Third, for all integers $a$ satisfying $\Delta - i + 1 \le a \le \Delta$, $\nodeconst_{i,\Delta}$ contains the configuration $L_1 \ldots L_\Delta$ satisfying the following properties.
\begin{itemize}[noitemsep]
	\item $L_a = \reE{a}$,
	\item for all $j$ such that  $1 \le j \le \Delta - i$, $L_j = \reMM{j}$, and
	\item for all $j$ such that  $\Delta - i + 1 \le j \le \Delta$ and $j \neq a$, $L_j = \reQ{j}$.
\end{itemize}

\paragraph{Example for the white constraint.}
To make the above definition more accessible, we now informally show how to obtain the white constraint in a concrete way, illustrated with an example for specific values of $\Delta$ and $i$.
Let $\Delta = 7$ and $i = 2$.
We obtain $\nodeconst_{2, 7}$ by setting $\nodeconst_{2, 7} = \emptyset$ and then performing the following procedure for each bit string $B$ of $\Delta-i$ bits (which we illustrate for the case $B = 01100$).

Add to $\nodeconst_{2, 7}$ the configuration $C$ of the first kind described above where we set $a := 1$, $b := \Delta - i + 1$, and $(y_a, \dots, y_{b-1}) := B$, i.e., we add 
\[
	C = \reMY{1}{0} ~\reXY{2}{0}{1} ~\reXY{3}{1}{1} ~\reXY{4}{1}{0} ~\reXY{5}{0}{0} ~\reQ{6} ~\reQ{7}.
\]
Starting from $C$, we now create additional configurations by \emph{canceling} some bits, which we will then all add to $\nodeconst_{2,7}$ (if not already present). We can cancel a prefix of the bit string $B$, a suffix, or both at the same time. However, at least one bit must remain. Moreover, if we cancel a suffix, the canceled suffix cannot be only the last bit. 
Each time we obtain a new bit string, we transform it into a configuration analogously to how we obtained $C$ from $B$, and add it to $\nodeconst_{2,7}$.
In this way, we add the following configurations.
\begin{itemize}
	\item Prefix canceling:
	\begin{align*}	
		\reMM{1} ~\reMY{2}{1} ~\reXY{3}{1}{1} ~\reXY{4}{1}{0} ~\reXY{5}{0}{0} ~\reQ{6} ~\reQ{7} \\
		\reMM{1} ~\reMM{2} ~\reMY{3}{1} ~\reXY{4}{1}{0} ~\reXY{5}{0}{0} ~\reQ{6} ~\reQ{7} \\
		\reMM{1} ~\reMM{2} ~\reMM{3} ~\reMY{4}{0} ~\reXY{5}{0}{0} ~\reQ{6} ~\reQ{7} \\
		\reMM{1} ~\reMM{2} ~\reMM{3} ~\reMM{4} ~\reMY{5}{0} ~\reQ{6} ~\reQ{7} 
	\end{align*}
	\item Suffix canceling (note that the rules do not allow a canceling of only the fifth bit in $B$):
	\begin{align*}
		\reMY{1}{0} ~\reXY{2}{0}{1} ~\reXY{3}{1}{1} ~\reXM{4}{1} ~\reMM{5} ~\reQ{6} ~\reQ{7}\\		\reMY{1}{0} ~\reXY{2}{0}{1} ~\reXM{3}{1} ~\reMM{4} ~\reMM{5} ~\reQ{6} ~\reQ{7}\\
		\reMY{1}{0} ~\reXM{2}{0} ~\reMM{3} ~\reMM{4} ~\reMM{5} ~\reQ{6} ~\reQ{7}
	\end{align*}
	\item Prefix and suffix canceling (again, the rules do not allow a suffix canceling of only the fifth bit in $B$):
	\begin{align*}	
		\reMM{1} ~\reMY{2}{1} ~\reXY{3}{1}{1} ~\reXM{4}{1} ~\reMM{5} ~\reQ{6} ~\reQ{7} \\
		\reMM{1} ~\reMY{2}{1} ~\reXM{3}{1} ~\reMM{4} ~\reMM{5} ~\reQ{6} ~\reQ{7}\\
		\reMM{1} ~\reMM{2} ~\reMY{3}{1}  ~\reXM{4}{1} ~\reMM{5} ~\reQ{6} ~\reQ{7}
	\end{align*}
\end{itemize}
Finally, we add a small number of additional configurations to $\nodeconst_{2,7}$ that we obtain from canceling the entire bit string $B$.
However, when we cancel $B$ entirely, we are required to select in the respectively obtained configuration a label at some position $j \neq \Delta - i$ and replace this label by $\reEE{j}$ if $j < \Delta-i$, and by $\reE{j}$ otherwise.
Hence, we obtain the following additional configurations that we add to $\nodeconst_{2,7}$.
\begin{align*}	
	\reEE{1} ~ \reMM{2} ~ \reMM{3} ~ \reMM{4} ~\reMM{5} ~\reQ{6} ~\reQ{7} \\
	\reMM{1} ~ \reEE{2} ~ \reMM{3} ~ \reMM{4} ~\reMM{5} ~\reQ{6} ~\reQ{7} \\
	\reMM{1} ~ \reMM{2} ~ \reEE{3} ~ \reMM{4} ~\reMM{5} ~\reQ{6} ~\reQ{7} \\
	\reMM{1} ~ \reMM{2} ~ \reMM{3} ~ \reEE{4} ~\reMM{5} ~\reQ{6} ~\reQ{7} \\
	\reMM{1} ~ \reMM{2} ~ \reMM{3} ~ \reMM{4} ~\reMM{5} ~\reE{6} ~\reQ{7} \\
	\reMM{1} ~ \reMM{2} ~ \reMM{3} ~ \reMM{4} ~\reMM{5} ~\reQ{6} ~\reE{7}
\end{align*}

\paragraph{Black constraint.}
We now define the black constraint $\edgeconst_{i,\Delta}$ of problem $\Pi_{i,\Delta}$. Let $j$ be an integer satisfying $\Delta-i+1 \le j \le \Delta$, i.e., a color that is gone. The allowed configurations are exactly those that are described by the set of condensed configurations containing the following: 
\begin{itemize}
	\item $[\reE{j},\reQ{j}]^2 ~~~ [\reQ{j}]$
\end{itemize}
Then, the allowed configurations are all the ones described by the set of condensed configurations containing the following:
\begin{itemize}
	\item $[\reEE{1},\reMM{1},\reMY{1}{0},\reMY{1}{1}] ~~~ [\reMM{1}] ~~~ [\reMM{1},\reMY{1}{0}]$
	\item $[\reMM{1},\reMY{1}{0}]^2 ~~~ [\reMM{1},\reMY{1}{1}]$
\end{itemize}
Let $j$ be an integer satisfying $2 \le j \le \Delta - i - 1$. The allowed configurations are all the ones described by the set of condensed configurations containing the following:
\begin{itemize}[noitemsep]
	\item $[\reMM{j}] ~~~ [\reEE{j},\reMM{j},\reMY{j}{0},\reMY{j}{1},\reXM{j}{0},\reXY{j}{0}{0},\reXY{j}{0}{1},\reXM{j}{1},\reXY{j}{1}{0},\reXY{j}{1}{1}]^2$
	\item $[\reMM{j},\reMY{j}{0},\reXY{j}{0}{0},\reXM{j}{1},\reXY{j}{1}{0},\reXY{j}{1}{1}] ~~~ [\reMM{j},\reXY{j}{0}{0}]^2$
	\item $[\reMM{j},\reMY{j}{1},\reXY{j}{0}{1},\reXM{j}{1},\reXY{j}{1}{0},\reXY{j}{1}{1}] ~~~ [\reMM{j},\reXY{j}{0}{0}] ~~~ [\reMM{j},\reXY{j}{0}{1}]$
	\item $[\reMM{j},\reMY{j}{0},\reXY{j}{0}{0},\reXM{j}{1},\reXY{j}{1}{0},\reXY{j}{1}{1}] ~~~ [\reMM{j},\reXY{j}{0}{1}]^2$
	\item $[\reMM{j},\reMY{j}{1},\reXM{j}{0},\reXY{j}{0}{0},\reXY{j}{0}{1},\reXY{j}{1}{1}] ~~~ [\reMM{j},\reXY{j}{0}{0}] ~~~ [\reMM{j},\reXY{j}{1}{0}]$
	\item $[\reMM{j},\reMY{j}{0},\reXM{j}{0},\reXY{j}{0}{0},\reXY{j}{0}{1},\reXY{j}{1}{0}] ~~~ [\reMM{j},\reXY{j}{0}{0}] ~~~ [\reMM{j},\reXY{j}{1}{1}]$
	\item $[\reMM{j},\reMY{j}{0},\reXM{j}{0},\reXY{j}{0}{0},\reXY{j}{0}{1},\reXY{j}{1}{0}] ~~~ [\reMM{j},\reXY{j}{0}{1}] ~~~ [\reMM{j},\reXY{j}{1}{0}]$
	\item $[\reMM{j},\reMY{j}{1},\reXM{j}{0},\reXY{j}{0}{0},\reXY{j}{0}{1},\reXY{j}{1}{1}] ~~~ [\reMM{j},\reXY{j}{0}{1}] ~~~ [\reMM{j},\reXY{j}{1}{1}]$
	\item $[\reMM{j},\reMY{j}{1},\reXY{j}{0}{1},\reXM{j}{1},\reXY{j}{1}{0},\reXY{j}{1}{1}] ~~~ [\reMM{j},\reXY{j}{1}{0}]^2$
	\item $[\reMM{j},\reMY{j}{0},\reXY{j}{0}{0},\reXM{j}{1},\reXY{j}{1}{0},\reXY{j}{1}{1}] ~~~ [\reMM{j},\reXY{j}{1}{0}] ~~~ [\reMM{j},\reXY{j}{1}{1}]$
	\item $[\reMM{j},\reMY{j}{1},\reXY{j}{0}{1},\reXM{j}{1},\reXY{j}{1}{0},\reXY{j}{1}{1}] ~~~ [\reMM{j},\reXY{j}{1}{1}]^2$
	\item $[\reMM{j},\reMY{j}{0},\reXY{j}{0}{0},\reXY{j}{1}{0}] ~~~  [\reMM{j},\reXY{j}{0}{0}] ~~~ [\reMM{j},\reXY{j}{0}{0},\reXY{j}{1}{1}]$
	\item $[\reMM{j},\reMY{j}{1},\reXY{j}{0}{1},\reXY{j}{1}{1}] ~~~ [\reMM{j},\reXY{j}{0}{0}] ~~~ [\reMM{j},\reXY{j}{0}{1},\reXY{j}{1}{0}]$
	\item $[\reMM{j},\reMY{j}{0},\reXY{j}{0}{0},\reXY{j}{1}{0}] ~~~ [\reMM{j},\reXY{j}{0}{1}] ~~~ [\reMM{j},\reXY{j}{0}{1},\reXY{j}{1}{0}]$
	\item $[\reMM{j},\reMY{j}{1},\reXY{j}{0}{1},\reXY{j}{1}{1}] ~~~ [\reMM{j},\reXY{j}{0}{0},\reXY{j}{1}{1}] ~~~ [\reMM{j},\reXY{j}{0}{1}]$
	\item $[\reMM{j},\reMY{j}{0},\reXY{j}{0}{0},\reXY{j}{1}{0}] ~~~ [\reMM{j},\reMY{j}{1},\reXY{j}{0}{1},\reXY{j}{1}{1}] ~~~ [\reMM{j},\reXY{j}{1}{0}]$
	\item $[\reMM{j},\reXY{j}{1}{1}] ~~~ [\reMM{j},\reMY{j}{0},\reXY{j}{0}{0},\reXY{j}{1}{0}]^2$
	\item $[\reMM{j},\reXY{j}{1}{0}] ~~~ [\reMM{j},\reXY{j}{0}{1},\reXY{j}{1}{0}]^2$
	\item $[\reMM{j},\reXY{j}{1}{0}] ~~~ [\reMM{j},\reXY{j}{0}{0},\reXY{j}{1}{1}]^2$
	\item $[\reMM{j},\reXY{j}{0}{0},\reXY{j}{1}{1}] ~~~ [\reMM{j},\reXY{j}{0}{1},\reXY{j}{1}{0}] ~~~ [\reMM{j},\reXY{j}{1}{1}]$
	\item $[\reMM{j},\reXY{j}{1}{1}] ~~~ [\reMM{j},\reMY{j}{1},\reXY{j}{0}{1},\reXY{j}{1}{1}]^2$
	\item $[\reMM{j},\reXM{j}{1},\reXY{j}{1}{0},\reXY{j}{1}{1}] ~~~ [\reMM{j},\reXM{j}{0},\reXY{j}{0}{0},\reXY{j}{0}{1}]^2$
	\item $[\reMM{j},\reXM{j}{1},\reXY{j}{1}{0},\reXY{j}{1}{1}]^3$
\end{itemize}
For the case $j = \Delta - i$, i.e., for the special color, the list of allowed configurations are all the ones described by the set of condensed configurations given by the following process:
\begin{itemize}[noitemsep]
	\item Start from the condensed configurations given above for the case $2 \le j \le \Delta - i - 1$.
	\item From each disjunction, remove labels $\reXM{j}{0}$, $\reXM{j}{1}$, and $\reEE{j}$ (i.e., labels that for color $\Delta-i$ do not exist).
\end{itemize}

\paragraph{Alternative definition of the black constraint.}
We provided the allowed configurations for the case $2 \le j \le \Delta - i - 1$ as a long list. We now provide an equivalent but compact description of such configurations, that will be useful later.
Such configurations are based on the allowed configurations of the iterated GHZ problem. Each configuration is composed of 3 labels with subscript $j$, and for the moment we only consider labels containing two bits, i.e., $\reXY{j}{0}{0}$,  $\reXY{j}{0}{1}$, $\reXY{j}{1}{0}$, and $\reXY{j}{1}{1}$. Later we will provide additional configurations as a function of the ones that we provide now. The first bit of each label is called input bit, and the second is called output bit. If the number of input bits equal to $1$ is even, is required that the XOR of the output bits is equal to the OR of the input bits. Instead, if the number of input bits equal to $1$ is odd, there is no constraint on the output bits. The configurations satisfying these requirements, which we call \emph{bit configurations}, are the following:
\begin{itemize}[noitemsep]
	\item $\reXY{j}{0}{0} ~~~ \reXY{j}{0}{0} ~~~ \reXY{j}{0}{0}$
	\item $\reXY{j}{0}{0} ~~~ \reXY{j}{0}{1} ~~~ \reXY{j}{0}{1}$
	\item $\reXY{j}{0}{0} ~~~ \reXY{j}{0}{0} ~~~ \reXY{j}{1}{0}$
	\item $\reXY{j}{0}{0} ~~~ \reXY{j}{0}{0} ~~~ \reXY{j}{1}{1}$
	\item $\reXY{j}{0}{0} ~~~ \reXY{j}{0}{1} ~~~ \reXY{j}{1}{0}$
	\item $\reXY{j}{0}{0} ~~~ \reXY{j}{0}{1} ~~~ \reXY{j}{1}{1}$
	\item $\reXY{j}{0}{1} ~~~ \reXY{j}{0}{1} ~~~ \reXY{j}{1}{0}$
	\item $\reXY{j}{0}{1} ~~~ \reXY{j}{0}{1} ~~~ \reXY{j}{1}{1}$
	\item $\reXY{j}{0}{0} ~~~ \reXY{j}{1}{0} ~~~ \reXY{j}{1}{1}$
	\item $\reXY{j}{0}{1} ~~~ \reXY{j}{1}{0} ~~~ \reXY{j}{1}{0}$
	\item $\reXY{j}{0}{1} ~~~ \reXY{j}{1}{1} ~~~ \reXY{j}{1}{1}$
	\item $\reXY{j}{1}{0} ~~~ \reXY{j}{1}{0} ~~~ \reXY{j}{1}{0}$
	\item $\reXY{j}{1}{0} ~~~ \reXY{j}{1}{0} ~~~ \reXY{j}{1}{1}$
	\item $\reXY{j}{1}{0} ~~~ \reXY{j}{1}{1} ~~~ \reXY{j}{1}{1}$
	\item $\reXY{j}{1}{1} ~~~ \reXY{j}{1}{1} ~~~ \reXY{j}{1}{1}$
\end{itemize}
We now compute the set of maximal configurations satisfying the universal quantifier w.r.t.\ the above set of configurations. 
\begin{lemma}\label{lem:onlybits-maximized}
	Let $S_j$ denote the set of maximal configurations satisfying the universal quantifier w.r.t.\ the set of bit configurations of color $j$.
	Then $S_j$ contains precisely the following configurations.
	\begin{itemize}[noitemsep]
		\item $\{\reXY{j}{0}{0},\reXY{j}{1}{0},\reXY{j}{1}{1}\} ~~~ \{\reXY{j}{0}{0}\}^2$
		\item $\{\reXY{j}{0}{1},\reXY{j}{1}{0},\reXY{j}{1}{1}\} ~~~ \{\reXY{j}{0}{0}\} ~~~ \{\reXY{j}{0}{1}\}$
		\item $\{\reXY{j}{0}{0},\reXY{j}{1}{0},\reXY{j}{1}{1}\} ~~~ \{\reXY{j}{0}{1}\}^2$
		\item $\{\reXY{j}{0}{0},\reXY{j}{0}{1},\reXY{j}{1}{1}\} ~~~ \{\reXY{j}{0}{0}\} ~~~ \{\reXY{j}{1}{0}\}$
		\item $\{\reXY{j}{0}{0},\reXY{j}{0}{1},\reXY{j}{1}{0}\} ~~~ \{\reXY{j}{0}{0}\} ~~~ \{\reXY{j}{1}{1}\}$
		\item $\{\reXY{j}{0}{0},\reXY{j}{0}{1},\reXY{j}{1}{0}\} ~~~ \{\reXY{j}{0}{1}\} ~~~ \{\reXY{j}{1}{0}\}$
		\item $\{\reXY{j}{0}{0},\reXY{j}{0}{1},\reXY{j}{1}{1}\} ~~~ \{\reXY{j}{0}{1}\} ~~~ \{\reXY{j}{1}{1}\}$
		\item $\{\reXY{j}{0}{1},\reXY{j}{1}{0},\reXY{j}{1}{1}\} ~~~ \{\reXY{j}{1}{0}\}^2$
		\item $\{\reXY{j}{0}{0},\reXY{j}{1}{0},\reXY{j}{1}{1}\} ~~~ \{\reXY{j}{1}{0}\} ~~~ \{\reXY{j}{1}{1}\}$
		\item $\{\reXY{j}{0}{1},\reXY{j}{1}{0},\reXY{j}{1}{1}\} ~~~ \{\reXY{j}{1}{1}\}^2$
		\item $\{\reXY{j}{0}{0},\reXY{j}{1}{0}\} ~~~ \{\reXY{j}{0}{0}\} ~~~ \{\reXY{j}{0}{0},\reXY{j}{1}{1}\}$
		\item $\{\reXY{j}{0}{1},\reXY{j}{1}{1}\} ~~~ \{\reXY{j}{0}{0}\} ~~~ \{\reXY{j}{0}{1},\reXY{j}{1}{0}\}$
		\item $\{\reXY{j}{0}{0},\reXY{j}{1}{0}\} ~~~ \{\reXY{j}{0}{1}\} ~~~ \{\reXY{j}{0}{1},\reXY{j}{1}{0}\}$
		\item $\{\reXY{j}{0}{1},\reXY{j}{1}{1}\} ~~~ \{\reXY{j}{0}{0},\reXY{j}{1}{1}\} ~~~ \{\reXY{j}{0}{1}\}$
		\item $\{\reXY{j}{0}{0},\reXY{j}{1}{0}\} ~~~ \{\reXY{j}{0}{1},\reXY{j}{1}{1}\} ~~~ \{\reXY{j}{1}{0}\}$
		\item $\{\reXY{j}{1}{1}\} ~~~ \{\reXY{j}{0}{0},\reXY{j}{1}{0}\}^2$
		\item $\{\reXY{j}{1}{0}\} ~~~ \{\reXY{j}{0}{1},\reXY{j}{1}{0}\}^2$
		\item $\{\reXY{j}{1}{0}\} ~~~ \{\reXY{j}{0}{0},\reXY{j}{1}{1}\}^2$
		\item $\{\reXY{j}{0}{0},\reXY{j}{1}{1}\} ~~~ \{\reXY{j}{0}{1},\reXY{j}{1}{0}\} ~~~ \{\reXY{j}{1}{1}\}$
		\item $\{\reXY{j}{1}{1}\} ~~~ \{\reXY{j}{0}{1},\reXY{j}{1}{1}\}^2$
		\item $\{\reXY{j}{1}{0},\reXY{j}{1}{1}\} ~~~ \{\reXY{j}{0}{0},\reXY{j}{0}{1}\}^2$
		\item $\{\reXY{j}{1}{0},\reXY{j}{1}{1}\}^3$
	\end{itemize}
\end{lemma}
\begin{proof}
	By inspecting each configuration of $S_j$, we can observe that the universal quantifier is satisfied on all of them.
	It is straightforward to verify that none of the displayed configurations dominates any of the other displayed configurations. Hence, it suffices to show that $S_j$ contains all the maximal configurations satisfying the universal quantifier.
	Assume for a contradiction that there exists a maximal configuration $C$ satisfying the universal quantifier that is not dominated by any configuration present in $S_j$.
	
	We start by proving that $C$ cannot contain two sets of size $1$, let them be $C_1$ and $C_2$. As it is straightforward to check, the set $S_j$ contains all  configurations satisfying the universal quantifier composed of two sets of size $1$ and one set of size $3$, and that  all  configurations satisfying the universal quantifier composed of two sets of size $1$ and one set of size at most $2$ are dominated by configurations in $S_j$. Hence, for the configuration $C$ to not be in $S_j$, it must contain a set of size $4$, let it be $C_3$.
	Consider the input bits $i_1$ and $i_2$ of the labels in the sets $C_1$ and $C_2$. If the number of input bits in $(i_1,i_2)$ equal to $1$ is even, then $C_3$ cannot contain both $\reXY{j}{0}{0}$ and $\reXY{j}{0}{1}$. If the number of input bits in $(i_1,i_2)$ equal to $1$ is odd, then $C_3$ cannot contain both $\reXY{j}{1}{0}$ and $\reXY{j}{1}{1}$. Hence, $C_3$ has size at most $3$.
	
	Since a set cannot be empty, and hence it must have size at least $1$, the above also implies that $C$ cannot contain sets of size $4$ (even if the other sets have size strictly larger than $1$).
	
	Moreover, consider a configuration $C$ containing two sets of size $1$ ($C_1$ and $C_2$) and one set of size $3$ ($C_3$). By adding any element $L$ to some set of size $1$, w.l.o.g.\ to $C_1$, the obtained configuration does not satisfy the universal quantifier anymore, for the following reasons.
	\begin{itemize}
		\item Consider the case in which $L$ has an input bit that is the same as $i_1$, and hence the opposite output bit. Since $C_3$ has size $3$, from $C$ we can always pick a configuration with an even number of input bits set to $1$. If we now replace the element picked from $C_1$ with $L$, we obtain an invalid configuration.
		\item If $L$ has an input bit that is different from $i_1$, then we can pick two configurations that only differ in the element picked from $C_3$ and that have an even number of input bits set to $1$, but that have different parity in the amount of output bits set to $1$, and hence one of the two is invalid.
	\end{itemize}
	Hence, if the number of sets of size $1$ is strictly less than two, it is not possible to have sets of size $3$.
	
	We now consider the case in which $C$ contains exactly one set of size $1$, let it be $C_1$. The other sets, $C_2$ and $C_3$, must be of size $2$. Since $C$ is maximal, it is not possible to add an element to $C_1$,  and in particular the label with the same input and opposite output bit of the label in $C_1$, without losing the property that the configuration satisfies the universal quantifier. This implies that it must be possible to pick, from $C$, a configuration with an even number of input bits set to $1$.
	Let such a configuration be $P_1 ~ P_2 ~ P_3$, and let the non-picked labels be $N_2 \in C_2$ and $N_3 \in C_3$.
	It cannot be that $P_2$ and $N_2$ (resp.\ $P_3$ and $N_3$) have the same input and opposite output bit, as replacing $P_2$ with $N_2$ (resp.\ $P_3$ with $N_3$) in $C$ would result in a configuration that is not allowed (since the amount of input bits equal to $1$ stays even, but the parity of the output changes). Moreover, the configuration $P_1 ~ N_2 ~ N_3$ still needs to be allowed. It is straightforward to verify that all configurations satisfying these requirements are listed.
	
	Finally, let us consider the case in which $C$ contains only sets of size exactly $2$. We prove that each set must satisfy that all its elements have the same input bit, and that the number of sets whose elements have input bit $1$ must be odd. Given such a proof, we get a contradiction to our initial assumption by observing that there are only two configurations satisfying these requirements, and they are both listed.

	Suppose that all configurations that can be picked from $C$ have an odd number of input bits set to $1$. This implies that all sets have elements with the same input bit, since otherwise we could replace a picked element with a non-picked one and change the parity.
	Now, suppose for a contradiction that we can pick a configuration with an even number of input bits set to $1$. If any set contains a non-picked element with the same input bit as the picked one but the opposite output bit, we would get a contradiction. Hence, each non-picked element must have opposite input bit from the picked one of the same set. Hence, it must be possible to pick  $\reXY{j}{0}{0}^3$, or $\reXY{j}{0}{0} ~ \reXY{j}{0}{1}^2$. In the former case, we must be able to pick $\reXY{j}{0}{0} ~ \reXY{j}{1}{0}^2$ or $\reXY{j}{0}{0} ~ \reXY{j}{1}{1}^2$, which are both not allowed. In the latter case, 
	we can pick $\reXY{j}{0}{0} ~ \reXY{j}{1}{0}^2$, $\reXY{j}{0}{0} ~ \reXY{j}{1}{1}^2$, or $\reXY{j}{1}{0} ~ \reXY{j}{0}{1} ~ \reXY{j}{1}{1}$, which are all not allowed. This completes our proof by contradiction.
\end{proof}
We now modify the set $S_j$ defined in \Cref{lem:onlybits-maximized} by applying the following rules.
\begin{itemize}[noitemsep]
	\item If a set contains both $\reXY{j}{0}{0}$ and $\reXY{j}{0}{1}$, add $\reXM{j}{0}$.
	\item If a set contains both $\reXY{j}{1}{0}$ and $\reXY{j}{1}{1}$, add $\reXM{j}{1}$.
	\item If a set contains both $\reXY{j}{0}{0}$ and $\reXY{j}{1}{0}$, add $\reMY{j}{0}$.
	\item If a set contains both $\reXY{j}{0}{1}$ and $\reXY{j}{1}{1}$, add $\reMY{j}{1}$.
	\item Always add $\reMM{j}$.
\end{itemize}
Finally, we add the following configuration $C^*$: 
\[
\{\reMM{j}\} ~~~ \{\reEE{j},\reMM{j},\reMY{j}{0},\reMY{j}{1},\reXM{j}{0},\reXY{j}{0}{0},\reXY{j}{0}{1},\reXM{j}{1},\reXY{j}{1}{0},\reXY{j}{1}{1}\}^2.
\]
Let the result be $S'_j$. Observe that, for all $2 \le j \le \Delta - i -1$, $S'_j$ is exactly equal to the set of condensed configurations listed when defining $\edgeconst_{i,\Delta}$.
We observe the following (which will be useful later).
\begin{observation}\label{obs:already-maximized}
	For all $2 \le j \le \Delta - i -1$, the set $S'_j$ is maximal-complete.
\end{observation}
\begin{proof}
	We prove that, if we take two arbitrary configurations $C^1$ and $C^2$ in $S'_j$, and we combine them w.r.t.\ an arbitrary $u$ and an arbitrary $\sigma$, we obtain a configuration $C = C_1 ~ C_2 ~ C_3$ that is dominated by some other configuration in $S'_j$. By \Cref{lem:combining} this implies that $S'_j$ is maximal-complete.
	
	We start by observing that, by \Cref{obs:no-union-comparable}, we do not need to consider the case in which $C^1 = C^*$ or $C^2 = C^*$.
	
	If an intersection results in ${\reMM{j}}$, then $C$ is dominated by the configuration $C^*$.
	Otherwise, let $C'^1$ and $C'^2$ be the configurations obtained by starting from $C^1$ and $C^2$ and removing the labels $\reXM{j}{0}$, $\reXM{j}{1}$, $\reMY{j}{0}$, $\reMY{j}{1}$, and $\reMM{j}$ from all the sets of $C^1$ and $C^2$. We obtain that $C'^1$ and $C'^2$ are configurations listed in \Cref{lem:onlybits-maximized}.
	Let $C'$ be the configuration obtained by combining $C'^1$ and $C'^2$ w.r.t.\ $u$ and $\sigma$, and let $D'$ be the configuration dominating $C'$ among the configurations listed in \Cref{lem:onlybits-maximized} (which is guaranteed to exist by \Cref{lem:onlybits-maximized}).
	In each set of $D'$ add the labels $\reXM{j}{0}$, $\reXM{j}{1}$, $\reMY{j}{0}$, $\reMY{j}{1}$, $\reMM{j}$ by using the same rules applied when defining $S'_j$ as a function of $S_j$, obtaining the configuration $D = D_1 ~ D_2 ~ D_3$.
	Clearly, the configuration $D$ is present in $S'_j$. Moreover, the following observations show that $D$ dominates $C$.
	\begin{itemize}
		\item The label $\reMM{j}$ is present in all sets of both $C$ and $D$, and hence does not affect domination.
		\item If some label $L \in \{\reXY{j}{0}{0}, \reXY{j}{0}{1}, \reXY{j}{1}{0}, \reXY{j}{1}{1}\}$ is present in $C_i$, for some $i \in \{1,2,3\}$, then $L$ is present also in $D_i$.
		\item If $C_u$ contains a label $L$ in $\{\reXM{j}{0}, \reXM{j}{1}, \reMY{j}{0}, \reMY{j}{1}\}$, then $L \in D_u$. To see this, consider w.l.o.g.\ $L = \reXM{j}{0}$. We get that $\reXY{j}{0}{0} \in C^1_u$ or $\reXY{j}{0}{0} \in C^2_{\sigma(u)}$, and $\reXY{j}{0}{1} \in C^1_u$ or $\reXY{j}{0}{1} \in C^2_{\sigma(u)}$. Hence, $\{\reXY{j}{0}{0},\reXY{j}{0}{1}\} \subseteq D_u$, and hence $\reXM{j}{0} \in D_u$.
		\item Let $i \in \{1,2,3\}$ and $i \neq u$. If $C_i$ contains a label $L$ in $\{\reXM{j}{0}, \reXM{j}{1}, \reMY{j}{0}, \reMY{j}{1}\}$, then $L \in D_i$. To see this, consider w.l.o.g.\ $L = \reXM{j}{0}$. We get that $\reXY{j}{0}{0} \in C^1_i$, $\reXY{j}{0}{0} \in C^2_{\sigma(i)}$, $\reXY{j}{0}{1} \in C^1_i$, and $\reXY{j}{0}{1} \in C^2_{\sigma(i)}$. Hence, $\{\reXY{j}{0}{0},\reXY{j}{0}{1}\} \subseteq D_i$, and hence $\reXM{j}{0} \in D_i$.
	\end{itemize}
	Observe that this covers all cases, since $\reEE{j}$ only occurs in $C^*$.
\end{proof}

\subsection{The diagram of \texorpdfstring{$\Pi_{i,\Delta}$}{Pi}}\label{ssec:diagram}
Throughout the rest of \Cref{sec:lower-bound}, we assume that $i \le \Delta - 3$.
We now provide the diagram of $\Pi_{i,\Delta}$, that is, we prove the strength relation between the labels of $\Pi_{i,\Delta}$ w.r.t.\ its black constraint.
In \Cref{sec:lb-proofs-omitted-first-step} we will first observe that, a necessary condition for a label to be at least as strong as another label, is for the two labels to be of the same color. Then, we will prove \Cref{lem:diagram-gone,lem:diagram-first,lem:diagram-special,lem:diagram-present} (which are just a long case analysis).

\begin{figure}[ht!]
	\centering
	\includegraphics[scale=0.3]{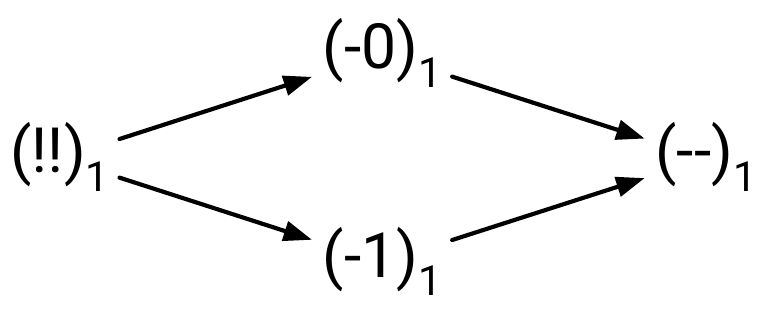}
	\caption{The strength relation of the first color.}
	\label{fig:diagram-first}
\end{figure}
\begin{figure}[ht!]
	\centering
	\includegraphics[scale=0.3]{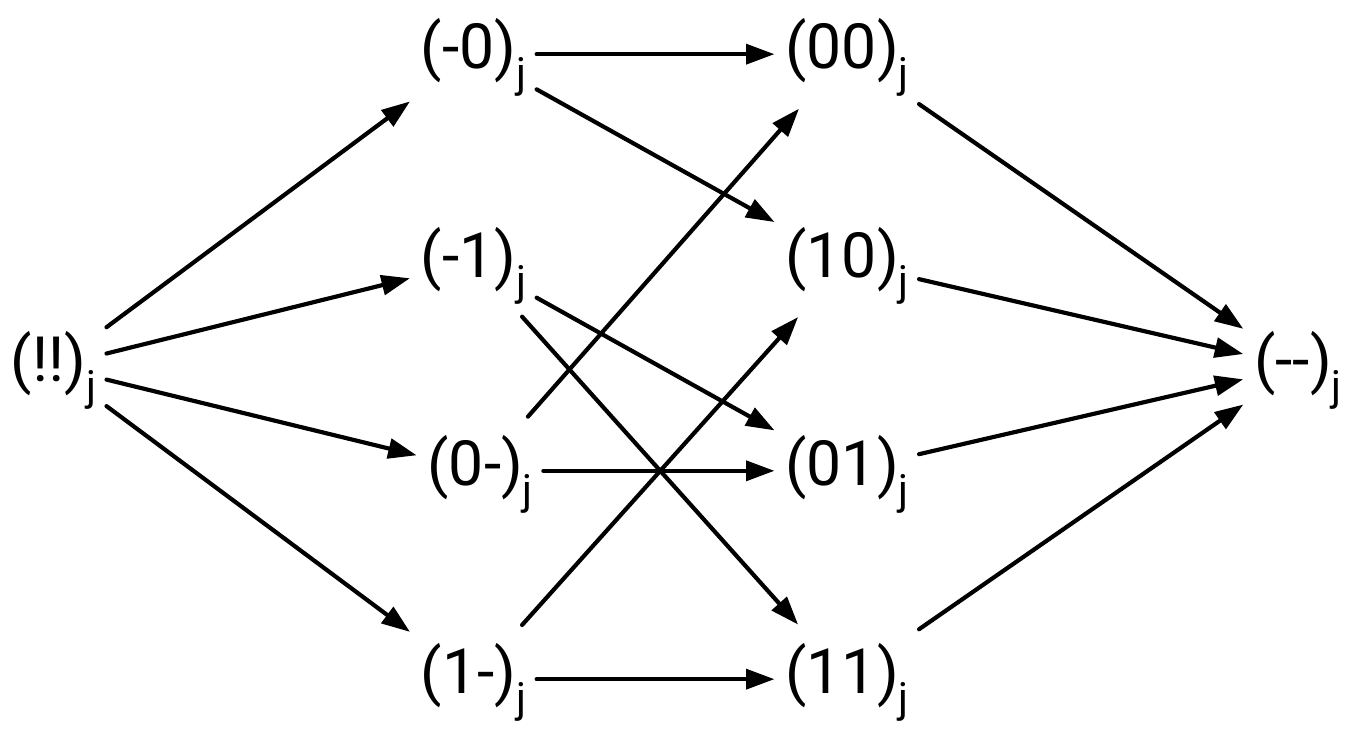}
	\caption{The strength relation of each present color $j$ different from the first and the special color.}
	\label{fig:diagram-present}
\end{figure}
\begin{figure}[ht!]
	\centering
	\includegraphics[scale=0.3]{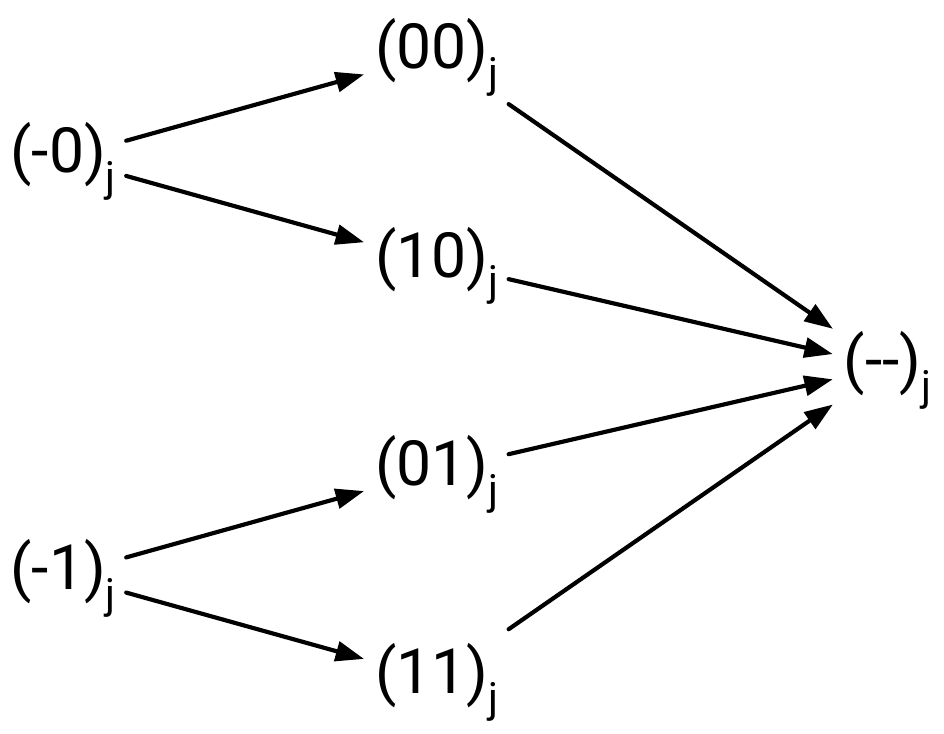}
	\caption{The strength relation of the special color $j=\Delta-i-1$.}
	\label{fig:diagram-special}
\end{figure}
\begin{figure}[ht!]
	\centering
	\includegraphics[scale=0.3]{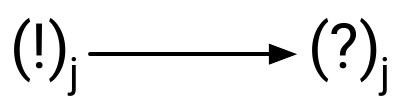}
	\caption{The strength relation of gone colors.}
	\label{fig:diagram-gone}
\end{figure}

\begin{lemma}\label{lem:diagram-first}
	The strength relation of the labels of color $1$ is given by the diagram of \Cref{fig:diagram-first}.
\end{lemma}
\begin{lemma}\label{lem:diagram-present}
	Let $j$ be a color satisfying $2 \le j \le \Delta - i - 1$. The strength relation of the labels of color $j$ is given by the diagram of \Cref{fig:diagram-present}.
\end{lemma}
\begin{lemma}\label{lem:diagram-special}
	Let $j$ be the special color, that is, $j = \Delta - i$. The strength relation of the labels of color $j$ is given by the diagram of \Cref{fig:diagram-special}.
\end{lemma}
\begin{lemma}\label{lem:diagram-gone}
	Let $j$ be a gone color, i.e., an integer satisfying $\Delta - i + 1 \le j \le \Delta$. The strength relation of the labels of color $j$ is given by the diagram of \Cref{fig:diagram-gone}.
\end{lemma}

\subsection{The first round elimination step}\label{ssec:re-first}
In this section, we define a family of problems $\Pi'_{i,\Delta}$, and we prove that $\re(\Pi_{i,\Delta})$ can be relaxed to $\Pi'_{i,\Delta}$. Each label of $\Pi'_{i,\Delta}$ is going to be a set of labels of $\Pi_{i,\Delta}$. Recall that with $\gen{L_1,\ldots,L_k}$ we denote the set of labels containing $L_1,\ldots,L_k$ and all of their (not necessarily direct) successors in the diagram of $\Pi_{i,\Delta}$.

\paragraph{Black constraint.}
We now define the black constraint $\edgeconst'_{i,\Delta}$ of $\Pi'_{i,\Delta}$. For each $j$ where $j$ is a gone color, the following configuration is allowed:
\begin{itemize}[noitemsep]
	\item $\gen{\reE{j}}^2 ~~~ \gen{\reQ{j}}$
\end{itemize}
For color $1$, the following configurations are allowed:
\begin{itemize}[noitemsep]
	\item $\gen{\reEE{1}} ~~~ \gen{\reMM{1}} ~~~ \gen{\reMY{1}{0}}$
	\item $\gen{\reMY{1}{0}}^2 ~~~ \gen{\reMY{1}{1}}$
\end{itemize}
For each color $j$ satisfying $2 \le j \le \Delta - i - 1$, the following configurations are allowed:
\begin{itemize}[noitemsep]
	\item $\gen{\reMM{j}} ~~~ \gen{\reEE{j}}^2$
	\item $\gen{\reMY{j}{0},\reXM{j}{1}} ~~~ \gen{\reXY{j}{0}{0}}^2$
	\item $\gen{\reMY{j}{1},\reXM{j}{1}} ~~~ \gen{\reXY{j}{0}{0}} ~~~ \gen{\reXY{j}{0}{1}}$
	\item $\gen{\reMY{j}{0},\reXM{j}{1}} ~~~ \gen{\reXY{j}{0}{1}}^2$
	\item $\gen{\reMY{j}{1},\reXM{j}{0}} ~~~ \gen{\reXY{j}{0}{0}} ~~~ \gen{\reXY{j}{1}{0}}$
	\item $\gen{\reMY{j}{0},\reXM{j}{0}} ~~~ \gen{\reXY{j}{0}{0}} ~~~ \gen{\reXY{j}{1}{1}}$
	\item $\gen{\reMY{j}{0},\reXM{j}{0}} ~~~ \gen{\reXY{j}{0}{1}} ~~~ \gen{\reXY{j}{1}{0}}$
	\item $\gen{\reMY{j}{1},\reXM{j}{0}} ~~~ \gen{\reXY{j}{0}{1}} ~~~ \gen{\reXY{j}{1}{1}}$
	\item $\gen{\reMY{j}{1},\reXM{j}{1}} ~~~ \gen{\reXY{j}{1}{0}}^2$
	\item $\gen{\reMY{j}{0},\reXM{j}{1}} ~~~ \gen{\reXY{j}{1}{0}} ~~~ \gen{\reXY{j}{1}{1}}$
	\item $\gen{\reMY{j}{1},\reXM{j}{1}} ~~~ \gen{\reXY{j}{1}{1}}^2$
	\item $\gen{\reMY{j}{0}} ~~~  \gen{\reXY{j}{0}{0}} ~~~ \gen{\reXY{j}{0}{0},\reXY{j}{1}{1}}$
	\item $\gen{\reMY{j}{1}} ~~~ \gen{\reXY{j}{0}{0}} ~~~ \gen{\reXY{j}{0}{1},\reXY{j}{1}{0}}$
	\item $\gen{\reMY{j}{0}} ~~~ \gen{\reXY{j}{0}{1}} ~~~ \gen{\reXY{j}{0}{1},\reXY{j}{1}{0}}$
	\item $\gen{\reMY{j}{1}} ~~~ \gen{\reXY{j}{0}{0},\reXY{j}{1}{1}} ~~~ \gen{\reXY{j}{0}{1}}$
	\item $\gen{\reMY{j}{0}} ~~~ \gen{\reMY{j}{1}} ~~~ \gen{\reXY{j}{1}{0}}$
	\item $\gen{\reXY{j}{1}{1}} ~~~ \gen{\reMY{j}{0}}^2$
	\item $\gen{\reXY{j}{1}{0}} ~~~ \gen{\reXY{j}{0}{1},\reXY{j}{1}{0}}^2$
	\item $\gen{\reXY{j}{1}{0}} ~~~ \gen{\reXY{j}{0}{0},\reXY{j}{1}{1}}^2$
	\item $\gen{\reXY{j}{0}{0},\reXY{j}{1}{1}} ~~~ \gen{\reXY{j}{0}{1},\reXY{j}{1}{0}} ~~~ \gen{\reXY{j}{1}{1}}$
	\item $\gen{\reXY{j}{1}{1}} ~~~ \gen{\reMY{j}{1}}^2$
	\item $\gen{\reXM{j}{1}} ~~~ \gen{\reXM{j}{0}}^2$
	\item $\gen{\reXM{j}{1}}^3$
\end{itemize}
For the special color $j = \Delta - i$, the following configurations are allowed:
\begin{itemize}[noitemsep]
	\item $\gen{\reXY{j}{0}{0},\reXY{j}{0}{1}} ~~~ \gen{\reMY{j}{0},\reMY{j}{1}}^2$
	\item $\gen{\reXY{j}{1}{0},\reXY{j}{1}{1}} ~~~ \gen{\reMY{j}{0},\reMY{j}{1}}^2$
\end{itemize}
This concludes the definition of the black constraint.

Before proving the relation between the black constraint of $\re(\Pi_{i,\Delta})$ and $\edgeconst'_{i,\Delta}$, we will state a useful observation, based on the following definitions.
For any set $S$ that contains only labels of the same color, say $j$, we define the color of $S$ as $j$. For any configuration $C$ that contains only sets of color $j$, we define the color of $C$ as $j$.

In order for a configuration of the black constraint of $\re(\Pi_{i,\Delta})$ to satisfy the universal quantifier, the following must clearly hold.
\begin{observation}\label{obs:configurations-same-color}
	All labels of $\re(\Pi_{i,\Delta})$ are sets of labels of $\Pi_{i,\Delta}$, where each set $S$ satisfies that all labels in $S$ have the same color.
	Moreover, each configuration of the black constraint of $\re(\Pi_{i,\Delta})$ consists of sets of the same color.
\end{observation}

In the following, we will implicitly use the fact that, for any set $\mathcal{L}$ of labels, if two labels satisfy $L_i \le L_j$, and $L_i,L_j \in \mathcal{L}$, then $\gen{L \mid L \in \mathcal{L}} = \gen{L \mid L \in \mathcal{L} \setminus \{L_j\}}$.
\begin{lemma}\label{lem:step11}
	The black constraint of $\re(\Pi_{i,\Delta})$ is dominated by $\edgeconst'_{i,\Delta}$.
\end{lemma}
\begin{proof}
	By \Cref{obs:configurations-same-color}, each configuration of $\re(\Pi_{i,\Delta})$ is composed of labels with the same color.
	In order to prove the statement, we consider each color $j$ separately.
	\begin{itemize}
		\item For all $j \neq \Delta - i$, the configurations of $\edgeconst'_{i,\Delta}$ are exactly the condensed configurations provided when describing $\edgeconst_{i,\Delta}$. We prove that, by combining configurations of color $j$, we do not get any additional configuration. By \Cref{lem:combining}, this implies that the configurations of color $j$ are maximal-complete.
		\item For $j = \Delta -i$, we prove that each configuration of color $j$ of $\re(\Pi_{i,\Delta})$ is dominated by (at least) one of the two configurations of color $j$ that we provided.
	\end{itemize}

	Let $j$ be a gone color, i.e., $j > \Delta-i$. By \Cref{obs:no-union-comparable}, when combining the unique configuration of color $j$ with itself (for any index $u$ and permutation $\sigma$), we do not obtain any new configuration, since $\gen{\reE{j}}$ and $\gen{\reQ{j}}$ are comparable.

	Let us now consider color $1$. By \Cref{obs:no-union-comparable}, by combining $\gen{\reEE{1}} ~ \gen{\reMM{1}} ~ \gen{\reMY{1}{0}}$ with itself, we cannot get additional configurations (not dominated by ones already present), since all sets in this configuration are comparable with each other.
	By \Cref{obs:no-union-comparable}, in order to get new configurations when combining $\gen{\reMY{1}{0}}^2 ~ \gen{\reMY{1}{1}}$ with itself, we need to take the union of $\gen{\reMY{1}{0}}$ with $ \gen{\reMY{1}{1}}$, thus obtaining the configuration $\gen{\reMY{1}{0},\reMY{1}{1}} ~ \gen{\reMY{1}{0}} ~ \gen{\reMM{1}}$, which is dominated by $\gen{\reEE{1}} ~ \gen{\reMM{1}} ~ \gen{\reMY{1}{0}}$.
	Finally, if we combine $\gen{\reEE{1}} ~ \gen{\reMM{1}} ~ \gen{\reMY{1}{0}}$ with $\gen{\reMY{1}{0}}^2 ~ \gen{\reMY{1}{1}}$, by \Cref{obs:no-union-comparable} we need to take the union of $\gen{\reMY{1}{0}}$ with $ \gen{\reMY{1}{1}}$, thus obtaining the configuration  $\gen{\reMY{1}{0},\reMY{1}{1}} ~ \gen{\reMY{1}{0}} ~ \gen{\reMM{1}}$, which again is dominated by $\gen{\reEE{1}} ~ \gen{\reMM{1}} ~ \gen{\reMY{1}{0}}$.
	
	In the case $2 \le j \le \Delta - i - 1$, the claim follows by applying \Cref{obs:already-maximized} (and the definition of $\re$).
	
	Finally, we consider the case $j =  \Delta - i$, i.e., the special color. In this case, we prove that all the configurations satisfying the universal quantifier are dominated by (at least) one of the two configurations present in $\edgeconst'_{i,\Delta}$. We first observe that, since $\gen{\reMY{j}{0},\reMY{j}{1}}$ contains all labels of color $j$ (see \Cref{fig:diagram-special}), any configuration $C = C_1 ~ C_2 ~ C_3$ of color $j$ in the black constraint of $\re(\Pi_{i,\Delta})$ that is not dominated by the two configurations present in $\edgeconst'_{i,\Delta}$ must necessarily satisfy that all its sets $C_i$ satisfy $C_i \not\subseteq \gen{\reXY{j}{0}{0},\reXY{j}{0}{1}}$ and $C_i \not\subseteq \gen{\reXY{j}{1}{0},\reXY{j}{1}{1}}$. By \Cref{obs:rcs}, each set in such a configuration must be right-closed, and hence it holds that, for all $i$, $C_i$ is a superset of at least one of the following four sets: $X_1 =\{\reXY{j}{0}{0},\reXY{j}{1}{0}\}$, $X_2 = \{\reXY{j}{0}{0},\reXY{j}{1}{1}\}$, $X_3 = \{\reXY{j}{0}{1},\reXY{j}{1}{0}\}$, $X_4 = \{\reXY{j}{0}{1},\reXY{j}{1}{1}\}$.
	We consider all possible choices in $\{X_1,X_2,X_3,X_4\}^3$ (excluding permutations):
	\begin{itemize}
		\item For any choice in $\{X_3,X_4\}^3$, we can pick the configuration  $\reXY{j}{0}{1}^3$, which is not present in $\edgeconst_{i,\Delta}$.
		\item For any choice in $\{X_1,X_2\}^3$, we can pick the configuration  $\reXY{j}{0}{0} ~ \reXY{j}{1}{0}^2$ or the configuration $\reXY{j}{0}{0} ~ \reXY{j}{1}{1}^2$, which are not present in $\edgeconst_{i,\Delta}$.
		\item For any choice in $\{X_1,X_2\}^2 ~ \{X_3,X_4\}$, we can pick the configuration  $\reXY{j}{0}{0}^2 ~ \reXY{j}{0}{1}$, which is not present in $\edgeconst_{i,\Delta}$.
		\item For any choice in $\{X_1,X_2\} ~ \{X_3\}^2$, we can pick the configuration $\reXY{j}{0}{0} ~ \reXY{j}{1}{0}^2$, which is not present in $\edgeconst_{i,\Delta}$.
		\item For any choice in $\{X_1,X_2\} ~ \{X_4\}^2$, we can pick the configuration $\reXY{j}{0}{0} ~ \reXY{j}{1}{1}^2$, which is not present in $\edgeconst_{i,\Delta}$.
		\item For any choice in $\{X_1,X_2\} ~ \{X_3\} ~ \{X_4\}$, we can pick the configuration $\reXY{j}{1}{0} ~ \reXY{j}{0}{1} ~ \reXY{j}{1}{1}$, which is not present in $\edgeconst_{i,\Delta}$. \qedhere
	\end{itemize}
\end{proof}

\paragraph{White constraint.}
We first define an intermediate constraint $\nodeconst^*_{i,\Delta}$. Then, we define  $\nodeconst'_{i,\Delta}$  as a function of $\nodeconst^*_{i,\Delta}$. The constraint $\nodeconst^*_{i,\Delta}$ is defined as follows. Let $x = \Delta - i$ be the special color.
For each configuration $C \in  \nodeconst_{i,\Delta}$, the constraint $\nodeconst^*_{i,\Delta}$ contains the condensed configuration $C^*$ obtained by replacing each label of $C$ with the disjunction of sets according to the following rules (where $j$ denotes the color of $L$):
\begin{itemize}[noitemsep]
	\item $[\gen{L}]$, if $j \neq x$;
	\item $[\gen{\reXY{j}{0}{0},\reXY{j}{0}{1}}]$, if $j = x$ and $L \in \{\reXY{j}{0}{0},\reXY{j}{0}{1}\}$;
	\item $[\gen{\reXY{j}{1}{0},\reXY{j}{1}{1}}]$, if $j = x$ and $L \in \{\reXY{j}{1}{0},\reXY{j}{1}{1}\}$;
	\item $[\gen{\reMY{j}{0},\reMY{j}{1}}]$, if $j = x$ and $L \in \{\reMY{j}{0},\reMY{j}{1}\}$;
	\item $[\gen{\reXY{j}{0}{0},\reXY{j}{0}{1}}, \gen{\reXY{j}{1}{0},\reXY{j}{1}{1}}]$, if $j = x$ and $L = \reMM{j}$.
\end{itemize}

Let $\Sigma'_{i,\Delta}$ be the set of labels that appear in $\edgeconst'_{i,\Delta}$.
The white constraint $\nodeconst'_{i,\Delta}$ of $\Pi'_{i,\Delta}$ contains all the configurations described by the set of condensed configurations given by the following process. Take a condensed configuration $C^*$ in $\nodeconst^*_{i,\Delta}$, and replace each disjunction of labels $[L_1 \ldots L_k]$ with the disjunction containing $L_1,\ldots,L_k$ and all the labels $L \in \Sigma'_{i,\Delta}$ that are supersets of at least one label in $\{L_1,\ldots,L_k\}$.

\begin{lemma}\label{lem:step12}
	$\nodeconst'_{i,\Delta}$ is the set of all configurations $S_1 \ldots S_\Delta$ of sets contained in $\Sigma'_{i,\Delta}$ such that there exists some tuple  $(\ell_1, \ldots, \ell_\Delta) \in S_1 \times \ldots \times S_\Delta$ satisfying that $\ell_1~\ldots~\ell_\Delta $ is contained in $\nodeconst_{i,\Delta}$.
\end{lemma}
\begin{proof}
	Consider the collection $\mathcal{C}$ of all configurations obtained from $\nodeconst_{i,\Delta}$, by replacing, in each configuration, each label $L$ with the disjunction containing all sets $L' \in \Sigma'_{i,\Delta}$ satisfying $L \in L'$.

	It is sufficient to prove that the collection of configurations that can be picked from $\mathcal{C}$ is equal to $\nodeconst'_{i,\Delta}$. 
	This follows directly by the definition of $\nodeconst'_{i,\Delta}$ and \Cref{lem:diagram-first,lem:diagram-gone,lem:diagram-present,lem:diagram-special}.
\end{proof}

By combining \Cref{lem:step11} and \Cref{lem:step12}, by the definition of relaxation, we obtain the following.
\begin{lemma}\label{lem:all-step-1}
	The problem $\Pi'_{i,\Delta}$ is a relaxation of $\re(\Pi_{i,\Delta})$.
\end{lemma} 

\paragraph{White constraint example.}
We now provide an example of $\nodeconst'_{i,\Delta}$, for the case $\Delta=7$ and $i=2$, by showing some examples of replacements.

In $\nodeconst_{2,7}$, the following configuration is present:
\[
\reMY{1}{0} ~\reXY{2}{0}{1} ~\reXY{3}{1}{1} ~\reXY{4}{1}{0} ~\reXY{5}{0}{0} ~\reQ{6} ~\reQ{7}.
\]
We get that, in $\nodeconst'_{2,7}$, the following configuration is present:
\[
\gen{\reMY{1}{0}} ~ \gen{\reXY{2}{0}{1}} ~ \gen{\reXY{3}{1}{1}} ~ \gen{\reXY{4}{1}{0}} ~ \gen{\reXY{5}{0}{0},\reXY{5}{0}{1}} ~ \gen{\reQ{6}} ~ \gen{\reQ{7}}.
\]
Moreover,  $\nodeconst'_{2,7}$ contains all configurations that can be obtained by starting from such a configuration and replacing some labels with arbitrary supersets. For example, the following configuration is also present:
\[
\gen{\reMY{1}{0}} ~ \gen{\reXY{2}{0}{1}} ~ \gen{\reXY{3}{1}{1}} ~ \gen{\reXY{4}{1}{0}} ~ \gen{\reXY{5}{0}{0},\reXY{5}{0}{1}} ~ \gen{\reE{6}} ~ \gen{\reE{7}}.
\]

For another example, consider the following configuration present in $\nodeconst_{2,7}$:
\[
\reMY{1}{0} ~\reXY{2}{0}{1} ~\reXM{3}{1} ~\reMM{4} ~\reMM{5} ~\reQ{6} ~\reQ{7}.
\]
We get that, in $\nodeconst'_{2,7}$, all the configurations given by the following condensed configuration (plus all the ones obtained by replacing labels with supersets) are present:
\[
[\gen{\reMY{1}{0}}] ~[\gen{\reXY{2}{0}{1}}] ~ [\gen{\reXM{3}{1}}] ~[\gen{\reMM{4}}] ~ [\gen{\reXY{5}{0}{0},\reXY{5}{0}{1}},\gen{\reXY{5}{1}{0},\reXY{5}{1}{1}}] ~[\gen{\reQ{6}}] ~[\gen{\reQ{7}}].
\]

For another example, consider the following configuration present in $\nodeconst_{2,7}$:
\[
\reMM{1} ~\reMM{2} ~\reMM{3} ~\reMM{4} ~\reMY{5}{0} ~\reQ{6} ~\reQ{7} 
\]
We get that, in $\nodeconst'_{2,7}$, the following configuration is present, plus all the ones obtained by replacing labels with supersets:
\[
\gen{\reMM{1}} ~\gen{\reMM{2}} ~\gen{\reMM{3}} ~\gen{\reMM{4}} ~ \gen{\reMY{5}{0},\reMY{5}{1}} ~\gen{\reQ{6}} ~\gen{\reQ{7}} 
\]

We observe that the constraint $\nodeconst'_{i,\Delta}$ can be described in the following alternative (and more explicit) way, that will be useful later.
Let $\nodeconst^+_{i,\Delta}$ be the set containing the following condensed configurations.
\begin{itemize}
	\item \textbf{Strikethrough configurations}. For all integers $a,b$ satisfying $1 \le a<b \le \Delta - i$, and for all vectors $(y_j \mid y_j \in \{0,1\} \text{ and } a \le j < b)$, the condensed configuration $L_1 \ldots L_\Delta$ is allowed, where:
	\begin{itemize}[noitemsep]
		\item For all $j$ such that $1 \le j < a$, $L_j = [\gen{\reMM{j}}]$;
		\item $L_a = [\gen{\reMY{a}{y_a}}]$;
		\item For all $j$ such that $a < j < b$, $L_j = [\gen{\reXY{j}{y_{j-1}\,}{y_{j}}}]$;
		\item If $b \le \Delta - i - 1$, $L_b = [\gen{\reXM{b}{y_{b-1}}}]$;
		\item For all $j$ such that $b < j \le \Delta - i -1 $, $L_j = [\gen{\reMM{j}}]$;
		\item If $b \neq \Delta - i$, $L_{\Delta-i} = [\gen{\reXY{\Delta - i}{0}{0},\reXY{\Delta - i}{0}{1}},\gen{\reXY{\Delta - i}{1}{0},\reXY{\Delta - i}{1}{1}}]$;
		\item If $b = \Delta - i$, $L_{\Delta-i} = [\gen{\reXY{\Delta-i}{y_{b-1}}{0},\reXY{\Delta-i}{y_{b-1}}{1}}]$;
		\item For all $j$ such that $\Delta - i + 1 \le j \le \Delta$, $L_j = [\gen{\reQ{j}}]$.
	\end{itemize}
	\item \textbf{Grabbing-present configurations}. For all integers $a$ satisfying $1 \le a \le \Delta - i - 1$, the condensed configuration $L_1 \ldots L_\Delta$ is allowed, where:
	\begin{itemize}
		\item $L_a = [\gen{\reEE{a}}]$;
		\item For all $j$ such that $1 \le j \le \Delta - i - 1$ and $j \neq a$, $L_j = [\gen{\reMM{j}}]$;
		\item $L_{\Delta-i} = [\gen{\reXY{\Delta - i}{0}{0},\reXY{\Delta - i}{0}{1}},\gen{\reXY{\Delta - i}{1}{0},\reXY{\Delta - i}{1}{1}}]$;
		\item For all $j$ such that $\Delta - i + 1 \le j \le \Delta$, $L_j = [\gen{\reQ{j}}]$.
	\end{itemize}
	\item \textbf{Grabbing-gone configurations}. For all integers $a$ satisfying $\Delta - i + 1 \le a \le \Delta$, the configuration $L_1 \ldots L_\Delta$ is allowed, where:
	\begin{itemize}
		\item $L_a = [\gen{\reE{a}}]$;
		\item For all $j$ such that  $1 \le j \le \Delta - i - 1$, $L_j = [\gen{\reMM{j}}]$;
		\item $L_{\Delta-i} = [\gen{\reXY{\Delta - i}{0}{0},\reXY{\Delta - i}{0}{1}},\gen{\reXY{\Delta - i}{1}{0},\reXY{\Delta - i}{1}{1}}]$;
		\item For all $j$ such that  $\Delta - i + 1 \le j \le \Delta$ and $j \neq a$, $L_j = [\gen{\reQ{j}}]$.
	\end{itemize}
	\item \textbf{Grabbing-special configuration}. The configuration $L_1 \ldots L_\Delta$ is allowed where:
	\begin{itemize}
		\item For all $j$ such that  $1 \le j \le \Delta - i - 1$, $L_j = [\gen{\reMM{j}}]$;
		\item $L_{\Delta-i} = [\gen{\reMY{\Delta-i}{0},\reMY{\Delta-i}{1}}]$;
		\item For all $j$ such that  $\Delta - i + 1 \le j \le \Delta$, $L_j = [\gen{\reQ{j}}]$. 
	\end{itemize}
\end{itemize}
Then, let $\nodeconst^*_{i,\Delta}$ be the set of condensed configurations obtained as follows. Take a configuration $C^+$ from $\nodeconst^+_{i,\Delta}$. Replace, in $C^+$, each disjunction of labels $[L_1 \ldots L_k]$ with the disjunction containing $L_1,\ldots,L_k$ and all the labels $L \in \Sigma'_{i,\Delta}$ that are supersets of at least one label in $\{L_1,\ldots,L_k\}$.

By the definitions of  $\nodeconst^*_{i,\Delta}$ and $\nodeconst'_{i,\Delta}$, we obtain the following observation.
\begin{observation}\label{obs:black-intermediate-explicit}
	The white constraint $\nodeconst'_{i,\Delta}$ is the set containing all the configurations $C$ that can be picked from $\nodeconst^*_{i,\Delta}$.
\end{observation}

\subsection{The diagram of \texorpdfstring{$\Pi'_{i,\Delta}$}{Pi'}}\label{ssec:re-diagram-2}
We now provide the diagram of $\Pi'_{i,\Delta}$, that is, we prove the strength relation between the labels of $\Pi'_{i,\Delta}$ w.r.t.\ its white constraint. Again, a necessary condition for a label to be at least as strong as another label is for the two labels to be of the same color.
\begin{figure}[ht!]
	\centering
	\includegraphics[scale=0.3]{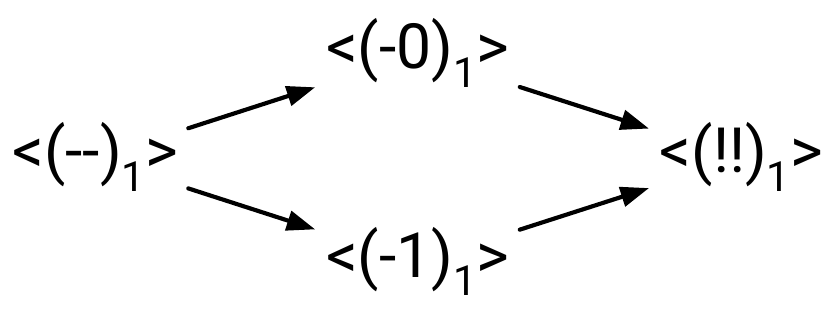}
	\caption{The strength relation of the first color.}
	\label{fig:diagram-intermediate-first}
\end{figure}
\begin{figure}[ht!]
	\centering
	\includegraphics[scale=0.3]{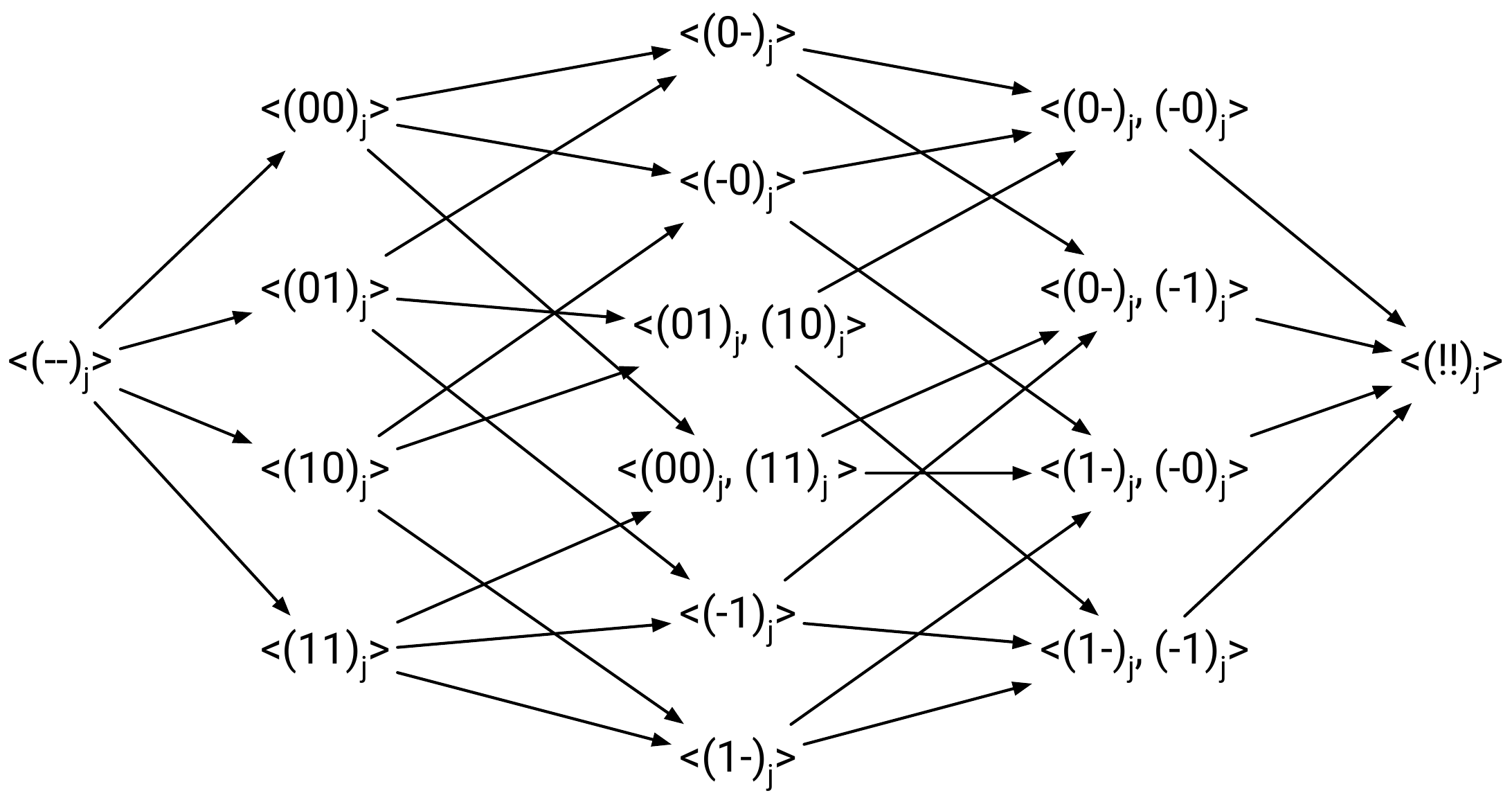}
	\caption{The strength relation of each present color $j$ different from the first and the special color.}
	\label{fig:diagram-intermediate-present}
\end{figure}
\begin{figure}[ht!]
	\centering
	\includegraphics[scale=0.3]{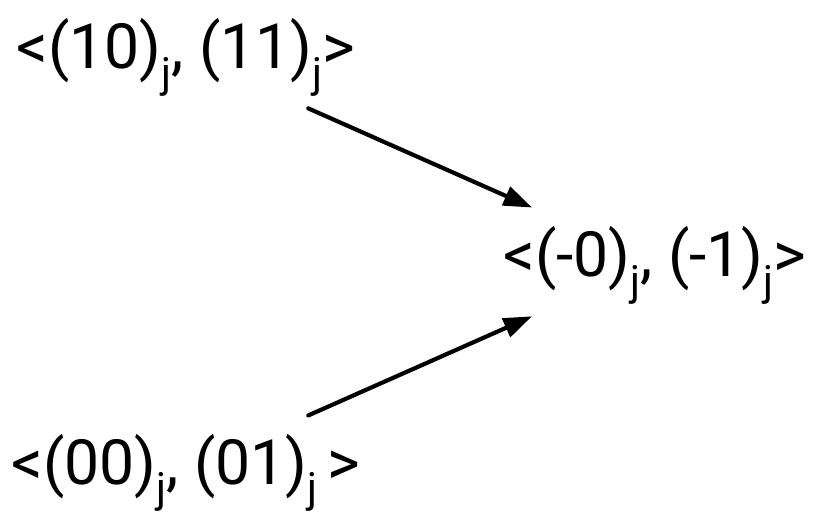}
	\caption{The strength relation of the special color $j=\Delta - i - 1$.}
	\label{fig:diagram-intermediate-special}
\end{figure}
\begin{figure}[ht!]
	\centering
	\includegraphics[scale=0.3]{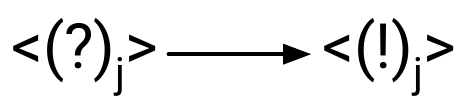}
	\caption{The strength relation of gone colors.}
	\label{fig:diagram-intermediate-gone}
\end{figure}

\begin{lemma}\label{lem:diagram-intermediate-first}
	The strength relation of the labels of color $1$ is given by the diagram of \Cref{fig:diagram-intermediate-first}.
\end{lemma}
\begin{lemma}\label{lem:diagram-intermediate-present}
	Let $j$ be a color satisfying $2 \le j \le \Delta - i - 1$. The strength relation of the labels of color $j$ is given by the diagram of \Cref{fig:diagram-intermediate-present}.
\end{lemma}
\begin{lemma}\label{lem:diagram-intermediate-special}
	Let $j$ be the special color, that is, $j = \Delta - i$. The strength relation of the labels of color $j$ is given by the diagram of \Cref{fig:diagram-intermediate-special}.
\end{lemma}
\begin{lemma}\label{lem:diagram-intermediate-gone}
	Let $j$ be a gone color, i.e., an integer satisfying $\Delta - i + 1 \le j \le \Delta$. The strength relation of the labels of color $j$ is given by the diagram of \Cref{fig:diagram-intermediate-gone}.
\end{lemma}

The proofs of \Cref{lem:diagram-intermediate-gone,lem:diagram-intermediate-first,lem:diagram-intermediate-special,lem:diagram-intermediate-present} are just a long case analysis, and hence they are deferred to \Cref{sec:lb-proofs-omitted-second-step}.
By the above lemmas, we get the following.
\begin{observation}\label{obs:diagram-by-set-inclusion}
	$L_1 \le L_2 \iff L_1 \subseteq L_2$.
\end{observation}

\subsection{The second round elimination step}\label{ssec:re-second}
In this section, we define a family of problems $\Pi''_{i,\Delta}$, and we prove that $\rere(\Pi'_{i,\Delta})$ can be relaxed to $\Pi''_{i,\Delta}$. Then, we show how to relax $\rere(\Pi'_{i,\Delta})$ even more in order to obtain $\Pi_{i+1,\Delta}$.
Each label of $\Pi''_{i,\Delta}$ is going to be a set of labels of $\Pi'_{i,\Delta}$. Recall that with $\gen{L_1,\ldots,L_k}$ (where $L_1,\ldots,L_k$ are labels of $\Pi'_{i,\Delta}$) we denote the set of labels obtained by taking $L_1,\ldots,L_k$ and all their successors in the diagram of $\Pi'_{i,\Delta}$. 

\paragraph{White constraint.}
The white constraint $\nodeconst''_{i,\Delta}$ of $\Pi''_{i,\Delta}$ is defined as the set $\nodeconst^*_{i,\Delta}$ described immediately before \Cref{obs:black-intermediate-explicit}.
By \Cref{obs:diagram-by-set-inclusion}, we obtain that an alternative but equivalent definition of $\nodeconst''_{i,\Delta}$ is the following. 
\begin{observation}\label{obs:from-nodeplus}
	$\nodeconst''_{i,\Delta}$ is equal to the set containing, for each condensed configuration $C^+$ in $\nodeconst^+_{i,\Delta}$, the configuration generated by $C^+$.
\end{observation}

We will prove that the white constraint of $\rere(\Pi'_{i,\Delta})$ is dominated by $\nodeconst''_{i,\Delta}$. For this purpose, we first define a function $f$ that takes as input two (condensed) strikethrough configurations $C_1 = L_{1,1} \ldots L_{1,\Delta}$ and $C_2 = L_{2,1} \ldots L_{2,\Delta}$ from $\nodeconst^+_{i,\Delta}$, an index $h$ satisfying $1 \le h \le \Delta - i - 1$, an index $1 \le u \le \Delta$, and a value $k \in \{1,2\}$ as follows. First, for $j \in \{1,2\}$, let $t_{j,h} = (t_{j,h,1},t_{j,h,2})$ be defined as:
\begin{itemize}[noitemsep]
	\item $t_{j,h} = (\minus,\minus)$ if $L_{j,h} = [\gen{\reMM{h}}]$;
	\item $t_{j,h} = (b,\minus)$ if $L_{j,h} = [\gen{\reXM{h}{b}}]$;
	\item $t_{j,h} = (\minus,b)$ if $L_{j,h} = [\gen{\reMY{h}{b}}]$;
	\item $t_{j,h} = (b_1,b_2)$ if $L_{j,h} = [\gen{\reXY{h}{b_1~}{b_2}}]$.
\end{itemize}
Then, let $f(C_1,C_2,h,u,k) = $
\begin{itemize}[noitemsep]
	\item $\minus$ if $h = u$;
	\item $\minus$ if  $h \neq u$ and $\{ t_{1,h,k}, t_{2,h,k} \} \in \{\{\minus\},\{0,1\}\}$;
	\item $0$ if $h \neq u$ and  $\{ t_{1,h,k}, t_{2,h,k} \} \in \{\{0\},\{0,\minus\}\}$;
	\item $1$ if $h \neq u$ and $\{ t_{1,h,k}, t_{2,h,k} \} \in \{\{1\},\{1,\minus\}\}$.
\end{itemize}

\begin{lemma}\label{lem:label-domination}
	Let $T_{h,k} = f(C_1,C_2,h,u,k)$.
	Let $L = [\gen{\reXY{h}{T_{h,1~}}{T_{h,2}}}]$. Then,
	\begin{itemize}[noitemsep]
		\item if $h = u$, $\gen{\ell \mid \ell \in L} \supseteq \gen{\ell \mid \ell \in L_{1,h}} \cup \gen{\ell \mid \ell \in L_{2,h}} $;
		\item if $h \neq u$, $\gen{\ell \mid \ell \in L} \supseteq \gen{\ell \mid \ell \in L_{1,h}} \cap \gen{\ell \mid \ell \in L_{2,h}} $.
	\end{itemize}
\end{lemma}
\begin{proof}
	By the definition of $f$, if $h = u$, then $L = [\gen{\reMM{h}}]$, and the set $\gen{\ell \mid \ell \in L}$ contains all labels of color $h$ (see \Cref{fig:diagram-intermediate-first,fig:diagram-intermediate-present}). Hence in this case the claim holds.
	
	If $h \neq u$, by inspecting the diagrams in \Cref{fig:diagram-intermediate-first,fig:diagram-intermediate-present}, it is straightforward to verify that it always holds that $\gen{\ell \mid \ell \in L_{1,h}} \cap \gen{\ell \mid \ell \in L_{2,h}}$ (i.e., the intersection between the set containing the unique label contained in the disjunction $L_{1,h}$ and all successors of that label, and the set containing the unique label contained in the disjunction $L_{2,h}$ and all successors of that label) is a subset of the set containing the unique label contained in the disjunction $L$ and all successors of that label.
\end{proof}

\begin{lemma}\label{lem:minus-on-one-side-of-union}
	Assume that $1\le u \le \Delta - i - 1$, and let $\ell_1$ denote the unique label in $L_{1,u}$ and $\ell_2$ the unique label in $L_{2,u}$.
	Assume further that $\ell_1$ and $\ell_2$ are incomparable.
	Then,
	\begin{itemize}
		\item for all $1 \le u < \Delta - i - 1$, it holds that, if $T_{u+1,1} \neq \minus$, then $\{t_{1,u,1}, t_{2,u,1}\} = \{0,1\}$ or $\{t_{1,u},t_{2,u}\} = \{(x_1,\minus),(\minus,x_2)\}$ for some $x_1,x_2 \in \{0,1\}$, and
		\item for all $1 < u \le \Delta - i - 1$, it holds that, if $T_{u-1,2} \neq \minus$, then $\{t_{1,u,2}, t_{2,u,2}\} = \{0,1\}$ or $\{t_{1,u},t_{2,u}\} = \{(x_1,\minus),(\minus,x_2)\}$ for some $x_1,x_2 \in \{0,1\}$.
	\end{itemize}
\end{lemma}
\begin{proof}
	 For all $1 \le u < \Delta - i - 1$, if $T_{u+1,1} \neq \minus$, it must hold that $\{t_{1,u,2}, t_{2,u,2}\}$ (i.e., the set of output bits in position $u$ of $C_1$ and $C_2$) is in $\{\{0,\minus\},\{1,\minus\},\{0\},\{1\}\}$.
By \Cref{lem:diagram-intermediate-present} (and \Cref{lem:diagram-intermediate-first}) it must hold that either $\{t_{1,u,1}, t_{2,u,1}\} = \{0,1\}$ or $\{t_{1,u},t_{2,u}\} = \{(x_1,\minus),(\minus,x_2)\}$ for some $x_1,x_2 \in \{0,1\}$. 
	 The other case is symmetric.
\end{proof}

\begin{observation}\label{obs:same-borders}
	 For all $1 \le j < \Delta - i - 1$ satisfying $j \neq u$ and $j+1 \neq u$, $T_{j,2} = T_{j+1,1}$.
\end{observation}

\begin{lemma}\label{lem:step21}	
	The white constraint of $\rere(\Pi'_{i,\Delta})$ is dominated by $\nodeconst''_{i,\Delta}$.
\end{lemma}
\begin{proof}
    By definition, $\nodeconst''_{i,\Delta}$ contains exactly the condensed configurations used to describe $\nodeconst'_{i,\Delta}$. Hence, we only need to prove that $\nodeconst''_{i,\Delta}$ is maximal-complete.

	In the following, we use the description of $\nodeconst''_{i,\Delta}$ given by \Cref{obs:from-nodeplus}.
	There are $4$ types of configurations in $\nodeconst^+_{i,\Delta}$, and by \Cref{lem:diagram-intermediate-gone,lem:diagram-intermediate-first,lem:diagram-intermediate-present,lem:diagram-intermediate-special}, for each color $1 \le j \le \Delta$, all sets of labels of color $j$ appearing in at least one of the configurations generated by the condensed configurations of types grabbing-present, grabbing-gone, and grabbing-special are comparable with all sets of labels of color $j$ appearing in at least one of the configurations generated by at least one of the condensed configurations in $\nodeconst^+_{i,\Delta}$. Moreover, in order to combine two configurations $C_1$ and $C_2$ and obtain a configuration $C$ that does not contain any empty set, it must clearly hold that the union and all the intersections must be taken on labels of the same color (i.e., the permutation $\sigma$ used to combine $C_1$ with $C_2$ must always pair labels of the same color). Hence, by \Cref{obs:no-union-comparable}, the only way we could get additional configurations is by combining two configurations generated by strikethrough configurations.

	Hence, let $C_1 = L_{1,1} \ldots L_{1,\Delta}$ and $C_2 = L_{2,1} \ldots L_{2,\Delta}$ be two strikethrough configurations. Let $a_1,a_2,b_1,b_2$ be their parameters and $y_{1,j},y_{2,j}$ their bits. Let $u$ be the position of the union, and note that the permutation $\sigma$ used to combine the configurations generated by $C_1$ and $C_2$ must be the identity function, since otherwise, when taking intersections, we would obtain empty sets. 
	Moreover, for the sake of contradiction, assume that the configuration obtained by combining the configuration generated by $C_1$ and the configuration generated by $C_2$ w.r.t.\ $u$ and $\sigma$ is not contained in $\nodeconst''_{i,\Delta}$.
	This implies that, by \Cref{obs:no-union-comparable}, $u \le \Delta - i$. Let $C$ be the combination of the configurations generated by $C_1$ and $C_2$ w.r.t.\ $u$ and the identity function $\sigma$.
	
	We start by handling a special case, that is, when $a_1 = b_2 = u$ or $a_2 = b_1 = u$. W.l.o.g., let $a_1 = b_2 = u$. Consider the configuration $C^*$ generated by the condensed configuration $L_1 \ldots L_\Delta$, where:
	\begin{itemize}[noitemsep]
		\item $L_j = L_{2,j}$ for all $j < u$;
		\item $L_j = L_{1,j}$ for all $j > u$;
		\item $L_u = [\gen{\reXY{u}{y_{2,u-1}~}{y_{1,u}}}]$.
	\end{itemize}
	We first prove that $C$ is dominated by $C^*$, and then that $C^*$ is contained in $\nodeconst''_{i,\Delta}$.
	For each $j \neq u$, it holds that the set $\gen{\ell \mid \ell \in L_{1,j}} \cap \gen{\ell \mid \ell \in L_{2,j}}$
	is a subset of both $\gen{\ell \mid \ell \in L_{1,j}}$ and $\gen{\ell \mid \ell \in L_{2,j}}$.
	In position $u$ of $C_1$ there is the disjunction $[\gen{\reMY{u}{y_{1,u}}}]$, and in position $u$ of $C_2$ there is the disjunction $[\gen{\reXM{u}{y_{2,u-1}}}]$. Observe that $L_u$ is a superset of both (see \Cref{fig:diagram-intermediate-present}). Also, by definition, the obtained configuration $C^*$ is generated by a strikethrough configuration with parameters $a_2$ and $b_1$, and hence it is part of $\nodeconst''_{i,\Delta}$.
	
	In the other cases, we proceed as follows. For $1 \le h \le \Delta - i - 1$, let \[T_h = (T_{h,1},T_{h,2}) = (f(C_1, C_2, h, u, 1), f(C_1, C_2, h, u, 2)).\]
	We now consider all possible cases for the values of $T_1,\ldots, T_{\Delta -i-1}$.
	
	\begin{itemize}
		\item \textbf{\boldmath{There exists an index $1 \le h \le \Delta - i - 1$ such that $T_{h,2} \neq \minus$ and $(h = \Delta - i - 1 \lor T_{h+1,1} \neq \minus)$}}. Let $a$ be the largest integer satisfying $T_a \in \{(\minus,0),(\minus,1)\}$, $a \le h$, $(u \ge h \lor u < a)$, if it exists, and $\bot$ otherwise. Let $b$ be the smallest integer satisfying $T_b \in \{(0,\minus),(1,\minus)\}$, $b \ge h+1$, $(u \le h +1 \lor u > b)$,  if it exists, and $\bot$ otherwise.
		In the following, we assume that, if there exists at least one index $h$ that satisfies the requirements and that makes $a$ and $b$ (which are defined as a function of $h$) both different from $\bot$, then $h$ is one of such values. 
		
		 If $a$ and $b$ are both different from $\bot$, consider the configuration $C^*$ generated by $L_1 \ldots L_\Delta$, where:
		\begin{itemize}
			\item $L_j = [\gen{\reMM{j}}]$ if $j < a$;
			\item $L_j = [\gen{\reXY{j}{T_{j,1~}}{T_{j,2}}}]$ if $a \le j \le b$;
			\item $L_j = [\gen{\reMM{j}}]$ if $b < j \le \Delta - i - 1$;
			\item $L_{\Delta - i} = [\gen{\reXY{\Delta - i}{0}{0},\reXY{\Delta - i}{0}{1}},\gen{\reXY{\Delta - i}{1}{0},\reXY{\Delta - i}{1}{1}}]$;
			\item $L_j = [\gen{\reQ{j}}]$ if $j > \Delta - i$.
		\end{itemize}
		We start by arguing that $u$ must satisfy $u < a$ or $u > b$.
		For a contradiction, assume otherwise.
		By the definition of $a$ and $b$, this implies $u \in \{ h, h + 1 \}$.
		If $u = h$, then the assumption  $T_{h,2} \neq \minus$ implies $T_u \neq (\minus, \minus)$.
		If $u = h + 1$, then the assumption $T_b \in \{(0,\minus),(1,\minus)\}$ implies $b \le \Delta - i - 1$ (as otherwise $T_b$ would not be defined), which, combined with the assumption $b \ge h+1$ implies $h < \Delta -i-1$,  which implies $T_{h+1,1} \neq \minus$, and hence $T_u \neq (\minus, \minus)$.
		In either case, we obtain $T_u \neq (\minus, \minus)$, contradicting the definition of $T_u$, and we conclude that  $u < a$ or $u > b$.

		This in particular implies that for all $a < h < b$ it holds that $T_{h,1}$ and $T_{h,2}$ are both different from $\minus$, and we get that $C^*$ is generated by a strikethrough configuration with parameters $a,b$ (where we use \Cref{obs:same-borders} in both arguments). Hence, $C^*$ is contained in $\nodeconst''_{i,\Delta}$.
		Moreover, $C^*$ dominates $C$, in fact:
		\begin{itemize}
			\item In all positions $j$ satisfying $j < a$ or $j > b$, $C^*$ contains all sets of color $j$. Hence, in such positions the set of $C^*$ is a superset of the set of $C$. Moreover, observe that $u$ satisfies $u < a $ or $u > b$ (as observed above).
			\item In all positions $j$ satisfying $a \le j \le b$, we apply \Cref{lem:label-domination}, obtaining that \[\gen{\gen{\reXY{j}{T_{j,1~}}{T_{j,2}}}} \supseteq \gen{\ell \mid \ell \in L_{1,j}} \cap \gen{\ell \mid \ell \in L_{2,j}} .\]
		\end{itemize}
		
		For similar reasons, if $a \neq \bot$, $u < a$, and $b = \bot$, $C$ is dominated by the configuration generated by $L_1 \ldots L_\Delta$, where:
		\begin{itemize}
			\item $L_j = [\gen{\reMM{j}}]$ if $j < a$;
			\item $L_j = [\gen{\reXY{j}{T_{j,1~}}{T_{j,2}}}]$ if $a \le j \le \Delta - i - 1$;
			\item $L_{j} = [\gen{\reXY{j}{T_{j-1,2}~}{0},\reXY{j}{T_{j-1,2}~}{1}}]$ if $j = \Delta - i$;
			\item $L_j = [\gen{\reQ{j}}]$ if $j > \Delta - i$.
		\end{itemize}
		Note that this configuration is contained in $\nodeconst''_{i,\Delta}$.
		
		Now consider the case that $a = \bot$.
		Recall that, by assumption, it holds that there exists an index $h$ such that $T_{h,2} \neq \minus$ and $(h = \Delta - i - 1  \lor T_{h+1,1} \neq \minus)$. By recursively applying \Cref{obs:same-borders} starting with $h$ and decreasing it step by step, since $a = \bot$ and since $T_{1,1} = \minus$, we get that $u < h$, and that for all $u < j \le h$, $T_{j,1}$ and $T_{j,2}$ are both different from $\minus$ (since otherwise we would have found a value for $a$).
		
		Similarly, if $b = \bot$ and $u > h$, it must also hold that for all $h +1 \le j < \min\{u,\Delta - i - 1\}$, $T_{j,1}$ and $T_{j,2}$ are both different from $\minus$. Moreover, the case $b = \bot$, $a \neq \bot$ and $u < a$ has already been handled, and, by the definition of $a$, it is not possible that $b = \bot$, $a \neq \bot$, and $a \le u \le h$ (where we use that $u \neq h$ since $T_{h,2} \neq \minus$ and $T_u = (\minus, \minus)$).
		
		Hence, we need to handle the following two cases: the case $a=\bot$ and $u < h$, and the case $b = \bot$, $a \neq \bot$, and $u > h$.
		
		Let us consider the case $a=\bot$ and $u < h$. Since $T_{u+1,1} \neq \minus$, by \Cref{lem:minus-on-one-side-of-union}, it must hold that either $\{t_{1,u,1}, t_{2,u,1}\} = \{0,1\}$ or $\{t_{1,u},t_{2,u}\} = \{(x_1,\minus),(\minus,x_2)\}$ for some $x_1,x_2 \in \{0,1\}$. The latter case is covered by the special case handled before, i.e., when $a_1 = b_2 = u$ or $a_2 = b_1 =u$.

	Hence, consider the former case, which implies that $u \neq 1$ and that $T_{u-1,2} = \minus$.
	If there exists an index $h'< u$ such that $T_{h',2} \neq \minus$, which by \Cref{obs:same-borders} implies $T_{h'+1,1} \neq \minus$, then, since $T_{u-1,2} = \minus$, we can find two indices $a \le h'$ and $b\ge h' + 1$ satisfying the requirements stated at the beginning of this case, contradicting the assumption that $a=\bot$.
		Hence, $T_j = (\minus,\minus)$ for all $j \le u$, which implies that $a_1 = a_2 < u$, because:
		\begin{itemize}
			\item if $a_1 > u$ or $a_2 > u$, then we get that the union is on comparable sets, which contradicts \Cref{obs:no-union-comparable};
			\item if $a_1 = u$ or $a_2 = u$, w.l.o.g. $a_1 = u$, then $t_{1,u,1} = t_{1,a_1,1} = \minus$, which contradicts $t_{1,u,1} \in \{0,1\}$;
			\item if $a_1 < u$, $a_2 < u$, and $a_1 \neq a_2$, w.l.o.g.\ $a_1 < a_2$, then it must hold that $T_{a_2} \neq (\minus,\minus)$, which is a contradiction.
		\end{itemize} 
		 Therefore, in position $a_1 = a_2$, since we must have $\{L_{1,a_1},L_{2,a_2}\} = \{[\gen{\reMY{a_1}{0}}],[\gen{\reMY{a_1}{1}}]\}$ (otherwise we would not have  $T_{a_1} = (\minus,\minus)$), we get $\gen{\ell \mid \ell \in L_{1,a_1}} \cap \gen{\ell \mid \ell \in L_{2,a_2}} = \gen{\gen{\reEE{a_1}}}$, by \Cref{lem:diagram-intermediate-first,lem:diagram-intermediate-present}. Hence the result is dominated by a configuration generated by a grabbing-present configuration.
					
		Let us now consider the case $b = \bot$, $a \neq \bot$, and $u > h$. If $u \le \Delta -i-1$ and $T_{\Delta - i - 1, 2} = \minus$, we consider two cases:
		\begin{itemize}
			\item $t_{1,\Delta-i-1,2} = t_{2,\Delta-i-1,2} = \minus$, i.e., the output bit  in position $\Delta-i-1$ is $\minus$ in both configurations. In this case, the claim holds for symmetric reasons as in the previous case (i.e., the case $a = \bot $ and $u < h$).
			\item $\{t_{1,\Delta-i-1,2} = t_{2,\Delta-i-1,2} \} = \{0,1\}$. In this case, the output is dominated by the configuration generated by the grabbing-special configuration, because $\{ L_{1,\Delta-i}, L_{2,\Delta-i} \} = \{ [\gen{\reXY{\Delta-i}{0}{0},\reXY{\Delta-i}{0}{1}}], [\gen{\reXY{\Delta-i}{1}{0},\reXY{\Delta-i}{1}{1}}] \}$.
		\end{itemize}

		We now consider the case that $u \le \Delta -i-1$ and $T_{\Delta - i - 1, 2} \neq \minus$.
		Since $b = \bot$ and $u > h$, we obtain that $T_{u-1, 2} \neq \minus$.
		By \Cref{lem:minus-on-one-side-of-union}, this implies that $\{t_{1,u,2}, t_{2,u,2}\} = \{0,1\}$, or $\{t_{1,u},t_{2,u}\} = \{(x_1,\minus),(\minus,x_2)\}$ for some $x_1,x_2 \in \{0,1\}$. The latter case is covered by the special case handled before, i.e., when $a_1 = b_2 = u$ or $a_2 = b_1 =u$. In the former case, we obtain that $u < \Delta - i - 1$ (because $T_{\Delta - i - 1, 2} \neq \minus$) and that $T_{u+1,1} = \minus$. We get that we could have picked a different value of $h$ that would have made $a \neq \bot$ and $u < a$, which is a case previously covered.
For the case $u = \Delta - i$, by applying \Cref{obs:no-union-comparable}, we get that $\{ L_{1,\Delta-i}, L_{2,\Delta-i} \} = \{ [\gen{\reXY{\Delta-i}{0}{0},\reXY{\Delta-i}{0}{1}}], [\gen{\reXY{\Delta-i}{1}{0},\reXY{\Delta-i}{1}{1}}] \}$ and hence that $\{t_{1,u-1,2}, t_{2,u-1,2}\} = \{0,1\}$. This implies that $T_{u-1,2} = \minus$, and hence $b$ cannot be $\bot$, a contradiction.
Finally, we can exclude the case $u > \Delta-i$ by \Cref{obs:no-union-comparable}.
		
		\item \textbf{\boldmath{For all $1 \le h \le \Delta - i - 1$, $T_h = (\minus,\minus)$}}. We consider positions $a_1$ and $b_1$. Recall that $a_1 < b_1$. Either $a_1 \neq u$ or $b_1 \neq u$. In the former case, in order to obtain $T_h = (\minus,\minus)$ it must hold that $\{t_{1,a_1}, t_{2,a_1}\} \in \{(\minus,0),(\minus,1)\}$ and that $t_{1,a_1} \neq t_{2,a_1}$. In this case, the result is dominated by the configuration generated by a grabbing-present configuration (i.e., the one with label $[\gen{\reEE{a_1}}]$). In the latter case, we distinguish two cases. If $b_1 < \Delta - i$, it must hold that $\{t_{1,b_1}, t_{2,b_1}\} \in \{(0,\minus),(1,\minus)\}$ and that $t_{1,b_1} \neq t_{2,b_1}$. As before,  the result is dominated by the configuration generated by a grabbing-present configuration (i.e, the one with label $[\gen{\reEE{b_1}}]$).
		If $b_1 = \Delta - i$, since it holds that $T_{\Delta - i - 1, 2} = \minus$, we get that $\{ L_{1,\Delta-i}, L_{2,\Delta-i} \} = \{ [\gen{\reXY{\Delta-i}{0}{0},\reXY{\Delta-i}{0}{1}}], [\gen{\reXY{\Delta-i}{1}{0},\reXY{\Delta-i}{1}{1}}] \}$, and hence the result is dominated by the configuration generated by the grabbing-special configuration.
		
		\item \textbf{\boldmath{There exists an index $1 \le h \le \Delta - i - 1$ such that $T_{h} = (\minus,b)$ and $T_{h+1} = (\minus,\minus)$, or such that $T_{h} = (b,\minus)$ and $T_{h-1} = (\minus,\minus)$, for some $b \in \{0,1\}$, and the previous cases do not apply}}.
		We consider positions $a_1$ and $b_1$. Recall that $a_1 < b_1$. Either $a_1 \neq u$ or $b_1 \neq u$. In the former case, if $\{t_{1,a_1}, t_{2,a_1}\} \in \{(\minus,0),(\minus,1)\}$ and $t_{1,a_1} \neq t_{2,a_1}$, the result is dominated by the configuration generated by a grabbing-present configuration. If $t_{1,a_1} =  t_{2,a_1}$, w.l.o.g.\ assume $t_{1,a_1} = t_{2,a_1} = (\minus,0)$, by \Cref{obs:same-borders}, in order to have $T_{a_1+1} = (\minus,\minus)$, it must hold that $u = a_1 + 1$. By \Cref{obs:no-union-comparable}, $L_{1,u}$ and $L_{2,u}$ must be incomparable, which implies that $\{t_{1,u}, t_{2,u}\} \in \{(0,0),(0,1)\}$. If $b_1 = u+1$, we get that $\{t_{1,b_1}, t_{2,b_1}\} \in \{(0,\minus),(1,\minus)\}$, and hence that the result is dominated by the configuration generated by the grabbing-special configuration. If $b_1 > u + 1$ and $\{t_{1,b_1}, t_{2,b_1}\} \in \{(0,\minus),(1,\minus)\}$, we again get that the result is dominated by the configuration generated by the grabbing-special configuration. The remaining case to handle is the one where $b_1 > u+1$ and $t_{1,b_1} = t_{2,b_1}$. We get that $T_{b_1 -1,2} \neq \minus$ and that $T_{b_1,1} \neq \minus$, which contradicts the assumption that previous cases do not apply.
		Finally, the case $b_1 \neq u$ is symmetric to the case $a_1 \neq u$, which we have already handled.
		
		\item \textbf{\boldmath{There exists an index $1 \le h \le \Delta - i - 1$ such that $T_{h} = (b_1,b_2)$ such that $b_h,b_j \in \{0,1\}$ and  for all $j \neq h$, it holds that $T_j = (\minus,\minus)$}}. This case cannot happen by \Cref{obs:same-borders}.
	\end{itemize}
	We obtained that, in all cases, $C$ is dominated by some configuration in $\nodeconst''_{i,\Delta}$, which is a contradiction.
\end{proof}

\paragraph{Black constraint.}
We define the black constraint $\edgeconst''_{i,\Delta}$ as follows. Let $\Sigma''_{i,\Delta}$ be the set of labels that appear in $\nodeconst''_{i,\Delta}$. The black constraint $\edgeconst''_{i,\Delta}$ contains all the configurations described by the set of condensed configurations given by the following process. Take a configuration $C^*$ in $\edgeconst'_{i,\Delta}$, and replace each label $L$ with the disjunction containing the labels of $\Sigma''_{i,\Delta}$ that are supersets of (or equal to) $\gen{L}$. 
By the definition of $\edgeconst''_{i,\Delta}$, we obtain the following.
\begin{lemma}\label{lem:step22}	
		$\edgeconst''_{i,\Delta}$ is the set of all configurations $S_1 ~ S_2 ~ S_3$ of sets contained in $\Sigma''_{i,\Delta}$ such that there exists some tuple  $(\ell_1, \ell_2, \ell_3) \in S_1 \times S_2 \times S_3$ satisfying that $\ell_1~\ell_2~\ell_3 $ is contained in $\edgeconst'_{i,\Delta}$.
\end{lemma}
By combining \Cref{lem:step21} and \Cref{lem:step22}, by the definition of relaxation, we obtain the following.
\begin{lemma}\label{lem:all-step-2}
	The problem $\Pi''_{i,\Delta}$ is a relaxation of $\rere(\Pi'_{i,\Delta})$.
\end{lemma} 

\paragraph{More explicit black constraint.}
While for the white constraint we have an explicit definition (given by \Cref{obs:from-nodeplus}), the definition of $\edgeconst''_{i,\Delta}$ is given implicitly. We now provide an explicit definition of $\edgeconst''_{i,\Delta}$.
First, we explicitly list the labels appearing in $\nodeconst''_{i,\Delta}$, and hence the labels in $\Sigma''_{i,\Delta}$.
\begin{itemize}
	\raggedright
	\item For color $1$, there are the following labels: $\gen{\gen{\reEE{1}}}$, $\gen{\gen{\reMM{1}}}$, $\gen{\gen{\reMY{1}{0}}}$, $\gen{\gen{\reMY{1}{1}}}$.
	\item For each present color $j$ satisfying $2 \le j \le \Delta - i - 1$, there are the following labels: $\gen{\gen{\reEE{j}}}$, $\gen{\gen{\reMM{j}}}$, $\gen{\gen{\reMY{j}{0}}}$, $\gen{\gen{\reMY{j}{1}}}$, $\gen{\gen{\reXM{j}{0}}}$, $\gen{\gen{\reXM{j}{1}}}$, $\gen{\gen{\reXY{j}{0}{0}}}$, $\gen{\gen{\reXY{j}{0}{1}}}$, $\gen{\gen{\reXY{j}{1}{0}}}$, $\gen{\gen{\reXY{j}{1}{1}}}$.
	\item For each gone color $j$, there are the following labels: $\gen{\gen{\reQ{j}}}$, $\gen{\gen{\reE{j}}}$.
	\item For the special color $j = \Delta - i$, there are the following labels: $\gen{\gen{\reXY{\Delta - i}{0}{0},\reXY{\Delta - i}{0}{1}},\gen{\reXY{\Delta - i}{1}{0},\reXY{\Delta - i}{1}{1}}}$, $\gen{\gen{\reXY{\Delta - i}{0}{0},\reXY{\Delta - i}{0}{1}}}$, $\gen{\gen{\reXY{\Delta - i}{1}{0},\reXY{\Delta - i}{1}{1}}}$, $\gen{\gen{\reMY{\Delta-i}{0},\reMY{\Delta-i}{1}}}$.
\end{itemize}

Now, by \Cref{obs:diagram-by-set-inclusion} (see \Cref{fig:diagram-intermediate-first,fig:diagram-intermediate-gone,fig:diagram-intermediate-present,fig:diagram-intermediate-special}), $\edgeconst''_{i,\Delta}$ can be obtained as follows. Take each configuration $C^*$ in $\edgeconst'_{i,\Delta}$, and replace labels as follows.
For gone colors $j$:
\begin{itemize}
	\item $\gen{\reE{j}} \longmapsto [\gen{\gen{\reQ{j}}}, \gen{\gen{\reE{j}}}]$;
	\item $\gen{\reQ{j}} \longmapsto [\gen{\gen{\reQ{j}}}]$.
\end{itemize}
For color $1$:
\begin{itemize}
	\item $\gen{\reEE{1}} \longmapsto  [\gen{\gen{\reEE{1}}}, \gen{\gen{\reMM{1}}}, \gen{\gen{\reMY{1}{0}}}, \gen{\gen{\reMY{1}{1}}}]$;
	\item $\gen{\reMM{1}} \longmapsto [\gen{\gen{\reMM{1}}}]$;
	\item For $b \in \{0,1\}$, $\gen{\reMY{1}{b}}  \longmapsto [\gen{\gen{\reMY{1}{b}}}, \gen{\gen{\reMM{1}}}]$.
\end{itemize}
For each present color $j$ satisfying $2 \le j \le \Delta - i - 1$:
\begin{itemize}
	\raggedright
	\item For $b_1 \in \{0,1\}$, $b_2 \in \{0,1\}$, $\gen{\reXY{j}{b_1~}{b_2}} \longmapsto [\gen{\gen{\reMM{j}}},\gen{\gen{\reXY{j}{b_1~}{b_2}}}]$;
	\item For $b \in \{0,1\}$, $\gen{\reXM{j}{b}} \longmapsto [\gen{\gen{\reMM{j}}},\gen{\gen{\reXY{j}{b~}{0}}},\gen{\gen{\reXY{j}{b~}{1}}},\gen{\gen{\reXM{j}{b}}}]$;
	\item For $b \in \{0,1\}$, $\gen{\reMY{j}{b}} \longmapsto [\gen{\gen{\reMM{j}}},\gen{\gen{\reXY{j}{0~}{b}}},\gen{\gen{\reXY{j}{1~}{b}}},\gen{\gen{\reMY{j}{b}}}]$;
	\item $\gen{\reXY{j}{0}{1},\reXY{j}{1}{0}} \longmapsto [\gen{\gen{\reXY{j}{0}{1}}},\gen{\gen{\reXY{j}{1}{0}}},\gen{\gen{\reMM{j}}}]$;
	\item $\gen{\reXY{j}{0}{0},\reXY{j}{1}{1}} \longmapsto [\gen{\gen{\reXY{j}{0}{0}}},\gen{\gen{\reXY{j}{1}{1}}},\gen{\gen{\reMM{j}}}]$;
	\item For $b_1 \in \{0,1\}$, $b_2 \in \{0,1\}$, $\gen{\reXM{j}{b_1},\reMY{j}{b_2}} \longmapsto [\gen{\gen{\reMM{j}}}, \allowbreak \gen{\gen{\reXM{j}{b_1}}},\gen{\gen{\reMY{j}{b_2}}}, \allowbreak
	\gen{\gen{\reXY{j}{b_1~}{0}}}, \allowbreak \gen{\gen{\reXY{j}{b_1~}{1}}}, \allowbreak \gen{\gen{\reXY{j}{0~}{b_2}}},  \allowbreak \gen{\gen{\reXY{j}{1~}{b_2}}} 
	]$;
	\item $\gen{\reEE{j}} \longmapsto  [\gen{\gen{\reEE{j}}}, \allowbreak\gen{\gen{\reMM{j}}}, \allowbreak\gen{\gen{\reMY{j}{0}}}, \allowbreak\gen{\gen{\reMY{j}{1}}}, \allowbreak\gen{\gen{\reXM{j}{0}}}, \allowbreak\gen{\gen{\reXM{j}{1}}}, \allowbreak\gen{\gen{\reXY{j}{0}{0}}}, \allowbreak\gen{\gen{\reXY{j}{0}{1}}},\allowbreak \gen{\gen{\reXY{j}{1}{0}}}, \allowbreak\gen{\gen{\reXY{j}{1}{1}}}]$.
\end{itemize}
For the special color $j = \Delta - i$:
\begin{itemize}
	\raggedright
	\item $\gen{\reXY{j}{0}{0},\reXY{j}{0}{1}} \longmapsto [\gen{\gen{\reXY{j}{0}{0},\reXY{j}{0}{1}}}, \gen{\gen{\reXY{\Delta - i}{0}{0},\reXY{\Delta - i}{0}{1}},\gen{\reXY{\Delta - i}{1}{0},\reXY{\Delta - i}{1}{1}}}]$;
	\item $\gen{\reXY{j}{1}{0},\reXY{j}{1}{1}} \longmapsto [\gen{\gen{\reXY{j}{1}{0},\reXY{j}{1}{1}}}, \gen{\gen{\reXY{\Delta - i}{0}{0},\reXY{\Delta - i}{0}{1}},\gen{\reXY{\Delta - i}{1}{0},\reXY{\Delta - i}{1}{1}}}]$;
	\item $\gen{\reMY{j}{0},\reMY{j}{1}} \longmapsto [\gen{\gen{\reXY{\Delta - i}{0}{0},\reXY{\Delta - i}{0}{1}}}, \allowbreak \gen{\gen{\reXY{\Delta - i}{1}{0},\reXY{\Delta - i}{1}{1}}}, \allowbreak \gen{\gen{\reMY{\Delta-i}{0},\reMY{\Delta-i}{1}}}, \allowbreak\gen{\gen{\reXY{\Delta - i}{0}{0},\reXY{\Delta - i}{0}{1}},\gen{\reXY{\Delta - i}{1}{0},\reXY{\Delta - i}{1}{1}}}]$.
\end{itemize}
Let $B$ be the defined mapping. We summarize the above observation in the following statement.
\begin{observation}\label{obs:replacement-black}
	The black constraint $\edgeconst''_{i,\Delta}$ can be obtained as follows. Take each configuration $C^*$ in $\edgeconst'_{i,\Delta}$, and replace labels using the mapping $B$.
\end{observation}

\paragraph{Relation between $\Pi''_{i,\Delta}$ and $\Pi_{i+1,\Delta}$.}
We define additional problems $\Pi'''_{i,\Delta}$ and $\Pi''''_{i,\Delta}$, which are, like $\Pi''_{i,\Delta}$, also going to be relaxations of $\rere(\Pi'_{i,\Delta})$,  and then we prove that $\Pi''''_{i,\Delta}$ is equivalent to $\Pi_{i+1,\Delta}$. Let $j$ be the special color. The problem $\Pi'''_{i,\Delta}$ is obtained by replacing each instance of $\gen{\gen{\reXY{j}{0}{0},\reXY{j}{0}{1}}}$, and each instance of $\gen{\gen{\reXY{j}{1}{0},\reXY{j}{1}{1}}}$, with the label $\gen{\gen{\reXY{j}{0}{0},\reXY{j}{0}{1}},\gen{\reXY{j}{1}{0},\reXY{j}{1}{1}}}$, in both the white and the black constraint of  $\Pi''_{i,\Delta}$. 

\begin{observation}\label{obs:no-suffix-at-new-end}
	In the white constraint of $\Pi'''_{i,\Delta}$, all configurations that, for $j = \Delta - i - 1$, contain the label $\gen{\gen{\reXM{j}{0}}}$ or $\gen{\gen{\reXM{j}{1}}}$ are non-maximal.
\end{observation}
\begin{proof}
	Consider the definition of $\nodeconst''_{i,\Delta}$ given by \Cref{obs:from-nodeplus}, and the replacement rules used to define $\Pi'''_{i,\Delta}$ as a function of $\Pi''_{i,\Delta}$. The proof follows by observing the following:
	\begin{itemize}
		\item Label $\gen{\gen{\reXM{j}{0}}}$ is a subset of both $\gen{\gen{\reXY{j}{0}{0}}}$ and $\gen{\gen{\reXY{j}{0}{1}}}$;
		\item Label $\gen{\gen{\reXM{j}{1}}}$ is a subset of both $\gen{\gen{\reXY{j}{1}{0}}}$ and $\gen{\gen{\reXY{j}{1}{1}}}$;
		\item Given a configuration $C$ that in position $j$ has a label in $\{\gen{\gen{\reXM{j}{0}}},\gen{\gen{\reXM{j}{1}}}\}$, there is also a configuration that is equal to $C$ in all positions except $j$, and in position $j$ contains a label from $\{\gen{\gen{\reXY{j}{0}{0}}},\gen{\gen{\reXY{j}{0}{1}}},\gen{\gen{\reXY{j}{1}{0}}},\gen{\gen{\reXY{j}{1}{1}}}\}$. \qedhere
	\end{itemize}
\end{proof}
Let $\Pi''''_{i,\Delta}$ be the problem obtained by removing from the white constraint of $\Pi'''_{i,\Delta}$ all configurations that, for $j = \Delta - i - 1$, contain the label $\gen{\gen{\reXM{j}{0}}}$ or $\gen{\gen{\reXM{j}{1}}}$.
By the definitions of $\Pi'''_{i,\Delta}$ and $\Pi''''_{i,\Delta}$, by \Cref{obs:no-suffix-at-new-end}, and by \Cref{lem:all-step-2}, we obtain the following.
\begin{lemma}\label{lem:all-step-22}
	The problem $\Pi''''_{i,\Delta}$ is a relaxation of $\rere(\Pi'_{i,\Delta})$.
\end{lemma} 

We now define a renaming $N$ of the labels of $\nodeconst''''_{i,\Delta}$. Let $N$ be the following mapping.
\begin{itemize}[noitemsep]
	\item If $j$ is not special, let $N(\gen{\gen{L}}) = L$.
	\item If $j$ is special, let:
	\begin{itemize}
		\item $N(\gen{\gen{\reXY{j}{0}{0},\reXY{j}{0}{1}},\gen{\reXY{j}{1}{0},\reXY{j}{1}{1}}}) = \reQ{j}$;
		\item $N(\gen{\gen{\reMY{j}{0},\reMY{j}{1}}}) = \reE{j}$.
	\end{itemize}
\end{itemize}
By combining the definition of $\nodeconst''_{i,\Delta}$ given by \Cref{obs:from-nodeplus}, \Cref{obs:replacement-black}, the replacement rules used to define $\Pi'''_{i,\Delta}$ as a function of $\Pi''_{i,\Delta}$, and the definition of $\nodeconst''''_{i,\Delta}$ as a function of $\Pi'''_{i,\Delta}$, we obtain the following.
\begin{observation}\label{obs:ren1}
	Under the renaming $N$, the white constraint $\nodeconst''''_{i,\Delta}$ of $\Pi''''_{i,\Delta}$ is equal to $\nodeconst_{i+1,\Delta}$.
\end{observation}
\begin{observation}\label{obs:ren2}
	Under the renaming $N$, the black constraint $\edgeconst''''_{i,\Delta}$ of $\Pi''''_{i,\Delta}$ is equal to $\edgeconst_{i+1,\Delta}$.
\end{observation}

By combining \Cref{lem:all-step-1}, \Cref{lem:all-step-22}, \Cref{obs:ren1}, and \Cref{obs:ren2}, we obtain the following.
\begin{lemma}\label{lem:onestep}
	For all $0 \le i \le \Delta - 3$, there exists a problem $\Pi^{\mathrm{rel}}_{i,\Delta}$ satisfying the following:
	\begin{itemize}[noitemsep]
		\item $\re(\Pi_{i,\Delta})$ can be relaxed to $\Pi^{\mathrm{rel}}_{i,\Delta}$;
		\item $\rere(\Pi^{\mathrm{rel}}_{i,\Delta})$ can be relaxed to $\Pi_{i+1,\Delta}$;
		\item The number of labels of $\Pi^{\mathrm{rel}}_{i,\Delta}$ are upper bounded by $O(\Delta)$.
	\end{itemize}
\end{lemma}

\subsection{The problem \texorpdfstring{$\Pi_{\Delta-2,\Delta}$}{Pi} is not trivial}\label{ssec:re-non-trivial}
We now prove that $\Pi_{\Delta-2,\Delta}$ is not trivial to solve, even when given a $\Delta$-edge coloring.
\begin{lemma}\label{lem:non-zero}
	The problem $\Pi_{\Delta-2,\Delta}$ is not $0$-round solvable in the deterministic port numbering model, even if a $\Delta$-edge coloring is given.
\end{lemma}
\begin{proof}
	Assume for a contradiction that there exists a deterministic $0$-round algorithm solving $\Pi_{\Delta-2,\Delta}$ when given a $\Delta$-edge coloring.
	Any $0$-round algorithm (operating on white nodes) in the deterministic port numbering model satisfies that all white nodes output the same configuration $C = \{L_1,\ldots,L_\Delta\}\in \nodeconst_{\Delta-2,\Delta}$. Moreover, there must exist a permutation $\phi$ such that all nodes output label $L_i$ on port $\phi(i)$.
	
	By the definition of $\nodeconst_{\Delta-2,\Delta}$, the configuration $C$ must contain at least one label among the following:
	\begin{itemize}[noitemsep]
		\item $\reEE{1}$;
		\item $\reE{j}$, for some $3 \le j \le \Delta$;
		\item $\reMY{j}{0}$, for some $1 \le j \le 2$;
		\item $\reMY{j}{1}$, for some $1 \le j \le 2$.
	\end{itemize}
	Let $L_j$ be such a label, and let $i = \phi(j)$ be the port on which the algorithm outputs $L_j$.
	
	Observe that, for any possible $L_j$ listed above, the configuration $L_j ~ L_j ~ L_j$ is not contained in $\edgeconst_{\Delta-2,\Delta}$. However, by considering a graph in which there are three white nodes $u,v,w$ connected via port $i$ to the same black node $b$ (satisfying that the $\Delta$-edge coloring assigns color $i$ to node $b$), we obtain a setting in which the algorithm fails to produce a valid solution for $\Pi_{\Delta-2,\Delta}$.
\end{proof}

\subsection{Relation between \texorpdfstring{$\mathcal{P}_\Delta$}{P} and \texorpdfstring{$\Pi_{0,\Delta}$}{Pi}}\label{ssec:re-start}
\begin{lemma}\label{lem:mapping-natural-first}
	Given a solution for $\mathcal{P}_\Delta$, it is possible to solve $\Pi_{0,\Delta}$ in $0$ rounds.
\end{lemma}
\begin{proof}
	It is easy to see that $\nodeconst_{\mathcal{P}_\Delta} \subsetneq \nodeconst_{0,\Delta}$ and that $\edgeconst_{\mathcal{P}_\Delta} \subsetneq \edgeconst_{0,\Delta}$. Hence, any solution for $\mathcal{P}_\Delta$ is also a solution for $\Pi_{0,\Delta}$.
\end{proof}

\subsection{Putting things together}\label{ssec:re-lift}
By applying \Cref{lem:onestep} multiple times, and then applying \Cref{lem:non-zero}, we obtain the following.
\begin{lemma}\label{lem:manysteps}
	Let $\Delta \ge 3$ be an integer. There exists a problem sequence $\Pi_{0,\Delta} \rightarrow \Pi_{1,\Delta} \rightarrow \ldots \Pi_{\Delta-2,\Delta}$ such that, for all $0 \le i < \Delta-2$, the following holds:
	\begin{itemize}[noitemsep]
		\item There exists a problem $\Pi'_{i,\Delta}$ that is a relaxation of $\re(\Pi_{i,\Delta})$;
		\item $\Pi_{i+1,\Delta}$ is a relaxation of $\rere(\Pi'_{i,\Delta})$;
		\item The number of labels of $\Pi_{i,\Delta}$, and the ones of $\Pi'_{i,\Delta}$, are upper bounded by $O(\Delta)$.
	\end{itemize}
	Also, $\Pi_{\Delta-2,\Delta}$ is not $0$-round solvable in the deterministic port numbering model, when given a $\Delta$-edge coloring.
\end{lemma}

By combining \Cref{lem:manysteps} with \Cref{thm:lifting}, we obtain the following.
\begin{lemma}\label{lem:lb-first}
	Let $\Delta \ge 3$ be an integer. The problem $\Pi_{0,\Delta}$ requires  $\Omega(\min\{\Delta, \log_\Delta n\})$ rounds in the deterministic LOCAL model and $\Omega(\min\{\Delta, \log_\Delta \log n\})$ rounds in the randomized LOCAL model, even if a $\Delta$-edge coloring is given.
\end{lemma}

By combining \Cref{lem:lb-first} with \Cref{lem:mapping-natural-first}, we obtain \Cref{thm:main-lb}.
\newcommand{\given}{\ \middle| \ }

\section{Networks of non-signaling games}
\label{sec:networks}

In this section, we show a classical upper bound for all networks of non-signaling games.
We start by defining the notion of a game.
Each player of a game receives a local input and must produce a local output
such that the outputs of all players combined are consistent with their inputs.

\begin{definition}[Game]
  Let $\Sigma$ be a finite set, and let $m \in \N_+$ be the number of players.
  We call $\mathfrak{g} \subseteq \Sigma^m \times \Sigma^m$ a \emph{game}.
  Each player $i$ receives one input $x_i \in \Sigma$ and produces one output
  $y_i \in \Sigma$.
  A \emph{move} $\mu = (x,y) \in \Sigma^m \times \Sigma^m$ is valid if $\mu \in
  \mathfrak{g}$.
  We overload the notation so that $\mathfrak{g}(x) = \{ y \in \Sigma^m \mid (x, y) \in \mathfrak{g} \}$.
  We say $\mathfrak{g}$ is \emph{solvable} if, for every $x$, $\mathfrak{g}(x)$
  is non-empty.
\end{definition}

Note that our games last for a single move.
(Later on we will connect a player to multiple games and hence each player will
be part of moves in multiple games.)
Since games that are not solvable cannot be solved by any strategy at all, in
what follows we assume that we are dealing exclusively with solvable games.

Let us now move to the setting where a player takes part in multiple games.
A game's inputs may depend arbitrarily (but in a computable way) on the outputs
of games that were played before.
Formally we model this using circuits.

\begin{definition}[Circuit of half-games]
  Let $d \in \N_+$, and let $D$ be a circuit that satisfies the following
  properties:
  \begin{itemize}[noitemsep]
    \item $D$ has a single input gate $\xi$ and no output gates.
    \item Every wire in $D$ carries a value in $\Sigma$.
    \item $D$ may contain arbitrary gates over $\Sigma$ (of unbounded fan-in
    and fan-out).
    \item $D$ contains $d$ unary gates $H_1,\dots,H_d$, which we call
    \emph{half-games}.
  \end{itemize}
  A circuit $C$ is a \emph{(degree-$d$) circuit of half-games} if it can be
  obtained from such a $D$ by detaching all the $H_i$, that is, by replacing
  every $H_i$ with a pair of input and output gates $x_i$ and $y_i$ (that are
  connected to the rest of the circuit in the same way the input and output of
  $H_i$ are).
  In a sense, $C$ is just $D$ where one can externally examine the inputs~$x_1 \cdots x_d$ to half-games $H_1, \dots, H_d$, and the outputs of half-games are fixed at $y_1 \cdots y_d$.
  As a result, $C$ has $d+1$ inputs $\xi,y_1,\dots,y_d \in \Sigma$ and $d$
  outputs $x_1,\dots,x_d \in \Sigma$.
  Equivalently, we may also view $C$ as having $2$ inputs, $\xi \in \Sigma$ and
  $y = y_1 \cdots y_d \in \Sigma^d$, and a single output $x = x_1 \cdots x_d \in
  \Sigma^d$.
  See \cref{fig:half-game-circuit} for an example.
\end{definition}

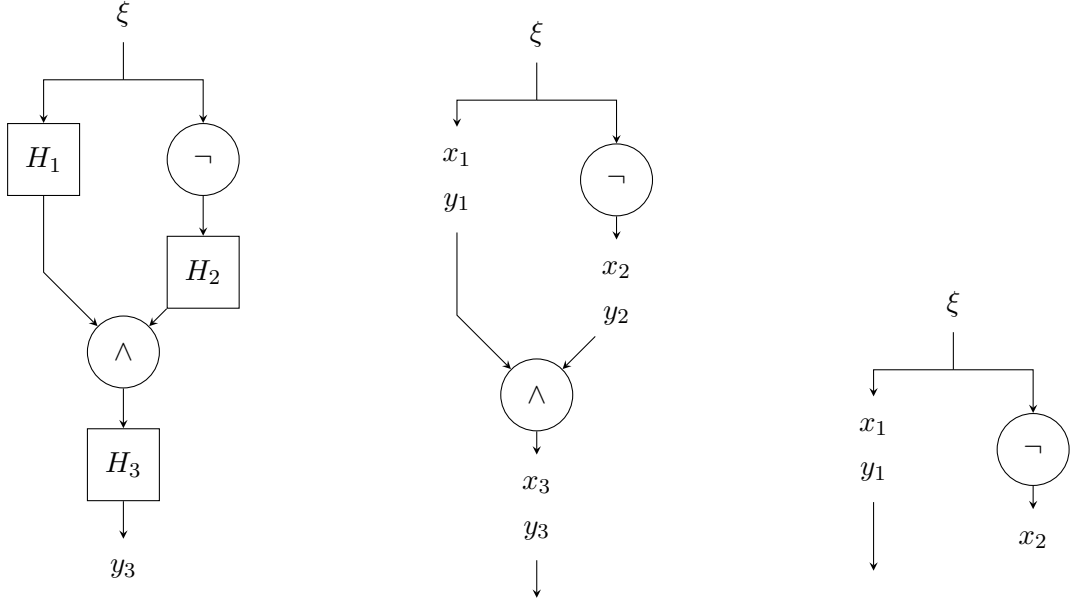
\begin{figure}
  \centering
  \begin{subfigure}[t]{.32\textwidth}
    \centering
    \begin{tikzpicture}[
    every node/.style={draw,circle,minimum size=2.5em},
    io/.style={draw=none,minimum size=0.5em},
    node distance=1.5cm
  ]
  \node[io] (y0) {$\xi$};
  \coordinate[below=0.5cm of y0] (y0s) {};
  \node[rectangle] (H1) [below left of=y0s] {$H_1$};
  \node (not) [below right of=y0s] {$\lnot$};
  \node[rectangle] (H2) [below of=not] {$H_2$};
  \node (and) [below left of=H2] {$\land$};
  \node[rectangle] (H3) [below of=and] {$H_3$};
  \node[io,below=0.5cm of H3] (y3) {$y_3$};

  \draw[] (y0) -- (y0s);
  \draw[->] (y0s) -| (H1);
  \draw[->] (y0s) -| (not);
  \draw[->] (not) -- (H2);
  \draw[->] (H1) -- (H1 |- H2) -- (and);
  \draw[->] (H2) -- (and);
  \draw[->] (and) -- (H3);
  \draw[->] (H3) -- (y3);
\end{tikzpicture}
  \end{subfigure}%
  ~
  \begin{subfigure}[t]{.32\textwidth}
    \centering
    \begin{tikzpicture}[
    every node/.style={draw,circle,minimum size=2.5em},
    io/.style={draw=none,minimum size=0.5em},
    node distance=1.5cm
  ]
  \node[io] (y0) {$\xi$};
  \coordinate[below=0.5cm of y0] (y0s) {};
  \node[io] (x1) [below left of=y0s,yshift=3mm] {$x_1$};
  \node[io] (y1) [below left of=y0s,yshift=-3mm] {$y_1$};
  \node (not) [below right of=y0s] {$\lnot$};
  \node[io] (x2) [below of=not,yshift=3mm] {$x_2$};
  \node[io] (y2) [below of=not,yshift=-3mm] {$y_2$};
  \node (and) [below left of=y2,yshift=0mm] {$\land$};
  \node[io] (x3) [below of=and,yshift=3mm] {$x_3$};
  \node[io] (y3) [below of=and,yshift=-3mm] {$y_3$};
  \coordinate[below=0.5cm of y3] (y3t)  {};

  \draw[] (y0) -- (y0s);
  \draw[->] (y0s) -| (x1);
  \draw[->] (y0s) -| (not);
  \draw[->] (not) -- (x2);
  \draw[->] (y1) -- (y1 |- y2) -- (and);
  \draw[->] (y2) -- (and);
  \draw[->] (and) -- (x3);
  \draw[->] (y3) -- (y3t);
\end{tikzpicture}
  \end{subfigure}%
  ~
  \begin{subfigure}[t]{.32\textwidth}
    \centering
    \begin{tikzpicture}[
    every node/.style={draw,circle,minimum size=2.5em},
    io/.style={draw=none,minimum size=0.5em},
    node distance=1.5cm
  ]
  \node[io] (y0) {$\xi$};
  \coordinate[below=0.5cm of y0] (y0s) {};
  \node[io] (x1) [below left of=y0s,yshift=3mm] {$x_1$};
  \node[io] (y1) [below left of=y0s,yshift=-3mm] {$y_1$};
  \node[io] (y1t) [below of=y1] {};
  \node (not) [below right of=y0s] {$\lnot$};
  \node[io] (x2) [below of=not,yshift=3mm] {$x_2$};

  \draw[] (y0) -- (y0s);
  \draw[->] (y0s) -| (x1);
  \draw[->] (y0s) -| (not);
  \draw[->] (not) -- (x2);
  \draw[->] (y1) -- (y1t);
\end{tikzpicture}
  \end{subfigure}
  \caption{Transforming $D$ (on the left) into $C$ (in the middle) by detaching
  all half-games.
  Note how $C$ may have multiple components and not use every one of its
  inputs.
  On the right, the prefix~$C^2$ is presented.
  }
  \label{fig:half-game-circuit}
\end{figure}

The intuition is that, in addition to an external input $\xi$, $D$ is
connecting together $d$ pairs $(x_i,y_i)$ that each correspond to a half-game
$H_i$.
Note that the roles of $x_i$ and $y_i$ as input and output from the perspective
of $D$, in contrast to that of the half-game $H_i$, are reversed:
How $y_i$ is obtained from $x_i$ is extrinsic to $D$---indeed, it is
determined by $H_i$---and hence what $D$ does is describe how $x_i$ (which is
the input of $H_i$) is obtained from the external input $\xi$ and the $y_j$
(which are the outputs of the other half-games).

With this in place, we now turn to the definition of the network of games
problem.
Each player is modeled as a circuit of half-games of degree $d$.
The input on white nodes describe what the circuit looks like for each white
node and what is their local input.
The input on black nodes specify which game is played with its neighbors.
The output on each edge is the respective player's input and output pair to the
game across that edge.

\begin{definition}[Network of games problem]
  Fix the number of players~$m \in \N_+$ and the number of games $d \in \N_+$
  each player plays.
  Let $\mathcal{G} = \{ \mathfrak{g}_i \subseteq \Sigma^m \times \Sigma^m \}_i$
  be a family of games on $\Sigma$.

  The \emph{network of games problem} $\GAMES_d[\mathcal{G}]$ is an LCL problem defined on $(d, m)$-regular graphs as follows:
  \begin{itemize}
    \item Each edge must be labeled with a pair from $\Sigma \times \Sigma$.
    \item Each white node $w$ is given an input element $\xi \in \Sigma$, a
    circuit of half-games $C$, and a permutation $\sigma_w$ of its neighboring
    edges.
    The node satisfies the white constraint if the output labels
    $(x_1,y_1),\dots,(x_d,y_d) \in \Sigma\times\Sigma$ around it satisfy
    $C(\xi,y^{\sigma_w}) = x^{\sigma_w}$.
    This ensures that the outputs of games flow correctly through the circuit to following games.
    \item Each black node is given as input a game $\mathfrak{g} \in \mathcal{G}$ and a
    permutation $\sigma_b$ of its neighboring edges.
    The node satisfies the black constraint if the output labels
    $(x_1,y_1),\dots,(x_m,y_m) \in \Sigma\times\Sigma$ around it satisfy
    $y^{\sigma_b} \in \mathfrak{g}(x^{\sigma_b})$.
    This ensures that each game is played correctly.
  \end{itemize}
\end{definition}

Note that, for simplicity, the problem above is defined with inputs on nodes.
To make this an LCL problem in the sense of \cref{def:black-white-lcl}, we can
move the node input labels to the adjacent edges and require that all the
inputs on the adjacent edges agree (and allow a node's output to be arbitrary if
this is not the case).
In particular, the permutation of neighboring edges is most naturally encoded by having the neighboring edges carry their indices in the permutation.
Moreover, note that both the circuits~$C$ and the games~$\mathfrak{g}$ can be seen as tables over a finite alphabet, and hence there are only finitely many distinct circuits and games.
This ensures that the input and output label sets are finite, as required by \cref{def:black-white-lcl}.

We note here that the iterated GHZ problem, as defined in \cref{sec:iterated-ghz}, is a relaxation of a network of games problem, namely $\GAMES_\Delta[\{\SYMM, \GHZ\}]$.
Here~SYMM is the symmetry-breaking game, where exactly one of the players must output~$1$ and the other two output~$0$, regardless of the input.
The easiest way to see that is indeed the case is to notice that the iterated GHZ problem is a~$\GAMES_\Delta[\{\SYMM, \GHZ\}]$ instance with the following promises on the input:
\begin{itemize}
  \item The circuit in each white node is the same, namely one where the first game is the SYMM game and the rest are GHZ games.
  The circuit is linear and the output of one game flows directly to become the input of the next game.
  \item The $\Delta$-edge coloring (recall that this is a coloring of the black nodes) gives the order of the games:
  Black node with color~$1$ plays the SYMM game while the rest of the black nodes play the GHZ game.
  Note that the setting is actually even more relaxed, as the algorithm is free to permute the order of the colors.
  \item The $\Delta$-edge coloring also promises that each white agrees on the order of the games with its adjacent white nodes.
  In the network of games problems in general, the first half-game of one white node may be played with the last half-game of another white node.
\end{itemize}

Using similar reasoning, one can also get convinced that the iterated CHSH game, defined in \cref{ssec:intro-iterated-def}, is a relaxation of a network of games problem, this time $\GAMES_\Delta[\{\CHSH\}]$.
Finally, we note that CHSH, SYMM and GHZ are non-signaling games, as defined in the next section (\cref{def:ns-game}), and hence our upper bound algorithm works for solving both $\GAMES_\Delta[\{\CHSH\}]$ and $\GAMES_\Delta[\{\SYMM, \GHZ\}]$, and therefore also for solving the iterated CHSH and GHZ games.

\subsection{Solving networks of non-signaling games deterministically}

An important special class of games are so-called non-signaling games.
For a survey on the topic, see \cite{Brunner_Review}.
We now give a definition of non-signaling games adapted to our purposes.
Before doing that, we introduce some notation and minor definitions.
For any function \(f: A \to B\) and any subset \(S \subseteq A\), we denote by \(f \restriction_S\) the function \(g: S \to B\) such that \(f(x) = g(x)\) for all \(x \in S\).
We now define the notion of a \emph{strategy} for a game.

\begin{definition}[Strategy]\label{def:strategy}
  For any fixed game \(\mathfrak{g} \subseteq \Sigma^m \times \Sigma^m\) of \(m\) players, consider a mapping from inputs to probability distributions over outputs \(\rho^{(\mathfrak{g})}: (x_1,\dots,x_m) \mapsto \{(\lambda_j : [m] \to \Sigma^m, p_j) : j \in J\}\). 
  Here, \((x_1, \dots, x_m) \in \Sigma^m\) is an input for \(\mathfrak{g}\), \(\lambda_j : [m] \to \Sigma^m\) is an assignment of outputs to the \([m]\) players, \(p_j\) is the probability that \(\lambda_j\) is the actual outputs of the players, and \(J\) is any set used for indexing.
  We call \(\rho^{(\mathfrak{g})}\) a \emph{strategy} for \(\mathfrak{g}\).
  Since \(\{(\lambda_j : [m] \to \Sigma^m, p_j) : j \in J\}\) is a probability distribution, \(\sum_{j \in J} p_j = 1\) and, for all \(j \in J\), \(p_j \ge 0\).
\end{definition}

For any game \(\mathfrak{g}\) and any strategy \(\rho^{(\mathfrak{g})}\) for \(\mathfrak{g}\), we say that \(\rho^{(\mathfrak{g})}\) solves \(\mathfrak{g}\) with probability at least \(p\) if, for any input \((x_1,\dots, x_m) \in \Sigma^m\) to \(\mathfrak{g}\), the distribution \(\{(\lambda_j : [m] \to \Sigma^m, p_j) : j \in J\} = \rho^{(\mathfrak{g})}((x_1, \dots, x_m))\) satisfies the following property:
\[
  \sum_{\substack{j \in J : \\ (\lambda_{j}(1), \dots, \lambda_j(m)) \in \mathfrak{g}(x_1, \dots, x_m)}} p_j \ge p.
\]

For any subset of players \(S = \{i_1, \dots, i_s\} \subseteq [m]\) (with \(i_1 < i_2 < \dots < i_s\)), let us define the restriction of \(\rho^{(\mathfrak{g})}\) to \(S\) as \(\rho^{(\mathfrak{g})}_{S} : (x_{i_1}, \dots, x_{i_s}) \mapsto \{(\lambda_{S,j} : S \to \Sigma^s, p_{S,j}) : j \in J_S\} \) where 
\[
  p_{S,j} = \sum_{\substack{i \in J : \\ \lambda_{S,j} = \lambda_i \restriction_{S}}} p_i.
\]

\begin{definition}[Non-signaling game]
  \label{def:ns-game}
  A game $\mathfrak{g} \subseteq \Sigma^m \times \Sigma^m$ is said to be
  \emph{non-signaling} if there exists a strategy \(\rho^{(\mathfrak{g})}: (x_1, \dots, x_m) \mapsto \{(y_j: [m] \to \Sigma^m, p_j): j \in J\}\) solving \(\mathfrak{g}\) with probability \(1\) such that, for every pair of inputs \((x_1^{(1)}, \dots, x_m^{(1)}), (x_1^{(2)}, \dots, x_m^{(2)}) \in \Sigma^m \times \Sigma^m\) to \(\mathfrak{g}\), and any subset of the players $S = \{i_1, \dots, i_s\} \subseteq [m]$ (with \(i_1 < i_2 < \dots < i_s\)), if \(x_i^{(1)} = x_i^{(2)}\) for all \(i \in S\), then \(\rho^{(\mathfrak{g})}_S(x_{i_1}^{(1)}, \dots, x_{i_s}^{(1)}) = \rho^{(\mathfrak{g})}_S(x_{i_1}^{(2)}, \dots, x_{i_s}^{(2)})\).
  Also, \(\rho^{(\mathfrak{g})}\) is said to be a \emph{non-signaling strategy} for \(\mathfrak{g}\).
\end{definition}

Interestingly, for every non-signaling game, the existence of a non-signaling strategy solving \(\mathfrak{g}\) guarantees a strong completability property, which we exploit for our major result of this section.
Let us now describe \emph{strongly completable games}.

Informally, a game is strongly completable if every move can be
\enquote{sequentialized} by splitting it into \emph{plies} where each player
$P_i$ does the following in their respective ply:
\begin{enumerate}
  \item\label{step:ns-intuition-1} $P_i$ sees the inputs and outputs of all players that completed their
  ply before it as well as its own input $x_i$.
  \item With \emph{only} this information, $P_i$ can produce an output $y_i$
  such that, no matter what the rest of the inputs are and in what order the
  rest of the players complete their plies, the other players are each able to
  find outputs that make the whole move a valid one using only information described in \cref{step:ns-intuition-1}.
\end{enumerate}

Note we are assuming that the process determining the order in which players
must complete their plies is non-adaptive (i.e., it follows a fixed
permutation).

\begin{definition}[Strongly completable game]
  \label{def:sc-game}
  A game $\mathfrak{g} \subseteq \Sigma^m \times \Sigma^m$ is said to be
  \emph{strongly completable} if, for every permutation $\sigma \colon [m] \to [m]$, we
  have that
  \[
    \forall x_{\sigma(1)} \in \Sigma: \exists y_{\sigma(1)} \in \Sigma:
    \cdots
    \forall x_{\sigma(m)} \in \Sigma: \exists y_{\sigma(m)} \in \Sigma:
    (y_1 \cdots y_m) \in \mathfrak{g}(x_1 \cdots x_m).
  \]
\end{definition}

As mentioned before, non-signaling games are strongly completable.

\begin{lemma}\label{lemma:ns-game-implies-sc-game}
  Any non-signaling game \(\mathfrak{g} \subseteq \Sigma^m \times \Sigma^m\) is strongly completable.
\end{lemma}
\begin{proof}
  Since \(\mathfrak{g} \) is non-signaling, there exists a non-signaling strategy \(\rho^{(\mathfrak{g} )}\) solving \(\mathfrak{g} \) with probability \(1\), which implies that all outputs sampled according to \(\rho^{(\mathfrak{g} )}\) are global solutions.
  Fix any permutation \(\sigma: [m] \to [m]\).
  Now, for the first player~\(\sigma(1)\) with input~\(x_{\sigma(1)}\), sample any output \(y_{\sigma(1)}\) from the restriction \(\rho^{(\mathfrak{g} )}_{\{\sigma(1)\}}(x_{\sigma(1)})\).
  Now we proceed recursively.
  Let \(i > 1\) and suppose that all inputs and outputs for players \(\sigma(1),\dots,\sigma(i-1)\) have been fixed.
  For player \(\sigma(i)\) with input \(x_{\sigma(i)}\), sample any output \(y_{\sigma(i)}\) using \(\rho^{(\mathfrak{g} )}_{\{\sigma(1), \dots, \sigma(i)\}}(x_{\sigma(1)},\dots, x_{\sigma(i)})\) as follows: For any \(j \in [m]\), let \(p(y_{\sigma(1)},\dots,y_{\sigma(j)})\) the probability of sampling \((y_{\sigma(1)},\dots,y_{\sigma(j)})\) from \(\rho^{(\mathfrak{g} )}_{\{\sigma(1), \dots, \sigma(j)\}}(x_{\sigma(1)},\dots, x_{\sigma(j)})\). 
  Then, any \(y_{\sigma(i)}\) is sampled with probability
  \[
    \frac{p(y_{\sigma(1)},\dots,y_{\sigma(i)})}{p(y_{\sigma(1)},\dots,y_{\sigma(i-1)})}.
  \]
  This procedure is well-defined because at any step there is always some output to choose with non-zero probability, due to the non-signaling property of \(\rho^{(\mathfrak{g})}\) (which solves \(\mathfrak{g}\) with probability 1).
  This is because, at any step \(i\), the probability of having sampled \((y_{\sigma(1)},\dots,y_{\sigma(i)})\) is exactly the unconditional probability associated to the output labeling \(\lambda_i: \{\sigma(1),\dots,\sigma(i)\} \to \Sigma^{i}\) defined by \(\lambda_i(\sigma(j)) = y_{\sigma(j)}\): such distribution is given by \(\rho^{(\mathfrak{g} )}_{\{\sigma(1), \dots, \sigma(i)\}}(x_{\sigma(1)},\dots, x_{\sigma(i)})\).
  Let us denote this probability by \(p_i\).
  We prove this statement by induction.
  If \(i = 1\), the thesis is trivial by construction.
  Assume \(i>1\) and the thesis to be true for \(i - 1\).
  Let \(Y_1, \dots, Y_i\) be the random variables describing the outputs for players \(\sigma(1),\dots,\sigma(i)\), respectively.
  Then, the probability of sampling \((y_{\sigma(1)},\dots,y_{\sigma(i)})\) can be expressed as follows:
  \begin{align}
    & \Pr\left[(Y_1, \dots, Y_i) = (y_{\sigma(1)},\dots,y_{\sigma(i)})\right] \nonumber \\
    = \ & \Pr\left[(Y_1, \dots, Y_i) = (y_{\sigma(1)},\dots,y_{\sigma(i)}) \given (Y_1, \dots, Y_{i-1}) = (y_{\sigma(1)},\dots,y_{\sigma(i-1)}) \right] \nonumber\\ 
    & \cdot \Pr\left[(Y_1, \dots, Y_{i-1}) = (y_{\sigma(1)},\dots,y_{\sigma(i-1)}) \right] \nonumber\\
    = \ & \frac{p_i}{p_{i-1}} \cdot p_{i-1} \label{eq:conditional-sampling} \\
    = \ & p_i, \nonumber
  \end{align}
  where \cref{eq:conditional-sampling} holds by construction of the sampling procedure.
  Finally, we conclude that \(\Pr\left[(Y_1, \dots, Y_{m}) = (y_{\sigma(1)},\dots,y_{\sigma(m)})\right] = p_m > 0\).
  Since all outputs sampled with non-zero probability according to \(\rho^{(\mathfrak{g})}\) are actual solutions, then  \((y_{\sigma(1)},\dots,y_{\sigma(m)})\) is a global solution, completing the lemma.
\end{proof}

We finally present the main result of this section:

\begin{theorem}
  For any collection of strongly completable games $\mathcal{G} = \{ \mathfrak{g}_i
  \subseteq \Sigma^m \times \Sigma^m \}_i$ and any $d \in \N_+$, there is a
  deterministic \local algorithm that solves $\GAMES_d[\mathcal{G}]$ in $O(d)$
  rounds.
  \label{thm:det-ub}
\end{theorem}
\cref{thm:det-ub} and \cref{lemma:ns-game-implies-sc-game} imply the following corollary.
\begin{corollary}\label{cor:det-ub}
  For any collection of non-signaling games $\mathcal{G} = \{ \mathfrak{g}_i
  \subseteq \Sigma^m \times \Sigma^m \}_i$ and any $d \in \N_+$, there is a
  deterministic \local algorithm that solves $\GAMES_d[\mathcal{G}]$ in $O(d)$
  rounds.
\end{corollary}

To make the proof simpler, we make the following trivial observation:
\begin{observation}
  \label{obs:circuit-ordering}
  Let $C$ be a degree-$d$ circuit of half-games.
  Without loss of generality, we may assume that half-games~$H_1, \dots, H_d$ appear in a topological order in circuit~$C$.
\end{observation}

With the help of this observation, we notice that the half-games of a circuit partition the circuit into connected components.
In particular, we can consider the circuit formed by only the first~$k$ games.
\begin{definition}[Prefix of a circuit of half-games]
  Let $C$ be a degree-$d$ circuit of half-games, and let $\xi$ and $y_1, \dots, y_d$ be the inputs and let $x_1, \dots, x_d$ be the outputs of~$C$.
  Let $k \in [d]$.
  Define~$C^k$ to be the circuit consisting of only components connected to half-games~$H_1, \dots, H_k$.
  In particular, circuit~$C^k$ has only inputs $\xi$ and (some subset of) $y_1, \dots, y_{k-1}$, and fully defines the outputs~$x_1, \dots, x_k$.
  This is because, by the above observation, the games~$H_1, \dots, H_k$ are in a topological order in~$C$.
  In particular, the input for game~$H_i$ depends only on~$\xi$ and the outputs of games~$H_j$ for $j < i$.
  See~\cref{fig:half-game-circuit} for a visualization.
\end{definition}

We can now present our algorithm~$\mathcal{A}$.
The algorithm proceeds in $d$ phases, each of which consists of two communication rounds:
one where information flows from white to black nodes and another where
it does so in the reverse direction.
We refer to these as the \emph{white} and \emph{black rounds}, accordingly.
Let us now describe what occurs in phase number $i \in [d]$.
\begin{enumerate}
  \item The white nodes keep track of the results~$y_1, \dots, y_{i-1}$ of the half-games they have determined already.
  In the white phase~$i$, each white node~$w$ evaluates~$C^i(\xi, y_1, \dots, y_{i-1})$ to get value~$x_i$.
  Node~$w$ then sends this value to the black node representing the said game according to the input permutation.

  \item Let $j \in [m]$ denote the number of messages black node~$b$ has received by black round~$i$, and let $s \in [m]^j$ be player numbers of the senders those messages in the order of reception, breaking ties arbitrarily among simultanously-arrived messages.
  The black node~$b$ can map the ports to player numbers using the local input permutation.
  Let~$x \in \Sigma^j$ be the values of those messages, in the same order.

  Now the node~$b$ now plays game~$\mathfrak{g}$ using $s$ as the playing order and $x$ as the inputs, producing partial output~$y \in \Sigma^j$ for those nodes.
  Because~$\mathfrak{g}$ is a strongly completable game, it can pick the outputs~$y$ one-by-one in such way that the game is still solvable for every possible future inputs.
  Moreover, the node plays the game in a deterministic manner in order not to change the values of previously-assigned outputs.
  This is possible by the definition of strongly completable games, see \cref{def:sc-game}.
  Finally, node~$b$ sends the player output~$y_k$ to each neighbor~$k$ it received a message in this round.
\end{enumerate}

The round complexity of $\mathcal{A}$ is evidently $2d = O(d)$.
The white constraint is satisfied by construction.
Hence for the correctness we need only argue that the black constraint is
satisfied, which we do in the following lemma:

\begin{lemma}
  Algorithm $\mathcal{A}$ is correct, that is, for every black node $b \in B$
  with input game $\mathfrak{g}$, the edges around it are labeled with
  $(x_i,y_i)$ such that $y \in \mathfrak{g}(x)$ (where $x = x_1 \cdots x_m$ and
  $y = y_1 \cdots y_m$).
  \label{lem:det-ub-correct}
\end{lemma}

\begin{proof}
  Fix such a black node $b \in B$.
  First notice that, since the algorithm is deterministic, for a fixed input to
  the algorithm we will have a fixed order in which $b$ receives inputs $x_i$
  from white nodes and, hence, also a fixed order in which it sets the outputs
  $y_i$.  
  Hence there is a permutation $\sigma$ as in \cref{def:sc-game} that
  corresponds to the order in which $b$ sets the $y_i$.
  Having observed this, note that there is no difference (from the point of view
  of $b$) if the inputs $x_i$ arrive in separate phases or all in the same one.
  Hence we can assume that the latter is the case.
  Having done so, then what $b$ does is exactly what is guaranteed to be
  possible by \cref{def:sc-game}, and thus the claim follows.
\end{proof}

This concludes the proof of \cref{thm:det-ub}.

\ifanon\else
\section*{Acknowledgments}

Augusto Modanese is supported by the Helsinki Institute for Information
Technology (HIIT).
Henrik Lievonen is supported by the Research Council of Finland, Grants 333837 and 359104.
Alkida Balliu and Dennis Olivetti are supported by MUR (Italy) Department of Excellence 2023 - 2027.
Francesco d'Amore is supported by  MUR FARE
2020 - Project PAReCoDi CUP J43C22000970001. Xavier Coiteux-Roy acknowledges funding from the BMW endowment fund and from the Swiss National Science Foundation. Fran{\c c}ois Le Gall is supported by JSPS KAKENHI grants Nos.~JP20H05966, 20H00579, 24H00071, MEXT Q-LEAP grant No.~JPMXS0120319794 and JST CREST grant No.~JPMJCR24I4.
Lucas Tendick, Marc-Olivier Renou and Isadora Veeren acknowledge funding from INRIA and CIEDS through the Action Exploratoire project DEPARTURE, and from the ANR for the JCJC grant LINKS (ANR-23-CE47-0003).

\fi

\urlstyle{same}
\DeclareUrlCommand{\path}{}
\bibliographystyle{plainurl}
\bibliography{references}

\newpage
\appendix

\section{Omitted proofs for the first round elimination step}\label{sec:lb-proofs-omitted-first-step}
We first observe that, a necessary condition for a label to be at least as strong as another label, is for the two labels to be of the same color.
\begin{observation}\label{obs:no-arrow-different-colors}
	Let $L_1$ and $L_2$ be two labels of different colors. Then, $L_1 \not\le L_2$ and $L_2 \not\le L_1$.
\end{observation}

Then, by the definition of strength relation, the following also holds.
\begin{observation}\label{obs:arrow-is-present}
	Let $L_1$ and $L_2$ be two labels. Let $C$ be the set of configurations described by a set $S$ of condensed configurations. Assume that, for all condensed configurations of $S$ it holds that, if a disjunction contains $L_1$, then the disjunction also contains $L_2$. Then, w.r.t.\ the constraint $C$, it holds that $L_1 \le L_2$.
\end{observation}

By transitivity of the strength relation, the following also holds.
\begin{observation}\label{obs:arrow-negation-transitivity}
	Let $L_1$, $L_2$, $L_3$ be labels satisfying $L_1 \le L_2$ and $L_3 \not\le L_2$. Then, $L_3 \not \le L_1$.
\end{observation}

\begin{proof}[Proof of \Cref{lem:diagram-gone}]
	By \Cref{obs:no-arrow-different-colors}, labels of color $j$ can only be in relation with labels of color $j$.
	By \Cref{obs:arrow-is-present}, the relation depicted in the diagram is present. Hence, we only need to prove that $\reQ{j} \not\le \reE{j}$, which holds since, while the configuration $\reE{j}^2 ~ \reQ{j}$ is allowed, the configuration $\reE{j}^3$ is not allowed.
\end{proof}

\begin{proof}[Proof of \Cref{lem:diagram-first}]
	By \Cref{obs:no-arrow-different-colors}, labels of color $1$ can only be in relation with labels of color $1$.
	By \Cref{obs:arrow-is-present}, the relations depicted in the diagram are present. Hence, we only need to prove that no additional relations than the ones of \Cref{fig:diagram-first} hold. We consider all possible cases:
	\begin{itemize}[noitemsep]
		\raggedright
		\item $\reMM{1} \not\le \reMY{1}{0}$ because the configuration $\reMM{1} ~ \reMY{1}{0}^2$ is allowed but $\reMY{1}{0}^3$ is not.
		\item $\reMM{1} \not\le \reMY{1}{1}$ because the configuration $\reMY{1}{1} ~ \reMM{1}^2$ is allowed but $\reMM{1} ~ \reMY{1}{1}^2$ is not.
		\item $\reMM{1} \not\le \reEE{1}$ holds by \Cref{obs:arrow-negation-transitivity}.
		
		\item $\reMY{1}{0} \not\le \reMY{1}{1}$ because the configuration $\reMY{1}{0} ~ \reMY{1}{1} ~ \reMM{1}$ is allowed but $\reMM{1} \reMY{1}{1}^2$ is not.
		\item $\reMY{1}{1} \not\le \reMY{1}{0}$ because the configuration $\reMY{1}{1} ~ \reMY{1}{0}^2$ is allowed but $\reMY{1}{0}^3$ is not.
		
		\item $\reMY{1}{0} \not\le \reEE{1}$ and $\reMY{1}{1} \not\le \reEE{1}$ hold by \Cref{obs:arrow-negation-transitivity}. \qedhere
	\end{itemize}
\end{proof}

\begin{proof}[Proof of \Cref{lem:diagram-present}]
	By \Cref{obs:no-arrow-different-colors}, labels of color $j$ can only be in relation with labels of color $j$.
	By \Cref{obs:arrow-is-present}, the relations depicted in the diagram are present. Hence, we only need to prove that no additional relations than the ones of \Cref{fig:diagram-present} hold. We consider all possible cases:
	\begin{itemize}[noitemsep]
		\raggedright
		\item $\reMM{j} \not\le \reXY{j}{0}{0}$ because $\reXY{j}{0}{0} ~ \reXY{j}{0}{1} ~ \reMM{j}$ is allowed but $\reXY{j}{0}{1} ~ \reXY{j}{0}{0}^2$ is not.
		\item $\reMM{j} \not\le \reXY{j}{0}{1}$ because $\reMM{j} ~ \reXY{j}{0}{0}^2$ is allowed but $\reXY{j}{0}{1} ~ \reXY{j}{0}{0}^2$ is not.
		\item $\reMM{j} \not\le \reXY{j}{1}{0}$ because $\reXY{j}{0}{0} ~ \reXY{j}{1}{0} ~ \reMM{j}$ is allowed but $\reXY{j}{0}{0} ~ \reXY{j}{1}{0}^2$ is not.
		\item $\reMM{j} \not\le \reXY{j}{1}{1}$ because $\reXY{j}{0}{0} ~ \reXY{j}{1}{1} ~ \reMM{j}$ is allowed but $\reXY{j}{0}{0} ~ \reXY{j}{1}{1}^2$ is not.
		
		\item $\reMM{j} \not\le \reXM{j}{0}$, $\reMM{j} \not\le \reXM{j}{1}$, $\reMM{j} \not\le \reMY{j}{0}$, $\reMM{j} \not\le \reMY{j}{1}$, and  $\reMM{j} \not\le \reEE{j}$ by \Cref{obs:arrow-negation-transitivity}.
		
		\item $\reXY{j}{0}{0} \not\le \reXY{j}{0}{1}$ because $\reXY{j}{0}{0}^3$ is allowed but $\reXY{j}{0}{1} ~ \reXY{j}{0}{0}^2$ is not.
		\item $\reXY{j}{0}{0} \not\le \reXY{j}{1}{0}$ because $\reXY{j}{1}{0} ~ \reXY{j}{0}{0}^2$ is allowed but $\reXY{j}{0}{0} ~ \reXY{j}{1}{0}^2$ is not.
		\item $\reXY{j}{0}{0} \not\le \reXY{j}{1}{1}$ because $\reXY{j}{1}{1} ~ \reXY{j}{0}{0}^2$ is allowed but $\reXY{j}{0}{0} ~ \reXY{j}{1}{1}^2$ is not.
		
		\item $\reXY{j}{0}{0} \not\le \reXM{j}{0}$, $\reXY{j}{0}{0} \not\le \reXM{j}{1}$, $\reXY{j}{0}{0} \not\le \reMY{j}{0}$, $\reXY{j}{0}{0} \not\le \reMY{j}{1}$, and $\reXY{j}{0}{0} \not\le \reEE{j}$ by \Cref{obs:arrow-negation-transitivity}.
		
		\item $\reXY{j}{0}{1} \not\le \reXY{j}{0}{0}$ because $\reXY{j}{0}{0} ~ \reXY{j}{0}{1}^2$ is allowed but $\reXY{j}{0}{1} ~ \reXY{j}{0}{0}^2$ is not.
		\item $\reXY{j}{0}{1} \not\le \reXY{j}{1}{0}$ because $\reXY{j}{0}{0} ~ \reXY{j}{0}{1} ~ \reXY{j}{1}{0}$ is allowed but $\reXY{j}{0}{0} ~ \reXY{j}{1}{0}^2$ is not.
		\item $\reXY{j}{0}{1} \not\le \reXY{j}{1}{1}$ because $\reXY{j}{0}{0} ~ \reXY{j}{0}{1} ~ \reXY{j}{1}{1}$ is allowed but $\reXY{j}{0}{0} ~ \reXY{j}{1}{1}^2$ is not.
		
		\item $\reXY{j}{0}{1} \not\le \reXM{j}{0}$, $\reXY{j}{0}{1} \not\le \reXM{j}{1}$, $\reXY{j}{0}{1} \not\le \reMY{j}{0}$, $\reXY{j}{0}{1} \not\le \reMY{j}{1}$, and $\reXY{j}{0}{1} \not\le \reEE{j}$ by \Cref{obs:arrow-negation-transitivity}.
		
		\item $\reXY{j}{1}{0} \not\le \reXY{j}{0}{0}$ because $\reXY{j}{0}{0} ~ \reXY{j}{0}{1} ~ \reXY{j}{1}{0}$ is allowed but $\reXY{j}{0}{1} ~ \reXY{j}{0}{0}^2$ is not.
		\item $\reXY{j}{1}{0} \not\le \reXY{j}{0}{1}$ because $\reXY{j}{1}{0} ~ \reXY{j}{0}{0}^2$ is allowed but $\reXY{j}{0}{1} ~ \reXY{j}{0}{0}^2$ is not.
		\item $\reXY{j}{1}{0} \not\le \reXY{j}{1}{1}$ because $\reXY{j}{0}{0} ~ \reXY{j}{1}{0} ~ \reXY{j}{1}{1}$ is allowed but $\reXY{j}{0}{0} ~ \reXY{j}{1}{1}^2$ is not.
		
		\item $\reXY{j}{1}{0} \not\le \reXM{j}{0}$, $\reXY{j}{1}{0} \not\le \reXM{j}{1}$, $\reXY{j}{1}{0} \not\le \reMY{j}{0}$, $\reXY{j}{1}{0} \not\le \reMY{j}{1}$, $\reXY{j}{1}{0} \not\le \reEE{j}$ by \Cref{obs:arrow-negation-transitivity}.

		\item $\reXY{j}{1}{1} \not\le \reXY{j}{0}{0}$ because $\reXY{j}{0}{0} ~ \reXY{j}{0}{1} ~ \reXY{j}{1}{1}$ is allowed but $\reXY{j}{0}{1} ~ \reXY{j}{0}{0}^2$ is not.
		\item $\reXY{j}{1}{1} \not\le \reXY{j}{0}{1}$ because $\reXY{j}{1}{1} ~ \reXY{j}{0}{0}^2$ is allowed but $\reXY{j}{0}{1} ~ \reXY{j}{0}{0}^2$ is not.
		\item $\reXY{j}{1}{1} \not\le \reXY{j}{1}{0}$ because $\reXY{j}{0}{0} ~ \reXY{j}{1}{0} ~ \reXY{j}{1}{1}$ is allowed but $\reXY{j}{0}{0} ~ \reXY{j}{1}{0}^2$ is not.
		
		\item $\reXY{j}{1}{1} \not\le \reXM{j}{0}$, $\reXY{j}{1}{1} \not\le \reXM{j}{1}$, $\reXY{j}{1}{1} \not\le \reMY{j}{0}$, $\reXY{j}{1}{1} \not\le \reMY{j}{1}$, and  $\reXY{j}{1}{1} \not\le \reEE{j}$ by \Cref{obs:arrow-negation-transitivity}.

		\item $\reXM{j}{0} \not\le \reXY{j}{1}{0}$ because $\reXY{j}{0}{0} ~ \reXY{j}{1}{0} ~ \reXM{j}{0}$ is allowed but $\reXY{j}{0}{0} ~ \reXY{j}{1}{0}^2$ is not.
		\item $\reXM{j}{0} \not\le \reXY{j}{1}{1}$ because $\reXY{j}{0}{0} ~ \reXY{j}{1}{1} ~ \reXM{j}{0}$ is allowed but $\reXY{j}{0}{0} ~ \reXY{j}{1}{1}^2$ is not.
		
		\item $\reXM{j}{0} \not\le \reXM{j}{1}$, $\reXM{j}{0} \not\le \reMY{j}{0}$, $\reXM{j}{0} \not\le \reMY{j}{1}$, and $\reXM{j}{0} \not\le \reEE{j}$ by \Cref{obs:arrow-negation-transitivity}.		
		
		\item $\reXM{j}{1} \not\le \reXY{j}{0}{0}$ because $\reXY{j}{0}{0} ~ \reXY{j}{0}{1} ~ \reXM{j}{1}$ is allowed but $\reXY{j}{0}{1} ~ \reXY{j}{0}{0}^2$ is not.
		\item $\reXM{j}{1} \not\le \reXY{j}{0}{1}$ because $\reXM{j}{1} ~ \reXY{j}{0}{0}^2$ is allowed but $\reXY{j}{0}{1} ~ \reXY{j}{0}{0}^2$ is not.
		
		\item $\reXM{j}{1} \not\le \reXM{j}{0}$, $\reXM{j}{1} \not\le \reMY{j}{0}$, $\reXM{j}{1} \not\le \reMY{j}{1}$, and $\reXM{j}{1} \not\le \reEE{j}$ by \Cref{obs:arrow-negation-transitivity}.		
		
		\item $\reMY{j}{0} \not\le \reXY{j}{0}{1}$ because $\reMY{j}{0} ~ \reXY{j}{0}{0}^2$ is allowed but $\reXY{j}{0}{1} ~ \reXY{j}{0}{0}^2$ is not.
		\item $\reMY{j}{0} \not\le \reXY{j}{1}{1}$ because $\reXY{j}{0}{0} ~  \reXY{j}{1}{1} ~ \reMY{j}{0}$ is allowed but $\reXY{j}{0}{0} ~ \reXY{j}{1}{1}^2$ is not.
		
		\item $\reMY{j}{0} \not\le \reXM{j}{0}$, $\reMY{j}{0} \not\le \reXM{j}{1}$, $\reMY{j}{0} \not\le \reMY{j}{1}$, and $\reMY{j}{0} \not\le \reEE{j}$ by \Cref{obs:arrow-negation-transitivity}.

		\item $\reMY{j}{1} \not\le \reXY{j}{0}{0}$ because $\reXY{j}{0}{0} ~ \reXY{j}{0}{1} ~ \reMY{j}{1}$ is allowed but $\reXY{j}{0}{1} ~ \reXY{j}{0}{0}^2$ is not.
		\item $\reMY{j}{1} \not\le \reXY{j}{1}{0}$ because $\reXY{j}{0}{0} ~ \reXY{j}{1}{0} ~ \reMY{j}{1}$ is allowed but $\reXY{j}{0}{0} ~ \reXY{j}{1}{0}^2$ is not.
		
		\item $\reMY{j}{1} \not\le \reXM{j}{0}$, $\reMY{j}{1} \not\le \reXM{j}{1}$, $\reMY{j}{1} \not\le \reMY{j}{0}$, and $\reMY{j}{1} \not\le \reEE{j}$ by \Cref{obs:arrow-negation-transitivity}. \qedhere
	\end{itemize}
\end{proof}

\begin{proof}[Proof of \Cref{lem:diagram-special}]
	By \Cref{obs:no-arrow-different-colors}, labels of color $j$ can only be in relation with labels of color $j$.
	By \Cref{obs:arrow-is-present}, the relations depicted in the diagram are present. Hence, we only need to prove that no additional relations than the ones of \Cref{fig:diagram-special} hold. We consider all possible cases:
	\begin{itemize}[noitemsep]
		\raggedright
		\item $\reMM{j} \not\le \reXY{j}{0}{0}$ because $\reXY{j}{0}{0} ~ \reXY{j}{0}{1} ~ \reMM{j}$ is allowed but $\reXY{j}{0}{1} ~ \reXY{j}{0}{0}^2$ is not.
		\item $\reMM{j} \not\le \reXY{j}{0}{1}$ because $\reMM{j} ~ \reXY{j}{0}{0}^2$ is allowed but $\reXY{j}{0}{1} ~ \reXY{j}{0}{0}^2$ is not.
		\item $\reMM{j} \not\le \reXY{j}{1}{0}$ because $\reXY{j}{0}{0} ~ \reXY{j}{1}{0} ~ \reMM{j}$ is allowed but $\reXY{j}{0}{0} ~\reXY{j}{1}{0}^2$ is not.
		\item $\reMM{j} \not\le \reXY{j}{1}{1}$ because $\reXY{j}{0}{0} ~ \reXY{j}{1}{1} ~ \reMM{j}$ is allowed but $\reXY{j}{0}{0} ~ \reXY{j}{1}{1}^2$ is not.
		\item $\reMM{j} \not\le \reMY{j}{0}$ and $\reMM{j} \not\le \reMY{j}{1}$ by \Cref{obs:arrow-negation-transitivity}.
		
		\item $\reXY{j}{0}{0}  \not\le  \reXY{j}{0}{1}$ because $\reXY{j}{0}{0}^3$ is allowed but $\reXY{j}{0}{1} ~ \reXY{j}{0}{0}^2$ is not.
		\item $\reXY{j}{0}{0}  \not\le  \reXY{j}{1}{0}$ because $\reXY{j}{1}{0} ~ \reXY{j}{0}{0}^2$ is allowed but $\reXY{j}{0}{0} ~ \reXY{j}{1}{0}^2$ is not.
		\item $\reXY{j}{0}{0} \not\le  \reXY{j}{1}{1}$ because $\reXY{j}{1}{1} \reXY{j}{0}{0}^2$ is allowed but $\reXY{j}{0}{0} ~ \reXY{j}{1}{1}^2$ is not.
		\item $\reXY{j}{0}{0}  \not\le \reMY{j}{0}$ and $\reXY{j}{0}{0}  \not\le  \reMY{j}{1}$ by \Cref{obs:arrow-negation-transitivity}.
		
		\item $\reXY{j}{0}{1}  \not\le  \reXY{j}{0}{0}$ because $\reXY{j}{0}{0} ~ \reXY{j}{0}{1}^2$ is allowed but $\reXY{j}{0}{1} ~ \reXY{j}{0}{0}^2$ is not.
		\item $\reXY{j}{0}{1} \not\le \reXY{j}{1}{0}$ because $\reXY{j}{0}{0} ~ \reXY{j}{0}{1} ~ \reXY{j}{1}{0}$ is allowed but $\reXY{j}{0}{0} ~ \reXY{j}{1}{0}^2$ is not.
		\item $\reXY{j}{0}{1} \not\le \reXY{j}{1}{1}$ because $\reXY{j}{0}{0} ~ \reXY{j}{0}{1} ~ \reXY{j}{1}{1}$ is allowed but $\reXY{j}{0}{0} ~ \reXY{j}{1}{1}^2$ is not.
		\item $\reXY{j}{0}{1} \not\le \reMY{j}{0}$ and $\reXY{j}{0}{1} \not\le \reMY{j}{1}$ by \Cref{obs:arrow-negation-transitivity}.
		
		\item $\reXY{j}{1}{0} \not\le \reXY{j}{0}{0}$ because $\reXY{j}{0}{0} ~ \reXY{j}{0}{1} ~ \reXY{j}{1}{0}$ is allowed but $\reXY{j}{0}{1} ~ \reXY{j}{0}{0}^2$ is not.
		\item $\reXY{j}{1}{0} \not\le \reXY{j}{0}{1}$ because $\reXY{j}{1}{0} ~ \reXY{j}{0}{0}^2$ is allowed but $\reXY{j}{0}{1} ~ \reXY{j}{0}{0}^2$ is not.
		\item $\reXY{j}{1}{0} \not\le \reXY{j}{1}{1}$ because $\reXY{j}{0}{0} ~ \reXY{j}{1}{0} ~ \reXY{j}{1}{1}$ is allowed but $\reXY{j}{0}{0} ~ \reXY{j}{1}{1}^2$ is not.
		\item $\reXY{j}{1}{0} \not\le \reMY{j}{0}$ and $\reXY{j}{1}{0} \not\le \reMY{j}{1}$  by \Cref{obs:arrow-negation-transitivity}.
		
		\item $\reXY{j}{1}{1} \not\le \reXY{j}{0}{0}$ because $\reXY{j}{0}{0} ~ \reXY{j}{0}{1} ~ \reXY{j}{1}{1}$ is allowed but $\reXY{j}{0}{1} ~ \reXY{j}{0}{0}^2$ is not.
		\item $\reXY{j}{1}{1} \not\le \reXY{j}{0}{1}$ because $\reXY{j}{1}{1} ~ \reXY{j}{0}{0}^2$ is allowed but $\reXY{j}{0}{1} ~ \reXY{j}{0}{0}^2$ is not.
		\item $\reXY{j}{1}{1} \not\le \reXY{j}{1}{0}$ because $\reXY{j}{0}{0} ~ \reXY{j}{1}{0} ~ \reXY{j}{1}{1}$ is allowed but $\reXY{j}{0}{0} ~ \reXY{j}{1}{0}^2$ is not.
		\item $\reXY{j}{1}{1} \not\le \reMY{j}{0}$ and $\reXY{j}{1}{1} \not\le \reMY{j}{1}$ by \Cref{obs:arrow-negation-transitivity}.
		
		\item $\reMY{j}{0} \not\le \reXY{j}{0}{1}$ because $\reMY{j}{0} ~ \reXY{j}{0}{0}^2$ is allowed but $\reXY{j}{0}{1} ~ \reXY{j}{0}{0}^2$ is not.
		\item $\reMY{j}{0} \not\le \reXY{j}{1}{1}$ because $\reXY{j}{0}{0} ~ \reXY{j}{1}{1} ~ \reMY{j}{0}$ is allowed but $\reXY{j}{0}{0} ~ \reXY{j}{1}{1}^2$ is not.
		\item $\reMY{j}{1} \not\le  \reXY{j}{0}{0}$ because $\reXY{j}{0}{0} ~ \reXY{j}{0}{1} ~ \reMY{j}{1}$ is allowed but $\reXY{j}{0}{1} ~ \reXY{j}{0}{0}^2$ is not.
		\item $\reMY{j}{1} \not\le \reXY{j}{1}{0}$ because $\reXY{j}{0}{0} ~ \reXY{j}{1}{0} ~ \reMY{j}{1}$ is allowed but $\reXY{j}{0}{0} ~ \reXY{j}{1}{0}^2$ is not.
		\item $\reMY{j}{0} \not\le \reMY{j}{1}$ and $\reMY{j}{1} \not\le \reMY{j}{0}$ by \Cref{obs:arrow-negation-transitivity}. \qedhere
	\end{itemize}
\end{proof}

\section{Omitted proofs for the second round elimination step}\label{sec:lb-proofs-omitted-second-step}
We observe that, as in the case of the first round elimination step, a necessary condition for a label to be at least as strong as another label, is for the two labels to be of the same color.
\begin{observation}\label{obs:no-arrow-different-colors-2}
	Let $L_1$ and $L_2$ be two labels of different colors. Then, $L_1 \not\le L_2$ and $L_2 \not\le L_1$.
\end{observation}

Again, by the definition of strength relation, the following also holds.
\begin{observation}\label{obs:arrow-is-present-2}
	Let $L_1$ and $L_2$ be two labels. Let $C$ be the set of configurations described by a set $S$ of condensed configurations. Assume that, for all condensed configurations of $S$ it holds that, if a disjunction contains $L_1$, then the disjunction also contains $L_2$. Then, w.r.t.\ the constraint $C$, it holds that $L_1 \le L_2$.
\end{observation}

Again, by transitivity of the strength relation, the following also holds.
\begin{observation}\label{obs:arrow-negation-transitivity-2}
	Let $L_1$, $L_2$, $L_3$ be labels satisfying $L_1 \le L_2$ and $L_3 \not\le L_2$. Then, $L_3 \not \le L_1$.
\end{observation}

\begin{proof}[Proof of \Cref{lem:diagram-intermediate-gone}]
	By \Cref{obs:no-arrow-different-colors-2}, labels of color $j$ can only be in relation with labels of color $j$.
	By \Cref{obs:arrow-is-present-2}, the relation depicted in the diagram is present. Hence, we only need to prove that $\gen{\reE{j}} \not\le \gen{\reQ{j}}$, which holds since
	the configuration $L_1 \ldots L_\Delta$ is allowed, where:
	\begin{itemize}[noitemsep]
		\item For $1 \le k \le \Delta - i - 1$, $L_k = \gen{\reMM{k}}$;
		\item $L_{\Delta-i} = \gen{\reXY{\Delta-i}{0}{0},\reXY{\Delta-i}{0}{1}}$;
		\item For $\Delta-i+1 \le k \le \Delta$ such that $j \neq k$, $L_k = \gen{\reQ{k}}$;
		\item $L_j = \gen{\reE{j}}$,
	\end{itemize}
	but by replacing $\gen{\reE{j}}$ with $\gen{\reQ{j}}$ we obtain a configuration that is not allowed.
\end{proof}

\begin{proof}[Proof of \Cref{lem:diagram-intermediate-first}]
	By \Cref{obs:no-arrow-different-colors-2}, labels of color $1$ can only be in relation with labels of color $1$.
	By \Cref{obs:arrow-is-present-2}, the relations depicted in the diagram are present. Hence, we only need to prove that no additional relations than the ones of \Cref{fig:diagram-intermediate-first} hold. We consider all possible cases:
	\begin{itemize}
	\item $\gen{\reEE{1}} \not\le \gen{\reMY{1}{0}}$ because the configuration $L_1 \ldots L_\Delta$ is allowed, where:
		\begin{itemize}
		\item $L_1 = \gen{\reEE{1}}$;
		\item For $2 \le k \le \Delta - i - 1$, $L_k = \gen{\reXY{k}{1}{1}}$;
		\item $L_{\Delta-i} = \gen{\reXY{\Delta-i}{1}{0},\reXY{\Delta-i}{1}{1}}$;
		\item For $\Delta-i+1 \le k \le \Delta$, $L_k = \gen{\reQ{k}}$;
	\end{itemize}
	but by replacing $\gen{\reEE{1}}$ with $\gen{\reMY{1}{0}}$ we obtain a configuration that is not allowed.
	\item $\gen{\reEE{1}} \not\le \gen{\reMY{1}{1}}$ for symmetric reasons as in the case $\gen{\reEE{1}} \not\le \gen{\reMY{1}{0}}$.
	\item $\gen{\reEE{1}} \not\le \gen{\reMM{1}}$ by \Cref{obs:arrow-negation-transitivity-2}.

	\item $\gen{\reMY{1}{0}} \not\le \gen{\reMY{1}{1}}$ because the configuration $L_1 \ldots L_\Delta$ is allowed, where:
	\begin{itemize}
		\item $L_1 = \gen{\reMY{1}{0}}$;
		\item For $2 \le k \le \Delta - i - 1$, $L_k = \gen{\reXY{k}{0}{0}}$;
		\item $L_{\Delta-i} = \gen{\reXY{\Delta-i}{0}{0},\reXY{\Delta-i}{0}{1}}$;
		\item For $\Delta-i+1 \le k \le \Delta$, $L_k = \gen{\reQ{k}}$;
	\end{itemize}
	but by replacing $\gen{\reMY{1}{0}}$ with $\gen{\reMY{1}{1}}$ we obtain a configuration that is not allowed.

	\item $\gen{\reMY{1}{1}} \not\le \gen{\reMY{1}{0}}$  for symmetric reasons as in the case  $\gen{\reMY{1}{0}} \not\le \gen{\reMY{1}{1}}$.
	\item $\gen{\reMY{1}{1}} \not\le \gen{\reMM{1}}$ and $\gen{\reMY{1}{0}} \not\le \gen{\reMM{1}}$ hold by \Cref{obs:arrow-negation-transitivity-2}. \qedhere
	\end{itemize}
\end{proof}

\begin{proof}[Proof of \Cref{lem:diagram-intermediate-present}]
	By \Cref{obs:no-arrow-different-colors-2}, labels of color $j$ can only be in relation with labels of color $j$.
	By \Cref{obs:arrow-is-present-2}, the relations depicted in the diagram are present. Hence, we only need to prove that no additional relations than the ones of \Cref{fig:diagram-intermediate-present} hold. We consider all possible cases.
	\begin{itemize}
		\raggedright
		\item $\gen{\reXY{j}{0}{0}} \not\le \gen{\reXM{j}{1},\reMY{j}{1}}$ because the configuration $L_1 \ldots L_\Delta$ is allowed, where:
		\begin{itemize}[noitemsep]
			\item $L_1 = \gen{\reMY{1}{0}}$;
			\item For $2 \le k \le \Delta - i - 1$, $L_k = \gen{\reXY{k}{0}{0}}$;
			\item $L_{\Delta-i} = \gen{\reXY{\Delta-i}{0}{0},\reXY{\Delta-i}{0}{1}}$;
			\item For $\Delta-i+1 \le k \le \Delta$, $L_k = \gen{\reQ{k}}$;
		\end{itemize}
		but by replacing $\gen{\reXY{j}{0}{0}}$ with $\gen{\reXM{j}{1},\reMY{j}{1}}$ we obtain a configuration that is not allowed.		
		
		\item $\gen{\reXY{j}{0}{0}} \not\le \gen{\reMM{j}}$, $\gen{\reXY{j}{0}{0}} \not\le \gen{\reXY{j}{0}{1}}$, $\gen{\reXY{j}{0}{0}} \not\le \gen{\reXY{j}{1}{0}}$, $\gen{\reXY{j}{0}{0}} \not\le \gen{\reXY{j}{1}{1}}$, $\gen{\reXY{j}{0}{0}} \not\le \gen{\reXY{j}{0}{1},\reXY{j}{1}{0}}$, $\gen{\reXY{j}{0}{0}} \not\le \gen{\reXM{j}{1}}$, $\gen{\reXY{j}{0}{0}} \not\le \gen{\reMY{j}{1}}$ hold by \Cref{obs:arrow-negation-transitivity-2}.
		
		\item $\gen{\reXY{j}{0}{1}} \not\le \gen{\reMM{j}}$, $\gen{\reXY{j}{0}{1}} \not\le \gen{\reXY{j}{0}{0}}$, $\gen{\reXY{j}{0}{1}} \not\le \gen{\reXY{j}{1}{0}}$, $\gen{\reXY{j}{0}{1}} \not\le \gen{\reXY{j}{1}{1}}$, $\gen{\reXY{j}{0}{1}} \not\le \gen{\reXY{j}{0}{0},\reXY{j}{1}{1}}$, $\gen{\reXY{j}{0}{1}} \not\le \gen{\reXM{j}{1}}$, $\gen{\reXY{j}{0}{1}} \not\le \gen{\reMY{j}{0}}$, $\gen{\reXY{j}{0}{1}} \not\le \gen{\reXM{j}{1},\reMY{j}{0}}$, $\gen{\reXY{j}{1}{0}} \not\le \gen{\reMM{j}}$, $\gen{\reXY{j}{1}{0}} \not\le \gen{\reXY{j}{0}{0}}$, $\gen{\reXY{j}{1}{0}} \not\le \gen{\reXY{j}{0}{1}}$, $\gen{\reXY{j}{1}{0}} \not\le \gen{\reXY{j}{1}{1}}$, $\gen{\reXY{j}{1}{0}} \not\le \gen{\reXY{j}{0}{0},\reXY{j}{1}{1}}$, $\gen{\reXY{j}{1}{0}} \not\le \gen{\reXM{j}{0}}$, $\gen{\reXY{j}{1}{0}} \not\le \gen{\reMY{j}{1}}$, $\gen{\reXY{j}{1}{0}} \not\le \gen{\reXM{j}{0},\reMY{j}{1}}$, $\gen{\reXY{j}{1}{1}} \not\le \gen{\reMM{j}}$, $\gen{\reXY{j}{1}{1}} \not\le \gen{\reXY{j}{0}{0}}$, $\gen{\reXY{j}{1}{1}} \not\le \gen{\reXY{j}{0}{1}}$, $\gen{\reXY{j}{1}{1}} \not\le \gen{\reXY{j}{1}{0}}$, $\gen{\reXY{j}{1}{1}} \not\le \gen{\reXY{j}{0}{1},\reXY{j}{1}{0}}$, $\gen{\reXY{j}{1}{1}} \not\le \gen{\reXM{j}{0}}$, $\gen{\reXY{j}{1}{1}} \not\le \gen{\reMY{j}{0}}$, $\gen{\reXY{j}{1}{1}} \not\le \gen{\reXM{j}{0},\reMY{j}{0}}$ for symmetric reasons.

		\item $\gen{\reXM{j}{0}} \not\le \gen{\reXM{j}{1},\reMY{j}{0}}$ because the configuration $L_1 \ldots L_\Delta$ is allowed, where:
		\begin{itemize}[noitemsep]
			\item $L_1 = \gen{\reMY{1}{0}}$;
			\item For $2 \le k < j$, $L_k = \gen{\reXY{k}{0}{0}}$;
			\item $L_j = \gen{\reXM{j}{0}}$
			\item For $j < k \le \Delta - i - 1$, $L_k = \gen{\reXY{k}{1}{1}}$;
			\item $L_{\Delta-i} = \gen{\reXY{\Delta-i}{1}{0},\reXY{\Delta-i}{1}{1}}$;
			\item For $\Delta-i+1 \le k \le \Delta$, $L_k = \gen{\reQ{k}}$;
		\end{itemize}
		but by replacing $\gen{\reXM{j}{0}}$ with $\gen{\reXM{j}{1},\reMY{j}{0}}$ we obtain a configuration that is not allowed.
		\item $\gen{\reXM{j}{0}} \not\le \gen{\reXM{j}{1},\reMY{j}{1}}$ because the configuration $L_1 \ldots L_\Delta$ is allowed, where:
		\begin{itemize}[noitemsep]
			\item $L_1 = \gen{\reMY{1}{0}}$;
			\item For $2 \le k < j$, $L_k = \gen{\reXY{k}{0}{0}}$;
			\item $L_j = \gen{\reXM{j}{0}}$
			\item For $j < k \le \Delta - i - 1$, $L_k = \gen{\reXY{k}{0}{0}}$;
			\item $L_{\Delta-i} = \gen{\reXY{\Delta-i}{0}{0},\reXY{\Delta-i}{0}{1}}$;
			\item For $\Delta-i+1 \le k \le \Delta$, $L_k = \gen{\reQ{k}}$;
		\end{itemize}
		but by replacing $\gen{\reXM{j}{0}}$ with $\gen{\reXM{j}{1},\reMY{j}{1}}$ we obtain a configuration that is not allowed.
		
		\item $\gen{\reXM{j}{0}} \not\le \gen{\reMM{j}}$, $\gen{\reXM{j}{0}} \not\le \gen{\reXY{j}{0}{0}}$, $\gen{\reXM{j}{0}} \not\le \gen{\reXY{j}{0}{1}}$, $\gen{\reXM{j}{0}} \not\le \gen{\reXY{j}{1}{0}}$, $\gen{\reXM{j}{0}} \not\le \gen{\reXY{j}{0}{1},\reXY{j}{1}{0}}$, $\gen{\reXM{j}{0}} \not\le \gen{\reXY{j}{1}{1}}$, $\gen{\reXM{j}{0}} \not\le \gen{\reXY{j}{0}{0},\reXY{j}{1}{1}}$, $\gen{\reXM{j}{0}} \not\le \gen{\reXM{j}{1}}$, $\gen{\reXM{j}{0}} \not\le \gen{\reMY{j}{0}}$, $\gen{\reXM{j}{0}} \not\le \gen{\reMY{j}{1}}$ hold by \Cref{obs:arrow-negation-transitivity-2}.

		\item $\gen{\reXM{j}{1}} \not\le \gen{\reMM{j}}$, $\gen{\reXM{j}{1}} \not\le \gen{\reXY{j}{0}{0}}$, $\gen{\reXM{j}{1}} \not\le \gen{\reXY{j}{0}{1}}$, $\gen{\reXM{j}{1}} \not\le \gen{\reXY{j}{1}{0}}$, $\gen{\reXM{j}{1}} \not\le \gen{\reXY{j}{0}{1},\reXY{j}{1}{0}}$, $\gen{\reXM{j}{1}} \not\le \gen{\reXY{j}{1}{1}}$, $\gen{\reXM{j}{1}} \not\le \gen{\reXY{j}{0}{0},\reXY{j}{1}{1}}$, $\gen{\reXM{j}{1}} \not\le \gen{\reXM{j}{0}}$, $\gen{\reXM{j}{1}} \not\le \gen{\reMY{j}{0}}$, $\gen{\reXM{j}{1}} \not\le \gen{\reXM{j}{0},\reMY{j}{0}}$, $\gen{\reXM{j}{1}} \not\le \gen{\reMY{j}{1}}$, $\gen{\reXM{j}{1}} \not\le \gen{\reXM{j}{0},\reMY{j}{1}}$  for symmetric reasons.
		
		\item $\gen{\reMY{j}{0}} \not\le \gen{\reXM{j}{0},\reMY{j}{1}}$ because the configuration $L_1 \ldots L_\Delta$ is allowed, where:
		\begin{itemize}[noitemsep]
			\item $L_1 = \gen{\reMY{1}{1}}$;
			\item For $2 \le k < j$, $L_k = \gen{\reXY{k}{1}{1}}$;
			\item $L_j = \gen{\reMY{j}{0}}$
			\item For $j < k \le \Delta - i - 1$, $L_k = \gen{\reXY{k}{0}{0}}$;
			\item $L_{\Delta-i} = \gen{\reXY{\Delta-i}{0}{0},\reXY{\Delta-i}{0}{1}}$;
			\item For $\Delta-i+1 \le k \le \Delta$, $L_k = \gen{\reQ{k}}$;
		\end{itemize}
		but by replacing $\gen{\reMY{j}{0}}$ with $\gen{\reXM{j}{0},\reMY{j}{1}}$ we obtain a configuration that is not allowed.
		
		\item $\gen{\reMY{j}{0}} \not\le \gen{\reXM{j}{1},\reMY{j}{1}}$ because the configuration $L_1 \ldots L_\Delta$ is allowed, where:
		\begin{itemize}[noitemsep]
			\item $L_1 = \gen{\reMY{1}{0}}$;
			\item For $2 \le k < j$, $L_k = \gen{\reXY{k}{0}{0}}$;
			\item $L_j = \gen{\reMY{j}{0}}$
			\item For $j < k \le \Delta - i - 1$, $L_k = \gen{\reXY{k}{0}{0}}$;
			\item $L_{\Delta-i} = \gen{\reXY{\Delta-i}{0}{0},\reXY{\Delta-i}{0}{1}}$;
			\item For $\Delta-i+1 \le k \le \Delta$, $L_k = \gen{\reQ{k}}$;
		\end{itemize}
		but by replacing $\gen{\reMY{j}{0}}$ with $\gen{\reXM{j}{1},\reMY{j}{1}}$ we obtain a configuration that is not allowed.
		
		\item $\gen{\reMY{j}{0}} \not\le \gen{\reMM{j}}$, $\gen{\reMY{j}{0}} \not\le \gen{\reXY{j}{0}{0}}$, $\gen{\reMY{j}{0}} \not\le \gen{\reXY{j}{0}{1}}$, $\gen{\reMY{j}{0}} \not\le \gen{\reXY{j}{1}{0}}$, $\gen{\reMY{j}{0}} \not\le \gen{\reXY{j}{0}{1},\reXY{j}{1}{0}}$, $\gen{\reMY{j}{0}} \not\le \gen{\reXY{j}{1}{1}}$, $\gen{\reMY{j}{0}} \not\le \gen{\reXY{j}{0}{0},\reXY{j}{1}{1}}$, $\gen{\reMY{j}{0}} \not\le \gen{\reXM{j}{0}}$, $\gen{\reMY{j}{0}} \not\le \gen{\reXM{j}{1}}$, $\gen{\reMY{j}{0}} \not\le \gen{\reMY{j}{1}}$  hold by \Cref{obs:arrow-negation-transitivity-2}.

		\item $\gen{\reMY{j}{1}} \not\le \gen{\reMM{j}}$, $\gen{\reMY{j}{1}} \not\le \gen{\reXY{j}{0}{0}}$, $\gen{\reMY{j}{1}} \not\le \gen{\reXY{j}{0}{1}}$, $\gen{\reMY{j}{1}} \not\le \gen{\reXY{j}{1}{0}}$, $\gen{\reMY{j}{1}} \not\le \gen{\reXY{j}{0}{1},\reXY{j}{1}{0}}$, $\gen{\reMY{j}{1}} \not\le \gen{\reXY{j}{1}{1}}$, $\gen{\reMY{j}{1}} \not\le \gen{\reXY{j}{0}{0},\reXY{j}{1}{1}}$, $\gen{\reMY{j}{1}} \not\le \gen{\reXM{j}{0}}$, $\gen{\reMY{j}{1}} \not\le \gen{\reXM{j}{1}}$, $\gen{\reMY{j}{1}} \not\le \gen{\reMY{j}{0}}$, $\gen{\reMY{j}{1}} \not\le \gen{\reXM{j}{0},\reMY{j}{0}}$, $\gen{\reMY{j}{1}} \not\le \gen{\reXM{j}{1},\reMY{j}{0}}$ for symmetric reasons.

		\item $\gen{\reXY{j}{0}{1},\reXY{j}{1}{0}} \not\le \gen{\reXM{j}{1},\reMY{j}{0}}$ because the configuration $L_1 \ldots L_\Delta$ is allowed, where:
		\begin{itemize}[noitemsep]
			\item $L_1 = \gen{\reMY{1}{0}}$;
			\item For $2 \le k < j$, $L_k = \gen{\reXY{k}{0}{0}}$;
			\item $L_j = \gen{\reXY{j}{0}{1},\reXY{j}{1}{0}}$;
			\item For $j < k \le \Delta - i - 1$, $L_k = \gen{\reXY{k}{1}{1}}$;
			\item $L_{\Delta-i} = \gen{\reXY{\Delta-i}{1}{0},\reXY{\Delta-i}{1}{1}}$;
			\item For $\Delta-i+1 \le k \le \Delta$, $L_k = \gen{\reQ{k}}$;
		\end{itemize}
		but by replacing $\gen{\reXY{j}{0}{1},\reXY{j}{1}{0}}$ with $\gen{\reXM{j}{1},\reMY{j}{0}}$ we obtain a configuration that is not allowed.
		
		\item $\gen{\reXY{j}{0}{1},\reXY{j}{1}{0}} \not\le \gen{\reXM{j}{0},\reMY{j}{1}}$  because the configuration $L_1 \ldots L_\Delta$ is allowed, where:
		\begin{itemize}[noitemsep]
			\item $L_1 = \gen{\reMY{1}{1}}$;
			\item For $2 \le k < j$, $L_k = \gen{\reXY{k}{1}{1}}$;
			\item $L_j = \gen{\reXY{j}{0}{1},\reXY{j}{1}{0}}$;
			\item For $j < k \le \Delta - i - 1$, $L_k = \gen{\reXY{k}{0}{0}}$;
			\item $L_{\Delta-i} = \gen{\reXY{\Delta-i}{0}{0},\reXY{\Delta-i}{0}{1}}$;
			\item For $\Delta-i+1 \le k \le \Delta$, $L_k = \gen{\reQ{k}}$;
		\end{itemize}
		but by replacing $\gen{\reXY{j}{0}{1},\reXY{j}{1}{0}}$ with $ \gen{\reXM{j}{0},\reMY{j}{1}}$ we obtain a configuration that is not allowed.
		
		\item $\gen{\reXY{j}{0}{1},\reXY{j}{1}{0}} \not\le \gen{\reMM{j}}$, $\gen{\reXY{j}{0}{1},\reXY{j}{1}{0}} \not\le \gen{\reXY{j}{0}{0}}$, $\gen{\reXY{j}{0}{1},\reXY{j}{1}{0}} \not\le \gen{\reXY{j}{0}{1}}$, $\gen{\reXY{j}{0}{1},\reXY{j}{1}{0}} \not\le \gen{\reXY{j}{1}{0}}$, $\gen{\reXY{j}{0}{1},\reXY{j}{1}{0}} \not\le \gen{\reXY{j}{1}{1}}$, $\gen{\reXY{j}{0}{1},\reXY{j}{1}{0}} \not\le \gen{\reXY{j}{0}{0},\reXY{j}{1}{1}}$, $\gen{\reXY{j}{0}{1},\reXY{j}{1}{0}} \not\le \gen{\reXM{j}{0}}$, $\gen{\reXY{j}{0}{1},\reXY{j}{1}{0}} \not\le \gen{\reXM{j}{1}}$, $\gen{\reXY{j}{0}{1},\reXY{j}{1}{0}} \not\le \gen{\reMY{j}{0}}$, $\gen{\reXY{j}{0}{1},\reXY{j}{1}{0}} \not\le \gen{\reMY{j}{1}}$ hold by \Cref{obs:arrow-negation-transitivity-2}.

		\item $\gen{\reXY{j}{0}{0},\reXY{j}{1}{1}} \not\le \gen{\reXM{j}{1},\reMY{j}{1}}$ because the configuration $L_1 \ldots L_\Delta$ is allowed, where:
		\begin{itemize}[noitemsep]
			\item $L_1 = \gen{\reMY{1}{0}}$;
			\item For $2 \le k < j$, $L_k = \gen{\reXY{k}{0}{0}}$;
			\item $L_j = \gen{\reXY{j}{0}{0},\reXY{j}{1}{1}}$;
			\item For $j < k \le \Delta - i - 1$, $L_k = \gen{\reXY{k}{0}{0}}$;
			\item $L_{\Delta-i} = \gen{\reXY{\Delta-i}{0}{0},\reXY{\Delta-i}{0}{1}}$;
			\item For $\Delta-i+1 \le k \le \Delta$, $L_k = \gen{\reQ{k}}$;
		\end{itemize}
		but by replacing $\gen{\reXY{j}{0}{0},\reXY{j}{1}{1}}$ with $\gen{\reXM{j}{1},\reMY{j}{1}}$ we obtain a configuration that is not allowed.
		
		\item $\gen{\reXY{j}{0}{0},\reXY{j}{1}{1}} \not\le \gen{\reXM{j}{0},\reMY{j}{0}}$ because the configuration $L_1 \ldots L_\Delta$ is allowed, where:
		\begin{itemize}[noitemsep]
			\item $L_1 = \gen{\reMY{1}{1}}$;
			\item For $2 \le k < j$, $L_k = \gen{\reXY{k}{1}{1}}$;
			\item $L_j = \gen{\reXY{j}{0}{0},\reXY{j}{1}{1}}$;
			\item For $j < k \le \Delta - i - 1$, $L_k = \gen{\reXY{k}{1}{1}}$;
			\item $L_{\Delta-i} = \gen{\reXY{\Delta-i}{1}{0},\reXY{\Delta-i}{1}{1}}$;
			\item For $\Delta-i+1 \le k \le \Delta$, $L_k = \gen{\reQ{k}}$;
		\end{itemize}
		but by replacing $\gen{\reXY{j}{0}{0},\reXY{j}{1}{1}}$ with $\gen{\reXM{j}{0},\reMY{j}{0}}$ we obtain a configuration that is not allowed.
		
		\item $\gen{\reXY{j}{0}{0},\reXY{j}{1}{1}} \not\le \gen{\reMM{j}}$, $\gen{\reXY{j}{0}{0},\reXY{j}{1}{1}} \not\le \gen{\reXY{j}{0}{0}}$, $\gen{\reXY{j}{0}{0},\reXY{j}{1}{1}} \not\le \gen{\reXY{j}{0}{1}}$, $\gen{\reXY{j}{0}{0},\reXY{j}{1}{1}} \not\le \gen{\reXY{j}{1}{0}}$, $\gen{\reXY{j}{0}{0},\reXY{j}{1}{1}} \not\le \gen{\reXY{j}{0}{1},\reXY{j}{1}{0}}$, $\gen{\reXY{j}{0}{0},\reXY{j}{1}{1}} \not\le \gen{\reXY{j}{1}{1}}$, $\gen{\reXY{j}{0}{0},\reXY{j}{1}{1}} \not\le \gen{\reXM{j}{0}}$, $\gen{\reXY{j}{0}{0},\reXY{j}{1}{1}} \not\le \gen{\reXM{j}{1}}$, $\gen{\reXY{j}{0}{0},\reXY{j}{1}{1}} \not\le \gen{\reMY{j}{0}}$, $\gen{\reXY{j}{0}{0},\reXY{j}{1}{1}} \not\le \gen{\reMY{j}{1}}$ hold by \Cref{obs:arrow-negation-transitivity-2}.

		\item $\gen{\reXM{j}{0},\reMY{j}{0}} \not\le \gen{\reXM{j}{1},\reMY{j}{0}}$  because the configuration $L_1 \ldots L_\Delta$ is allowed, where:
		\begin{itemize}[noitemsep]
			\item $L_1 = \gen{\reMY{1}{0}}$;
			\item For $2 \le k < j$, $L_k = \gen{\reXY{k}{0}{0}}$;
			\item $L_j = \gen{\reXM{j}{0},\reMY{j}{0}}$;
			\item For $j < k \le \Delta - i - 1$, $L_k = \gen{\reXY{k}{1}{1}}$;
			\item $L_{\Delta-i} = \gen{\reXY{\Delta-i}{1}{0},\reXY{\Delta-i}{1}{1}}$;
			\item For $\Delta-i+1 \le k \le \Delta$, $L_k = \gen{\reQ{k}}$;
		\end{itemize}
		but by replacing $\gen{\reXM{j}{0},\reMY{j}{0}} $ with $\gen{\reXM{j}{1},\reMY{j}{0}}$ we obtain a configuration that is not allowed.
		
		\item $\gen{\reXM{j}{0},\reMY{j}{0}} \not\le \gen{\reXM{j}{0},\reMY{j}{1}}$ because the configuration $L_1 \ldots L_\Delta$ is allowed, where:
		\begin{itemize}[noitemsep]
			\item $L_1 = \gen{\reMY{1}{1}}$;
			\item For $2 \le k < j$, $L_k = \gen{\reXY{k}{1}{1}}$;
			\item $L_j = \gen{\reXM{j}{0},\reMY{j}{0}}$;
			\item For $j < k \le \Delta - i - 1$, $L_k = \gen{\reXY{k}{0}{0}}$;
			\item $L_{\Delta-i} = \gen{\reXY{\Delta-i}{0}{0},\reXY{\Delta-i}{0}{1}}$;
			\item For $\Delta-i+1 \le k \le \Delta$, $L_k = \gen{\reQ{k}}$;
		\end{itemize}
		but by replacing $\gen{\reXM{j}{0},\reMY{j}{0}} $ with $\gen{\reXM{j}{0},\reMY{j}{1}}$ we obtain a configuration that is not allowed.
		
		\item $\gen{\reXM{j}{0},\reMY{j}{0}} \not\le \gen{\reXM{j}{1},\reMY{j}{1}}$  because the configuration $L_1 \ldots L_\Delta$ is allowed, where:
		\begin{itemize}[noitemsep]
			\item $L_1 = \gen{\reMY{1}{0}}$;
			\item For $2 \le k < j$, $L_k = \gen{\reXY{k}{0}{0}}$;
			\item $L_j = \gen{\reXM{j}{0},\reMY{j}{0}}$;
			\item For $j < k \le \Delta - i - 1$, $L_k = \gen{\reXY{k}{0}{0}}$;
			\item $L_{\Delta-i} = \gen{\reXY{\Delta-i}{0}{0},\reXY{\Delta-i}{0}{1}}$;
			\item For $\Delta-i+1 \le k \le \Delta$, $L_k = \gen{\reQ{k}}$;
		\end{itemize}
		but by replacing $\gen{\reXM{j}{0},\reMY{j}{0}} $ with $\gen{\reXM{j}{1},\reMY{j}{1}}$ we obtain a configuration that is not allowed.
		
		\item $\gen{\reXM{j}{0},\reMY{j}{0}} \not\le \gen{\reMM{j}}$, $\gen{\reXM{j}{0},\reMY{j}{0}} \not\le \gen{\reXY{j}{0}{0}}$, $\gen{\reXM{j}{0},\reMY{j}{0}} \not\le \gen{\reXY{j}{0}{1}}$, $\gen{\reXM{j}{0},\reMY{j}{0}} \not\le \gen{\reXY{j}{1}{0}}$, $\gen{\reXM{j}{0},\reMY{j}{0}} \not\le \gen{\reXY{j}{0}{1},\reXY{j}{1}{0}}$, $\gen{\reXM{j}{0},\reMY{j}{0}} \not\le \gen{\reXY{j}{1}{1}}$, $\gen{\reXM{j}{0},\reMY{j}{0}} \not\le \gen{\reXY{j}{0}{0},\reXY{j}{1}{1}}$, $\gen{\reXM{j}{0},\reMY{j}{0}} \not\le \gen{\reXM{j}{0}}$, $\gen{\reXM{j}{0},\reMY{j}{0}} \not\le \gen{\reXM{j}{1}}$, $\gen{\reXM{j}{0},\reMY{j}{0}} \not\le \gen{\reMY{j}{0}}$, $\gen{\reXM{j}{0},\reMY{j}{0}} \not\le \gen{\reMY{j}{1}}$ hold by \Cref{obs:arrow-negation-transitivity-2}.

		\item $\gen{\reXM{j}{0},\reMY{j}{1}} \not\le \gen{\reXM{j}{1},\reMY{j}{1}}$  because the configuration $L_1 \ldots L_\Delta$ is allowed, where:
		\begin{itemize}[noitemsep]
			\item $L_1 = \gen{\reMY{1}{0}}$;
			\item For $2 \le k < j$, $L_k = \gen{\reXY{k}{0}{0}}$;
			\item $L_j = \gen{\reXM{j}{0},\reMY{j}{1}}$;
			\item For $j < k \le \Delta - i - 1$, $L_k = \gen{\reXY{k}{0}{0}}$;
			\item $L_{\Delta-i} = \gen{\reXY{\Delta-i}{0}{0},\reXY{\Delta-i}{0}{1}}$;
			\item For $\Delta-i+1 \le k \le \Delta$, $L_k = \gen{\reQ{k}}$;
		\end{itemize}
		but by replacing $\gen{\reXM{j}{0},\reMY{j}{1}}$ with $\gen{\reXM{j}{1},\reMY{j}{1}}$ we obtain a configuration that is not allowed.

		\item $\gen{\reXM{j}{0},\reMY{j}{1}} \not\le \gen{\reXM{j}{1},\reMY{j}{0}}$ because the configuration $L_1 \ldots L_\Delta$ is allowed, where:
		\begin{itemize}[noitemsep]
			\item $L_1 = \gen{\reMY{1}{0}}$;
			\item For $2 \le k < j$, $L_k = \gen{\reXY{k}{0}{0}}$;
			\item $L_j = \gen{\reXM{j}{0},\reMY{j}{1}}$;
			\item For $j < k \le \Delta - i - 1$, $L_k = \gen{\reXY{k}{1}{1}}$;
			\item $L_{\Delta-i} = \gen{\reXY{\Delta-i}{1}{0},\reXY{\Delta-i}{1}{1}}$;
			\item For $\Delta-i+1 \le k \le \Delta$, $L_k = \gen{\reQ{k}}$;
		\end{itemize}
		but by replacing $\gen{\reXM{j}{0},\reMY{j}{1}}$ with $\gen{\reXM{j}{1},\reMY{j}{0}}$ we obtain a configuration that is not allowed.
		
		\item $\gen{\reXM{j}{0},\reMY{j}{1}} \not\le \gen{\reXM{j}{0},\reMY{j}{0}}$ because the configuration $L_1 \ldots L_\Delta$ is allowed, where:
		\begin{itemize}[noitemsep]
			\item $L_1 = \gen{\reMY{1}{1}}$;
			\item For $2 \le k < j$, $L_k = \gen{\reXY{k}{1}{1}}$;
			\item $L_j = \gen{\reXM{j}{0},\reMY{j}{1}}$;
			\item For $j < k \le \Delta - i - 1$, $L_k = \gen{\reXY{k}{1}{1}}$;
			\item $L_{\Delta-i} = \gen{\reXY{\Delta-i}{1}{0},\reXY{\Delta-i}{1}{1}}$;
			\item For $\Delta-i+1 \le k \le \Delta$, $L_k = \gen{\reQ{k}}$;
		\end{itemize}
		but by replacing $\gen{\reXM{j}{0},\reMY{j}{1}}$ with $\gen{\reXM{j}{0},\reMY{j}{0}}$ we obtain a configuration that is not allowed.
		
		\item $\gen{\reXM{j}{0},\reMY{j}{1}} \not\le \gen{\reMM{j}}$, $\gen{\reXM{j}{0},\reMY{j}{1}} \not\le \gen{\reXY{j}{0}{0}}$, $\gen{\reXM{j}{0},\reMY{j}{1}} \not\le \gen{\reXY{j}{0}{1}}$, $\gen{\reXM{j}{0},\reMY{j}{1}} \not\le \gen{\reXY{j}{1}{0}}$, $\gen{\reXM{j}{0},\reMY{j}{1}} \not\le \gen{\reXY{j}{0}{1},\reXY{j}{1}{0}}$, $\gen{\reXM{j}{0},\reMY{j}{1}} \not\le \gen{\reXY{j}{1}{1}}$, $\gen{\reXM{j}{0},\reMY{j}{1}} \not\le \gen{\reXY{j}{0}{0},\reXY{j}{1}{1}}$, $\gen{\reXM{j}{0},\reMY{j}{1}} \not\le \gen{\reXM{j}{0}}$, $\gen{\reXM{j}{0},\reMY{j}{1}} \not\le \gen{\reXM{j}{1}}$, $\gen{\reXM{j}{0},\reMY{j}{1}} \not\le \gen{\reMY{j}{0}}$, $\gen{\reXM{j}{0},\reMY{j}{1}} \not\le \gen{\reMY{j}{1}}$ hold by \Cref{obs:arrow-negation-transitivity-2}.

		\item $\gen{\reXM{j}{1},\reMY{j}{0}} \not\le \gen{\reMM{j}}$, $\gen{\reXM{j}{1},\reMY{j}{0}} \not\le \gen{\reXY{j}{0}{0}}$, $\gen{\reXM{j}{1},\reMY{j}{0}} \not\le \gen{\reXY{j}{0}{1}}$, $\gen{\reXM{j}{1},\reMY{j}{0}} \not\le \gen{\reXY{j}{1}{0}}$, $\gen{\reXM{j}{1},\reMY{j}{0}} \not\le \gen{\reXY{j}{0}{1},\reXY{j}{1}{0}}$,  $\gen{\reXM{j}{1},\reMY{j}{0}} \not\le \gen{\reXY{j}{1}{1}}$, $\gen{\reXM{j}{1},\reMY{j}{0}} \not\le \gen{\reXY{j}{0}{0},\reXY{j}{1}{1}}$, $\gen{\reXM{j}{1},\reMY{j}{0}} \not\le \gen{\reXM{j}{0}}$, $\gen{\reXM{j}{1},\reMY{j}{0}} \not\le \gen{\reXM{j}{1}}$, $\gen{\reXM{j}{1},\reMY{j}{0}} \not\le \gen{\reMY{j}{0}}$, $\gen{\reXM{j}{1},\reMY{j}{0}} \not\le \gen{\reXM{j}{0},\reMY{j}{0}}$, $\gen{\reXM{j}{1},\reMY{j}{0}} \not\le \gen{\reMY{j}{1}}$, $\gen{\reXM{j}{1},\reMY{j}{0}} \not\le \gen{\reXM{j}{0},\reMY{j}{1}}$, $\gen{\reXM{j}{1},\reMY{j}{0}} \not\le \gen{\reXM{j}{1},\reMY{j}{1}}$ for symmetric reasons.
					
		\item $\gen{\reXM{j}{1},\reMY{j}{1}} \not\le \gen{\reMM{j}}$, $\gen{\reXM{j}{1},\reMY{j}{1}} \not\le \gen{\reXY{j}{0}{0}}$, $\gen{\reXM{j}{1},\reMY{j}{1}} \not\le \gen{\reXY{j}{0}{1}}$, $\gen{\reXM{j}{1},\reMY{j}{1}} \not\le \gen{\reXY{j}{1}{0}}$, $\gen{\reXM{j}{1},\reMY{j}{1}} \not\le \gen{\reXY{j}{0}{1},\reXY{j}{1}{0}}$, $\gen{\reXM{j}{1},\reMY{j}{1}} \not\le \gen{\reXY{j}{1}{1}}$, $\gen{\reXM{j}{1},\reMY{j}{1}} \not\le \gen{\reXY{j}{0}{0},\reXY{j}{1}{1}}$, $\gen{\reXM{j}{1},\reMY{j}{1}} \not\le \gen{\reXM{j}{0}}$, $\gen{\reXM{j}{1},\reMY{j}{1}} \not\le \gen{\reXM{j}{1}}$, $\gen{\reXM{j}{1},\reMY{j}{1}} \not\le \gen{\reMY{j}{0}}$, $\gen{\reXM{j}{1},\reMY{j}{1}} \not\le \gen{\reXM{j}{0},\reMY{j}{0}}$, $\gen{\reXM{j}{1},\reMY{j}{1}} \not\le \gen{\reXM{j}{1},\reMY{j}{0}}$, $\gen{\reXM{j}{1},\reMY{j}{1}} \not\le \gen{\reMY{j}{1}}$, $\gen{\reXM{j}{1},\reMY{j}{1}} \not\le \gen{\reXM{j}{0},\reMY{j}{1}}$ for symmetric reasons.

		\item $\gen{\reEE{j}} \not\le \gen{\reXM{j}{0},\reMY{j}{1}}$ because the configuration $L_1 \ldots L_\Delta$ is allowed, where:
		\begin{itemize}[noitemsep]
			\item $L_1 = \gen{\reMY{1}{1}}$;
			\item For $2 \le k < j$, $L_k = \gen{\reXY{k}{1}{1}}$;
			\item $L_j = \gen{\reEE{j}}$;
			\item For $j < k \le \Delta - i - 1$, $L_k = \gen{\reXY{k}{0}{0}}$;
			\item $L_{\Delta-i} = \gen{\reXY{\Delta-i}{0}{0},\reXY{\Delta-i}{0}{1}}$;
			\item For $\Delta-i+1 \le k \le \Delta$, $L_k = \gen{\reQ{k}}$;
		\end{itemize}
		but by replacing $\gen{\reEE{j}}$ with $\gen{\reXM{j}{0},\reMY{j}{1}}$ we obtain a configuration that is not allowed.
		
		\item $\gen{\reEE{j}} \not\le \gen{\reXM{j}{0},\reMY{j}{0}}$ because the configuration $L_1 \ldots L_\Delta$ is allowed, where:
		\begin{itemize}[noitemsep]
			\item $L_1 = \gen{\reMY{1}{1}}$;
			\item For $2 \le k < j$, $L_k = \gen{\reXY{k}{1}{1}}$;
			\item $L_j = \gen{\reEE{j}}$;
			\item For $j < k \le \Delta - i - 1$, $L_k = \gen{\reXY{k}{1}{1}}$;
			\item $L_{\Delta-i} = \gen{\reXY{\Delta-i}{1}{0},\reXY{\Delta-i}{1}{1}}$;
			\item For $\Delta-i+1 \le k \le \Delta$, $L_k = \gen{\reQ{k}}$;
		\end{itemize}
		but by replacing $\gen{\reEE{j}}$ with $\gen{\reXM{j}{0},\reMY{j}{0}}$ we obtain a configuration that is not allowed.

		\item $\gen{\reEE{j}} \not\le \gen{\reXM{j}{1},\reMY{j}{1}}$, $\gen{\reEE{j}} \not\le \gen{\reXM{j}{1},\reMY{j}{0}}$ for symmetric reasons.
		
		\item $\gen{\reEE{j}} \not\le \gen{\reMM{j}}$, $\gen{\reEE{j}} \not\le \gen{\reXY{j}{0}{0}}$, $\gen{\reEE{j}} \not\le \gen{\reXY{j}{0}{1}}$, $\gen{\reEE{j}} \not\le \gen{\reXY{j}{1}{0}}$, $\gen{\reEE{j}} \not\le \gen{\reXY{j}{0}{1},\reXY{j}{1}{0}}$, $\gen{\reEE{j}} \not\le \gen{\reXY{j}{1}{1}}$, $\gen{\reEE{j}} \not\le \gen{\reXY{j}{0}{0},\reXY{j}{1}{1}}$, $\gen{\reEE{j}} \not\le \gen{\reXM{j}{0}}$, $\gen{\reEE{j}} \not\le \gen{\reXM{j}{1}}$, $\gen{\reEE{j}} \not\le \gen{\reMY{j}{0}}$, $\gen{\reEE{j}} \not\le \gen{\reMY{j}{1}}$ hold by \Cref{obs:arrow-negation-transitivity-2}. \qedhere
	\end{itemize}
\end{proof}

\begin{proof}[Proof of \Cref{lem:diagram-intermediate-special}]
	By \Cref{obs:no-arrow-different-colors-2}, labels of the special color $j = \Delta - i$ can only be in relation with labels of color $j$.
	By \Cref{obs:arrow-is-present-2}, the relations depicted in the diagram are present. Hence, we only need to prove that no additional relations than the ones of \Cref{fig:diagram-intermediate-special} hold. We consider all possible cases:
	\begin{itemize}
		\item $\gen{\reXY{j}{0}{0},\reXY{j}{0}{1}} \not\le \gen{\reXY{j}{1}{0},\reXY{j}{1}{1}}$ because the configuration $L_1 \ldots L_\Delta$ is allowed, where:
		\begin{itemize}[noitemsep]
			\item $L_1 = \gen{\reMY{1}{0}}$;
			\item For $2 \le k \le \Delta - i - 1$, $L_k = \gen{\reXY{k}{0}{0}}$;
			\item $L_{j} = \gen{\reXY{j}{0}{0},\reXY{j}{0}{1}}$;
			\item For $\Delta-i+1 \le k \le \Delta$, $L_k = \gen{\reQ{k}}$;
		\end{itemize}
		but by replacing $\gen{\reXY{j}{0}{0},\reXY{j}{0}{1}}$ with $\gen{\reXY{j}{1}{0},\reXY{j}{1}{1}}$ we obtain a configuration that is not allowed.
		
		\item $\gen{\reXY{j}{1}{0},\reXY{j}{1}{1}} \not\le \gen{\reXY{j}{0}{0},\reXY{j}{0}{1}}$ for symmetric reasons as in the case  $\gen{\reXY{j}{0}{0},\reXY{j}{0}{1}} \not\le \gen{\reXY{j}{1}{0},\reXY{j}{1}{1}}$.
		
		\item $\gen{\reMY{j}{0},\reMY{j}{1}} \not\le \gen{\reXY{j}{0}{0},\reXY{j}{0}{1}}$ because the configuration $L_1 \ldots L_\Delta$ is allowed, where:
		\begin{itemize}[noitemsep]
			\item $L_1 = \gen{\reMY{1}{1}}$;
			\item For $2 \le k \le \Delta - i - 1$, $L_k = \gen{\reXY{k}{1}{1}}$;
			\item $L_{j} = \gen{\reMY{j}{0},\reMY{j}{1}}$;
			\item For $\Delta-i+1 \le k \le \Delta$, $L_k = \gen{\reQ{k}}$;
		\end{itemize}
		but by replacing $\gen{\reMY{j}{0},\reMY{j}{1}} $ with $\gen{\reXY{j}{0}{0},\reXY{j}{0}{1}}$ we obtain a configuration that is not allowed.
		
		\item $\gen{\reMY{j}{0},\reMY{j}{1}} \not\le \gen{\reXY{j}{1}{0},\reXY{j}{1}{1}}$ for symmetric reasons as in the case $\gen{\reMY{j}{0},\reMY{j}{1}} \not\le \gen{\reXY{j}{0}{0},\reXY{j}{0}{1}}$. \qedhere
	\end{itemize}
\end{proof} 
\end{document}